\documentclass[12pt, reqno]{amsart}

\usepackage{eurosym}
\usepackage{amsmath}
\usepackage{pdfsync}
\usepackage{graphicx,amssymb,epsfig, float}
\usepackage{amsfonts, lscape}
\usepackage{color}
\usepackage{bbm}
\usepackage{booktabs}
\usepackage[authoryear]{natbib}
\usepackage{framed}
\usepackage{enumitem}
\usepackage{titletoc}

\usepackage{xcolor}

\usepackage{hyperref}

\usepackage{float}

\definecolor{darkblue}{rgb}{0, 0, 0.5}

\hypersetup{
colorlinks=true,
urlcolor=darkblue,
linkcolor=darkblue,
citecolor=darkblue
}

\usepackage[margin=1in]{geometry}

\newtheorem{theorem}{Theorem}
\theoremstyle{definition}

\theoremstyle{plain}

\newtheorem{assumption}{Assumption}

\newtheorem{corollary}{Corollary}

\newtheorem{lemma}{Lemma}

\newtheorem{proposition}{Proposition}

\numberwithin{equation}{section}

\theoremstyle{definition}
\newtheorem{remark}{Remark}[section]
\graphicspath{ {FigsEtc/} }

\newcommand{\G}{H}
\newcommand{\EL}{K}
\newcommand{\K}{K}

\newcommand{\Z}{\Psi}
\newcommand{\E}{\mathbb E}

\newcommand{\br}{L}

\newcommand{\est}{\hat\theta_N}
\newcommand{\true}{\theta^*}
\newcommand{\fp}{\bar\theta}
\newcommand{\param}{\theta}

\newcommand{\dparam}{{d_X+d_D}}
\newcommand{\dpl}{J}

\renewcommand{\equiv}{:=}

\begin{document}
\title[\textbf{Decentralization Estimators for IVQR}]{\textbf{Decentralization Estimators for Instrumental Variable Quantile Regression Models}\\}
\author[Kaido and W\"uthrich]{Hiroaki Kaido and Kaspar W\"uthrich \vspace{1cm} }

\date{\today. Kaido: Economics Department, Boston University, Boston, MA, email: \texttt{\hyperlink{hkaido@bu.edu}{hkaido@bu.edu}}. W\"uthrich: Department of Economics, University of California San Diego, La Jolla, CA, email: \texttt{\hyperlink{kwuthrich@ucsd.edu}{kwuthrich@ucsd.edu}}.  We are grateful to Alfred Galichon for many helpful discussions and valuable inputs. We would like to thank Victor Chernozhukov, Patrik Guggenberger, Hidehiko Ichimura, Yuichiro Kamada, David Kaplan, Hiroyuki Kasahara, Kengo Kato, Ivana Komunjer, Blaise Melly, Klaus Neusser, Quang Vuong (discussant), the editor (Christopher Taber), anonymous referees, and conference and seminar participants at Boston College, Columbia University, Penn State, the University of Arizona, the University of Bern, UCLA, UCR, UCSB, UCSD, UNC, and the Workshop for Advances in Microeconometrics for helpful comments. We are grateful to Chris Hansen for sharing the data for the empirical application. This project was started when both authors were visiting the MIT in 2015. }

\begin{abstract}
The instrumental variable quantile regression (IVQR) model \citep{CH2005} is a popular tool for estimating causal quantile effects with endogenous covariates. However, estimation is complicated by the non-smoothness and non-convexity of the IVQR GMM objective function. This paper shows that the IVQR estimation problem can be decomposed into a set of conventional quantile regression sub-problems which are convex and can be solved efficiently. This reformulation leads to new identification results and  to fast, easy to implement, and tuning-free estimators that do not require the availability of high-level ``black box'' optimization routines. 
 \vspace{0.2in}

\end{abstract}

\bigskip

\maketitle

\clearpage
\section{Introduction}

Quantile regression (QR), introduced by \citet{KoenkerBassett1978}, is a widely-used method for estimating the effect of regressors on the whole outcome distribution. QR is flexible, easy to interpret, and can be computed very efficiently as the solution to a convex problem. However, in many applications, the variables of interest are endogenous, rendering QR inconsistent for estimating causal quantile effects. The instrumental variable quantile regression (IVQR) model of \citet{CH2005,CH2006} generalizes QR to accommodate endogenous regressors. Unfortunately, in sharp contrast to QR and other IV estimators such as two-stage least squares (2SLS), estimation of IVQR models is computationally challenging because the resulting estimation problem, formulated as a generalized method of moments (GMM) problem, is non-smooth and non-convex. From an applied perspective, this issue is particularly troublesome as resampling methods are often used to avoid the choice of tuning parameters when estimating the asymptotic variance of IVQR estimators.

In this paper, we develop a new class of estimators for linear IVQR models. The proposed estimators are fast, easy to implement, tuning-free, and do not require the availability of high-level ``black box'' optimization routines. They are particularly suitable for settings with many exogenous regressors, a moderate number of endogenous regressors, and a large number of observations, which are ubiquitous in applied research. The key insight underlying our estimators is that the complicated and nonlinear IVQR estimation problem can be ``decentralized'', i.e., decomposed into a set of more tractable sub-problems, each of which is  solved by a ``player'' who best responds to  the other players' actions. Each subproblem is a conventional (weighted) QR problem, which is convex and can be solved very quickly using robust algorithms. The IVQR estimator is then characterized as a fixed point of such sub-problems, which can be viewed as the pure strategy Nash equilibrium of the ``game''. Computationally, this reformulation allows us to recast the original non-smooth and non-convex optimization problem as the problem of finding the fixed point of a low dimensional map, which leads to substantial reductions in computation times. 

Implementation of our preferred procedures is straightforward and only requires the availability of a routine for estimating quantile regressions and in some cases a univariate root-finder. The resulting estimation algorithms attain significant computational gains. For example, we show that in problems with two endogenous variables, a version of our estimator that uses a contraction algorithm is 149--308 times faster than the most popular existing approach for estimating IVQR models, the inverse quantile regression (IQR) estimator of \citet{CH2006}. Another version that uses a nested root-finding algorithm, which is guaranteed to converge under a milder condition, is 84--134 times faster than the IQR estimator. Importantly, these computational gains do not come at a cost in terms of the finite sample performance of our procedures, which is very similar to IQR. The computational advantages of our estimators are even more substantial with more than two endogenous variables. The reason is that the dimensionality of the grid search underlying IQR corresponds to the number of endogenous variables, which renders IQR computationally prohibitive in empirically relevant settings whenever the number of endogenous variables exceeds two or three. 

The fixed point reformulation also provides new insights into global identification of IVQR models. In particular, it allows us to study identification and stability of the algorithms (at the population level) in the same framework. Exploiting the equivalence of global identification and uniqueness of the fixed point, we give a new identification result and population algorithms based on the contraction mapping theorem. We then compare our identification conditions to those of \citet{CH2006}. Further, our reformulation is shown to be useful beyond setups where the contraction mapping theorem applies as long as the parameter of interest is globally identified. For such settings, algorithms based on root-finding methods are proposed. Finally, we show that, by recursively nesting fixed point problems, it is  possible to recast the IVQR estimation problem as a univariate root-finding problem. While adding nests incurs additional computational costs, our Monte Carlo experiments suggest that  estimation procedures based on nesting perform well and are computationally reasonable when the number of endogenous variables is moderate.

We establish consistency and asymptotic normality of the proposed estimators and prove validity of the empirical bootstrap for estimating the limiting laws. The bootstrap is particularly attractive in conjunction with our computationally efficient estimation algorithms as it allows us to avoid the choice of tuning parameters inherent to estimating the asymptotic variance based on analytic formulas. The key technical ingredient for deriving our theoretical results is the Hadamard differentiability of the fixed point map. This result may be of independent interest.

To illustrate the usefulness of our estimation algorithms, we revisit the analysis of the impact of 401(k) plans on savings in \citet{CH2004}. Based on this application, we perform extensive Monte Carlo simulations, which demonstrate that our estimation and inference procedures have excellent finite sample properties.

\subsection{Literature}

We contribute to the literature on estimation and inference based on linear IVQR models. \citet{ChernozhukovHong2003} have proposed a quasi-Bayesian approach which can accommodate multiple endogenous variables but, as noted by \citet{CH2013}, requires careful tuning in applications. \citet{CH2006} have developed an inverse QR algorithm that combines grid search with convex QR problems. Because the dimensionality of the grid search equals the number of endogenous variables, this approach is computationally feasible only if the number of endogenous variables is very low. \citet{ChernozhukovHansen2008} and \citet{Jun2008} have studied weak instrument robust inference procedures based on  the inversion of Anderson-Rubin-type tests. \citet{Chernozhukovetal2009} have proposed a finite sample inference approach. \citet{AndrewsMikusheva2016} have developed a general conditional inference approach and derived sufficient conditions for the IVQR model. \citet{KaplanSun2017} and \citet{deCastroetal2018} have suggested to use smoothed estimating equations to overcome the non-smoothness of the IVQR estimation problem, although the non-convexity remains.
More recently, \citet{ChenLee2018} have proposed to reformulate the IVQR problem as a mixed-integer quadratic programming problem that can be solved using well-established algorithms. However, efficiently solving such a problem is still challenging even for low-dimensional settings. By replacing the $\ell_2$ norm by the $\ell_\infty$ norm, \citet{Zhu2018} has shown that the problem admits a reformulation as a mixed-integer linear programming problem, which can be computed more efficiently than the quadratic program in \citet{ChenLee2018}. This procedure typically requires an early termination of the algorithm to ensure computational tractability which is akin to a tuning parameter choice. In addition, \citet{Zhu2018} has proposed a $k$-step approach that allows for estimating models with multiple endogenous regressors based on large datasets, but requires estimating the gradient. \citet{pouliot2018} proposes a mixed integer linear programming formulation that allows for subvector inference via the inversion of a distribution-free rankscore test and can be modified to accommodate weak instruments. An important drawback of the estimation approaches based on mixed integer reformulations is that they rely on the availability of high-level ``black box'' optimization routines such as Gurobi and often require careful tuning in applications. Finally, imposing a location-scale model for the potential outcomes, \citet{Machado2018} propose moment-based estimators for the structural quantile function.

Compared to the existing literature on the estimation of linear IVQR models, the main advantages of the proposed estimation algorithms are the following. First, by relying on convex QR problems, our estimators are easy to implement, robust, and computationally efficient in settings with many exogenous variables, a moderate number of endogenous variables, and a large number of observations. Second, by exploiting the specific structure of the IVQR estimation problem, our estimators are tuning-free and do not require the availability of high-level ``black box'' optimization routines. Third, our estimators are based on the original IVQR estimation problem and thus avoid the choice of  smoothing bandwidths and do not rely on additional restrictions on the structural quantile function.

Semi- and nonparametric estimation of IVQR models has been studied by \citet{ChernozhukovImbensNewey2007}, \citet{HorowitzLee2007}, \citet{ChenPouzo2009}, \citet{ChenPouzo2012}, \citet{GagliardiniScaillet2012}, and \citet{Wuthrich2019}. \citet{CH2013} and \citet{Chernozhukov+17handbook} have provided surveys on the IVQR model including references to empirical applications.

\citet{AAI2002} have proposed an alternative approach to the identification and estimation of quantile effects with binary endogenous regressors, which builds on the local average treatment effects framework of \citet{AngristImbens1994}. Their approach has been extended and further developed by \citet{Frandsenetal2012}, \citet{FrolichMelly2013jbes}, and  \citet{Bellonietal2017} among others. We refer to \citet{MellyWuthrich2017} for a recent review of this approach and to \citet{Wuthrich2020} for a comparison between this approach and the IVQR model. Identification and estimation in nonseparable models with continuous endogenous regressors have been studied by \citet{Chesher2003}, \citet{KoenkerMa2006}, \citet{Lee2007}, \citet{Jun2009}, \citet{ImbensNewey2009}, \citet{DHaultfoeuilleFevrier2015}, and \citet{Torgovitsky2015} among others.

On a broader level, our paper contributes to the literature which proposes estimation procedures that rely on decomposing computationally burdensome estimation problems into several more tractable subproblems. This type of procedure, which we call decentralization, has been applied in different contexts. Examples include the estimation of single index models with unknown link function \citep{WeisbergWelsh94}, general maximum likelihood problems \citep{Smyth96}, linear models with high-dimensional fixed effects \citep[e.g.,][and the references therein]{GuimaraesPortugal10}, sample selection models \citep{MarraRadice13}, peer effects models \citep{Arcidiacono+12}, interactive fixed effects models \citep[e.g.,][]{Chen+14panel,MoonWeidner15}, and random coefficient logit demand models \citep{LeeSeo15ablp}. Most of these papers decompose a single estimation problem into two subproblems. The present paper explicitly considers cases in which the number of subproblems may exceed two. 

\subsection{Organization of the Paper}
The remainder of the paper is structured as follows. Section \ref{sec:setup} introduces the setup and the IVQR model. Section \ref{sec:decentralization} shows that the IVQR estimation problem can be decentralized into a series of (weighted) conventional QR problems. In Section \ref{sec:population_algorithms}, we introduce population algorithms based on the contraction mapping theorem and root-finders. Section \ref{sec:sample_algorithms} discusses the corresponding sample algorithms. In Section \ref{sec:theory}, we establish the asymptotic normality of our estimators and prove the validity of the bootstrap. Section \ref{sec:empirical_example} presents an empirical application. In Section \ref{sec:simulation_study}, we provide simulation evidence on the computational performance and the finite sample properties of our methods. Section \ref{sec:conclusion} concludes. All proofs as well as some additional theoretical and simulation results are collected in the appendix.

\section{Setup and Model}
\label{sec:setup}

Consider a setup with a continuous outcome variable $Y$, a $d_X\times 1$ vector of exogenous covariates $X$, a $d_D\times 1$ vector of endogenous treatment variables $D$, and a $d_Z\times 1$ vector of instruments $Z$.
The IVQR model is developed within the standard potential outcomes framework \citep[e.g.,][]{Rubin1974}. Let $\{Y_d\}$ denote the (latent) potential outcomes. The object of primary interest is the conditional quantile function of the potential outcomes, which we denote by $q(d,x,\tau)$. Having conditioned on covariates $X=x$, by the Skorokhod representation of random variables, potential outcomes can be represented as
\begin{align*}
Y_d=q(d,x,U_d)~~\text{with}~~U_d\sim U(0,1).
\end{align*}
This representation lies at the heart of the IVQR model.
With this notation at hand, we state the main conditions of the IVQR model \citep[][Assumptions A1-A5]{CH2005}.
\begin{assumption}
\label{ass:ivqr}
Given a common probability space $(\Omega,F,P)$, the following conditions hold jointly with probability one:
\begin{enumerate}
\item Potential outcomes: Conditional on $X=x$, for each $d$, $Y_d=q(d,x,U_d)$, where $q(d,x,\tau)$ is strictly increasing in $\tau$ and $U_d \sim U(0,1)$. \label{ass:ivqr1}
\item Independence: Conditional on $X=x$, $\{U_d\}$ are independent of $Z$.\label{ass:ivqr2}
\item Selection: $D:= \delta(Z,X,V)$ for some unknown function $\delta(\cdot)$ and random vector $V$.\label{ass:ivqr3}
\item Rank invariance or Rank similarity: Conditional on $X=x$, $Z=z$,
\begin{itemize}
\item[(a)]  $\{U_d\}$ are equal to each other; or, more generally,
\item[(b)]  $\{U_d\}$ are identically distributed, conditional on $V$.
\end{itemize}
\label{ass:ivqr4}
\item Observed variables: Observed variables consist of $Y:= q(D,X,U_D)$, $D$, $X$, and $Z$.\label{ass:ivqr5}
\end{enumerate}
\end{assumption}
We briefly discuss the most important aspects of Assumption \ref{ass:ivqr} and refer the interested reader to \citet{CH2005,CH2006,CH2013} for more comprehensive treatments. Assumption \ref{ass:ivqr}.\ref{ass:ivqr1} states the Skorohod representation of $Y_d$ and requires strict monotonicity of the potential outcome quantile function, which rules out discrete outcomes. Assumption \ref{ass:ivqr}.\ref{ass:ivqr2} imposes independence between the potential outcomes and the instrument. Assumption \ref{ass:ivqr}.\ref{ass:ivqr3} defines a general selection mechanism. The key restriction of the IVQR model is Assumption \ref{ass:ivqr}.\ref{ass:ivqr4}. Rank invariance (a) requires individual ranks $U_d$ to be the same across treatment states. Rank similarity (b) weakens this condition, allowing for random slippages of $U_d$ away from a common level $U$. Finally, Assumption \ref{ass:ivqr}.\ref{ass:ivqr5} summarizes the observables. 
\begin{remark} 
Assumption \ref{ass:ivqr} does not impose any restrictions on how the instrument $Z$ affects the endogenous variable $D$. As a consequence, Assumption \ref{ass:ivqr} alone does not guarantee point identification of the structural quantile function. In Sections \ref{sec:decentralization} and \ref{sec:population_algorithms} we discuss sufficient conditions for global (point) identification. Due to the nonlinearity of the IVQR problem, these conditions are stronger than the usual first stage assumptions in linear instrumental variables models and require the instrument to have a nontrivial impact on the joint distribution of $(Y,D)$. The identification conditions are particularly easy to interpret when $D$ and $Z$ are binary; see \citet[][Section 2.4]{CH2005} and Appendix \ref{app:local_contractions} for a further discussion. \qed
\end{remark}

The main implication of Assumption \ref{ass:ivqr} is the following conditional moment restriction \citep[][Theorem 1]{CH2005}:
\begin{equation}
P\left( Y\leq q(D,X,\tau) \mid X,Z \right)=\tau,  \quad \tau\in (0,1).\label{eq:cond_mr}
\end{equation}
In this paper, we focus on the commonly used linear-in-parameters model for $q(\cdot)$ \citep[e.g.,][]{CH2006}:
\begin{equation}
q(d,x,\tau)=x^{\prime }\param_X(\tau)+d'\param_D(\tau),
\end{equation}
where   $\param(\tau):=(\param_X(\tau)',\param_D(\tau)')'\in \mathbb{R}^{d_X+d_D}$ is the finite dimensional parameter vector of interest. The conditional moment restriction \eqref{eq:cond_mr} suggests GMM estimators based on the following unconditional population moment conditions:
\begin{equation}
\Z_{P}\left(\param(\tau) \right):=E_P\left[ \left( 1\left\{ Y\leq X^\prime \param_X(\tau)+D'\theta_D(\tau)
\right\} -\tau\right)\begin{pmatrix}X\\Z\end{pmatrix} \right].\label{eq:uncond_mr}
\end{equation}

Our primary goal here is to obtain estimators in a computationally efficient and reliable manner. We therefore focus on just-identified moment restrictions where $d_Z=d_D$, for which the construction of an estimator is straightforward. A potential caveat of this approach is that  estimators based on these restrictions do not achieve the pointwise (in $\tau$) semiparametric efficiency bound implied by the conditional moment restrictions \eqref{eq:cond_mr}. Appendix \ref{app:overid} provides a discussion of overidentified GMM problems and presents a two-step approach for constructing efficient estimators based on the proposed algorithms. 

In what follows, we will often suppress the dependence on $\tau$ to lighten-up the exposition.
We then define the true parameter value $\true$ as the solution to the moment conditions, i.e.,
\[
\Z_P\left(\true\right)=0.
\]
The resulting GMM objective function reads
\begin{eqnarray}
\mathcal{Q}^{GMM}_N\left(\param\right)=-\frac{1}{2}\left(\frac{1}{\sqrt{N}}\sum_{i=1}^Nm_i\left(\param\right) \right)'W_N\left(\param\right)\left(\frac{1}{\sqrt{N}}\sum_{i=1}^Nm_i\left(\param\right) \right),\label{eq:gmm}
\end{eqnarray}
where $m_i\left(\param\right):=\left( 1\left\{ Y_i\leq X_i^\prime \param_X + D_i'\theta_D
\right\} -\tau\right)(X_i',Z_i')'$ and $W_N\left(\param\right)$ is a positive definite weighting matrix. Estimation based on \eqref{eq:gmm} is complicated by the non-smoothness and, most importantly, the non-convexity of $\mathcal{Q}_N^{GMM}$. This paper proposes a new set of algorithms to address these challenges.

\section{Decentralization}
\label{sec:decentralization}

Here we describe the basic idea behind our decentralization estimators. To simplify the exposition, we first illustrate our approach with the  population problem of finding the true parameter value $\true$ in the IVQR model. Our estimator then adopts the analogy principle, which will be presented in Section \ref{sec:sample_algorithms}.
The key insight is that the complicated nonlinear IVQR estimation problem can be ``decentralized'', i.e., decomposed into a set of more tractable sub-problems, each of which is  solved by a ``player'' who best responds to  other players' actions.
Specifically, we first split the parameter vector $\theta$ into $J$ subvectors $\param_1,\dots,\param_{J}$, where $J=d_D+1$.
We then decompose the grand estimation problem into $J$ subproblems. Each of the subproblems is allocated to a distinct player.
For each $j$, player $j$'s choice variable is the $j$-th subvector $\theta_j$.
 Her problem is to find the value of $\theta_j$ such that a subset  of the moment restrictions is satisfied given the other players' actions $\theta_{-j}$, where $\theta_{-j}$  stacks the components of $\theta$ other than $\theta_j$. This reformulation allows us to view the estimation problem as a game of complete information and to characterize $\true$ as the game's pure strategy Nash equilibrium. 

We start by describing the parameter subvectors.
First, let $\theta_1=\theta_X\in\mathbb R^{d_X}$ denote the vector of coefficients on the exogenous variables. We allocate this subvector to player 1. Similarly, for each $j=2,\dots,J$,  let $\theta_j\in\mathbb R$ denote the coefficient on the $(j-1)$-th endogenous variable, which is allocated to player $j$. The coefficient vector for the endogenous variable can therefore be written as $\theta_D=(\theta_2,\dots,\theta_J)'$.
For each $\theta\in\mathbb{R}^{d_X+d_D}$, define the following (weighted) QR objective functions: 
\begin{eqnarray}
Q_{P,1}\left(\theta\right)&:= &E_P\left[\rho_\tau( Y-X'\theta_1-D_1\theta_2-\dots-D_{d_D}\theta_J)\right],\label{eq:Q1}\\
Q_{P,j}\left(\theta\right)&:=& E_P\left[\rho_\tau( Y-X'\theta_1-D_1\theta_2-\dots-D_{d_D}\theta_J)(Z_{j-1}/D_{j-1})\right],\label{eq:Qj} ~j=2,\dots,J,
\end{eqnarray}
where $\rho_\tau(u)=u(\tau-1\{u<0\}$) is the ``check-function''. We assume that the model is parametrized such that $Z_\ell/D_\ell$ is positive for all $\ell=1,\dots,d_D$.  Under our assumptions, we can always reparametrize the model such that this condition is met; see Appendix \ref{app_sec:reparametrization}  for more details.

The players then solve the following optimization problems:
\begin{align}
&\min_{\tilde\theta_1 \in \mathbb{R}^{d_X}}Q_{P,1}\left( \tilde\theta_1,\theta_{-1}\right)\label{eq:ivqr_br1},\\
&\min_{\tilde\theta_{j} \in \mathbb{R}}Q_{P,j}\left(\tilde\theta_{j} ,\theta_{-j}\right),\quad j=2,\dots,J\label{eq:ivqr_br2}.
\end{align}
Observe that each player's problem is a weighted QR problem, which is convex in its choice variable. For the sample analogues of these problems fast solution algorithms exist \citep[e.g.,][]{Koenker2017}.

For each $j$, let $\tilde L_j(\theta_{-j})$ denote the set of minimizers. Borrowing the terminology from game theory, we refer to these maps as \emph{best response (BR) maps}.
Under an assumption we specify below, the first-order optimality conditions imply that, for each $j$, any element $\tilde\theta^*_j\in \tilde L_j(\theta_{-j})$ of the BR map satisfies
\begin{align}
\Z_{P,1}(\tilde\theta_1^*,\theta_{-1})&\equiv E_P\left[ \left( 1\left\{ Y\leq X^\prime \tilde\theta_1^* 
+D'\param_{-1}\right\} -\tau\right)X \right]=0, \label{eq:ivqr_moment_br1}\\
\Z_{P,j}(\tilde\theta_j^*,\theta_{-j})&\equiv E_P\left[ \left( 1\left\{ Y\leq (X',D_{-(j-1)}')^\prime \param_{-j} + D_{j-1}\tilde\theta^*_j
\right\} -\tau\right)Z_{j-1} \right]=0, \quad j=2,\dots,J, \label{eq:ivqr_moment_br2}
\end{align}
where $D_{-(j-1)}$ stacks as a vector all endogenous variables except $D_{j-1}$.
Note that $\Z_{P}\left(\param\right)=\left(\Z_{P,1}(\param)',\dots, \Z_{P,J}(\param)\right)'$ is the set of unconditional IVQR moment conditions. Hence, $\true$ satisfies
\begin{align}
	\true_j\in \tilde\br_j(\true_{-j}),\quad j=1,\dots,J,
\end{align}
which implies that $\true$ is a fixed point  of the BR-maps (i.e.\ a Nash equilibrium of the game). 

In what follows, we work with conditions that ensure the existence of singleton-valued BR maps $L_j$, $j=1,\dots,J$, such that, for each $j$, $\Z_{P,j}\left(\br_j(\param_{-j}),\param_{-j}\right)=0$.\footnote{While it may be interesting to work with set-valued maps,  the existence of the BR functions greatly simplifies our analysis of identification and inference.} We say that the IVQR estimation problem admits \emph{decentralization} if there exist such BR functions defined over domains for which the moment conditions can be evaluated.\footnote{In Appendix \ref{app_sec:local_decentralization}, we also provide weaker conditions under which the decentralization holds on a local neighborhood of $\true$. We call such a result \emph{local decentralization}, which is sufficient for analyzing the (local) asymptotic behavior of the estimator.}
To ensure decentralization, we make the following assumption.

\begin{assumption}
\label{ass:global_decentralization}
The following conditions hold.
\begin{enumerate}
	\item $\Theta$ is a closed rectangle in $\mathbb{R}^{d_X+d_D}$. $\true$ is in the interior of $\Theta$. \label{ass:global_decentralization1}
	\item $E[|Z_\ell|^2]<\infty$ for $\ell=1,\dots,d_D$. $E[|X_k|^2]<\infty$ for all $k=1,\dots,d_X$. For each $\ell=1,\dots,d_D$, $D_\ell$ has a compact support;\label{ass:global_decentralization2}
	\item The conditional cdf $y\to F_{Y|D,X,Z}(y)$ is continuously differentiable for all $y\in\mathbb R$ $a.s.$ The conditional density $f_{Y|D,Z,X}$ is  uniformly bounded $a.s.$; \label{ass:global_decentralization3}
	\item For any $\theta\in\Theta$, the matrices
	\begin{eqnarray*}
	E_P[f_{Y|D,X,Z}\left(D'\param_{-1}+X^\prime \param_1\right) X X']
	\end{eqnarray*}
	and
	\begin{eqnarray*}
	E_P[f_{Y|D,X,Z}\left(D'\param_{-1} +X^\prime \param_1\right)D_{\ell}Z_{\ell} ],\quad \ell=1,\dots,d_D,
	\end{eqnarray*}
	are positive definite. \label{ass:global_decentralization4}
\end{enumerate}
\end{assumption}

For each $j$, let $\Theta_{-j}\subset \mathbb{R}^{d_{-j}}$ denote the parameter space for $\param_{-j}$.
Assumption \ref{ass:global_decentralization}.\ref{ass:global_decentralization1} ensures that $\Theta$ is compact. This assumption also ensures that each $\Theta_{-j}$ is a closed rectangle, which we use to show that $L_j$ is well-defined on a suitable domain.
Assumptions \ref{ass:global_decentralization}.\ref{ass:global_decentralization2} and  \ref{ass:global_decentralization}.\ref{ass:global_decentralization3} impose standard regularity conditions on the conditional density and the moments of the variables in the model. We assume $D_\ell$ has a compact support, which allows us to always reparameterize the model so that the objective function in \eqref{eq:Qj} is well-defined and convex (cf.\ Appendix \ref{app_sec:reparametrization}).
The first part of Assumption \ref{ass:global_decentralization}.\ref{ass:global_decentralization4} is a standard full rank condition which is a natural extension of the local full rank condition required for local identification and decentralization (cf.\ Assumption \ref{ass:local_decentralization} in the appendix). For the second part of Assumption \ref{ass:global_decentralization}.\ref{ass:global_decentralization4}, it suffices that the model is parametrized such that, for each $\ell\in\{1,\dots,d_D\}$, $D_\ell Z_\ell$ (and $Z_\ell/D_\ell$) is positive with probability one.

For each $j$, define
\begin{align}
R_{-j}\equiv \{\theta_{-j}\in \Theta_{-j}:\Z_{P,j}(\theta)=0,\text{ for some }\theta=(\theta_j,\theta_{-j})\in \Theta\}.\label{eq:defRj}
\end{align}
This is the set of  subvectors $\theta_{-j}$ for which one can find $\theta_j\in\Theta_j$ such that $\theta=(\theta_j,\theta_{-j})'$ solve the $j$-th moment restriction.
We take this set as the domain of player $j$'s best response function $L_j$.

The following lemma establishes that the IVQR model admits decentralization.
\begin{lemma}
\label{lem:global_decentralization}
\text{ } Suppose that Assumptions \ref{ass:ivqr} and \ref{ass:global_decentralization} hold. Then,
there exist functions $\br_j:R_{-j}\to \mathbb{R}^{d_j},j=1,\dots,\dpl$  such that, for $j=1,\dots,J$,
\begin{eqnarray}
\Z_{P,j}\left(\br_j(\param_{-j}),\param_{-j}\right)=0,\quad \text{for all}~\param_{-j}\in R_{-j}.
\end{eqnarray}
Further, $\br_j$ is continuously differentiable on the interior of $R_{-j}$ for all $j=1,\dots,J$.
\end{lemma}

We now introduce maps that represent all players' (joint) best responses. We consider two basic choices of such maps; one represents simultaneous responses and the other represents sequential responses. In what follows, for any subset $a\subset\{1,\dots,J\}$,  let $\theta_{-a}$ denote the subvector of $\theta$ that stacks the components of $\theta_j$'s for all $j\notin a$. If $a$ is a singleton (i.e. $a=\{j\}$ for some $j$), we simply write $\theta_{-j}$. For each $j$ and $a\subseteq\{1,\dots,J\}\setminus\{j\}$, let  $\pi_{-a}:\Theta_{-j}\to \prod_{k\in \{1,\dots,J\}\setminus(\{j\}\cup a)}\Theta_k$ be the coordinate projection of $\theta_{-j}$ to a further subvector that stacks all components of $\theta_{-j}$ except for those of $\theta_k$ with $k\in a$.

Let $D_K\equiv\{\theta\in\Theta: \pi_{-j}\theta\in R_{-j},~ j=1,\dots,J\}$. Let $K:D_K\to\mathbb{R}^{d_X+d_D}$ be a map defined by
\begin{eqnarray}
K(\param)=\begin{pmatrix} K_1(\param) \\  \vdots \\ K_J(\param)\end{pmatrix}=\begin{pmatrix} \br_1(\param_{-1})\\ \vdots \\ \br_{J}(\param_{-J})\end{pmatrix}.
\end{eqnarray}
This can be interpreted as the players' simultaneous best responses to the initial strategy $(\theta_1,\dots,\param_{J})$.  With one endogenous variable, this map simplifies to
\begin{eqnarray}
	K(\theta)=\begin{pmatrix}\br_1(\theta_{2})\\ \br_2(\theta_{1})\end{pmatrix}.
\end{eqnarray}
Here, $K$ maps $\theta=(\theta_1,\theta_2)$ to a new parameter value through the simultaneous best responses of players 1 and 2.

Similarly, let $D_M\subset \mathbb R^{d_D}$ and let $M:D_M\to\mathbb R^{d_D}$ be a map such that
\begin{eqnarray}
M\left(\theta_{-1} \right)=\begin{pmatrix}M_1(\theta_{-1})\\M_2(\theta_{-1})\\ \vdots \\M_{d_D}(\theta_{-1}) \end{pmatrix}=
\begin{pmatrix}L_2\left( L_1(\theta_{-1}),\theta_{-\{1,2\}}\right)\\ L_3\left( L_1(\theta_{-1}),L_2( L_1(\theta_{-1}),\theta_{-\{1,2\}}),\theta_{-\{1,2,3\}}\right)\\\vdots \\ L_J\left(L_{1}( \theta_{-1}),L_2( L_1(\theta_{-1}),\theta_{-\{1,2\}}),\cdots\right)\end{pmatrix},\label{eq:def_M}
\end{eqnarray}
which can be interpreted as the players' sequential responses (first by player 1, then player 2, etc.) to an initial strategy $\theta_{-1}=(\theta_{2},\dots,\theta_{J})$.\footnote{One may define $M$ by changing the order of responses as well. For theoretical analysis, it suffices to consider only one of them. Once the fixed point $\true_{-1}$ of $M$ is found, one may also obtain $\true_1$ using $\true_1=L_1(\true_{-1})$.} Note that the argument of $M$ is not the entire parameter vector. Rather, it is a subvector of $\theta$ consisting of the coefficients on the endogenous variables. In order to find a fixed point, this feature is particularly attractive  when the number of endogenous variables is small.
 With one endogenous variable (i.e. $\theta_2\in \mathbb R$ is a scalar),
the map simplifies to
\begin{eqnarray*}
	M(\theta_2)=L_2\left(L_1\left(\theta_2\right)\right),
\end{eqnarray*}
which is a univariate function whose fixed point is often straightforward to compute.

Define
\begin{align}
	\tilde R_1\equiv \big\{\theta_{-1}\in\Theta_{-1}:~&\Z_{P,1}(\theta_1,\theta_{-1})=0,\notag\\
	&\Z_{P,2}(\theta_1,\theta_2,\pi_{-\{1,2\}}\theta_{-1})=0,~\exists (\theta_1,\theta_2)\in\Theta_1\times\Theta_2\big\}.\label{eq:tildeR1}
\end{align}
This is the set on which the map $\theta_{-1}\to L_2\left( L_1\left(\theta_{-1}\right),\pi_{-\{1,2\}}\theta_{-1}\right)$, the first component of $M$, is well-defined. We then recursively define $\tilde R_j$ for $j=2,\dots,d_D$ in a similar manner. A precise definition of these sets is  given in Appendix \ref{app_sec:decentralization}. Now define
\begin{align}
	D_M\equiv \bigcap_{j=1}^{d_D}\tilde R_j=\tilde R_{d_D},\label{eq:def_DM}
\end{align}
where the second equality follows because $\tilde R_{d_D}$ turns out to be a subset of $\tilde R_j$ for all $j\le d_D$. The following corollary ensures that $K$ and $M$ are well-defined on $D_K$ and $D_M$ respectively.

\begin{corollary}\label{cor:global_decentralization}
	The maps $K:D_K\to \mathbb{R}^{d_X+d_D}$ and $M: D_M\to\mathbb{R}^{d_D} $ exist and are continuously differentiable on the interior of their domains.
\end{corollary}

The key insight that we exploit is that, by construction of the BR maps, the problem of finding a solution to $\Z_P(\theta)=0$ is equivalent to the problem of finding a fixed point of $K$ (or $M$). The following proposition states the formal result.
\begin{proposition}
\label{prop:equivalence}
Suppose that Assumptions \ref{ass:ivqr} and \ref{ass:global_decentralization} hold. Then,
\begin{enumerate}[label=(\roman*)]
\item $\Z_{P}(\true)=0$ if and only if $K\left(\true \right)=\true$
\item $\Z_{P}(\true)=0$ if and only if $M\left(\true_{-1}\right)=\true_{-1}$ and $\true_1=L_1(\true_{-1})$.
\end{enumerate}
\end{proposition}

In view of Proposition \ref{prop:equivalence},  the original IVQR estimation problem can be reformulated as the problem of finding the fixed point of $K$ (or $M$).
This reformulation naturally leads to discrete dynamical systems associated with these maps, which in turn provide straightforward iterative algorithms for computing $\true$.

\begin{enumerate}
\item \textsc{Simultaneous dynamical system:}\footnote{This algorithm is akin to the Jacobi computational procedure.}
\begin{align}
\theta^{(s+1)}=K\left(\theta^{(s)}\right),~~s=0,1,2,\dots,~~\param^{(0)}~\text{given}.\label{eq:simultaneous_ds}
\end{align}
\item \textsc{Sequential dynamical system:}\footnote{\citet{Smyth96} considers  this type of algorithm for $J=2$ and calls it ``zigzag'' algorithm. It is akin to a Gauss-Seidel procedure.}
\begin{align}
\param_{-1}^{(s+1)}&=M\left(\param_{-1}^{(s)} \right),~~s=0,1,2,\dots,~~\param_{-1}^{(0)}~\text{given},\label{eq:sequential_ds}
\end{align}
where $\param_1^{(s+1)}=L_1\left(\param_{-1}^{(s)}\right)$.
\end{enumerate}
These discrete dynamical systems constitute the basis for our estimation algorithms.\footnote{These discrete dynamical systems can also be viewed as learning dynamics in a game \citep{LiBasar87,FudenbergLevine07}.} 

\section{Population Algorithms}
\label{sec:population_algorithms}
In this section, we explore the implications of the fixed point reformulation for constructing population-level algorithms for computing fixed points.

\subsection{Contraction-based Algorithms} \label{ssec:contraction_algorithms}
We first consider conditions under which $K$ and $M$ are contraction mappings. They ensure that the discrete dynamical systems induced by $K$ and $M$ are convergent to unique fixed points. Moreover,
in view of Proposition \ref{prop:equivalence}, (point) identification is equivalent to the uniqueness of the fixed point of $K$ (or $M$). Therefore, the conditions we provide below are also sufficient for the point identification of $\true.$ We will discuss the relationship between our conditions and existing ones in the next section.

For any vector-valued map $E$, let $J_E(x)$ denote its Jacobian matrix evaluated at its argument $x$. For any matrix $A$, let $\|A\|$ denote its operator norm. We provide conditions in terms of the Jacobian matrices of $K$ and $M$, which are well-defined by Corollary \ref{cor:global_decentralization}.

\begin{assumption} \label{ass:contraction_global}
There exist open  strictly convex sets $\tilde D_K\subseteq D_K$ and $\tilde D_M\subseteq D_M$ such that
\begin{enumerate}
	\item	 \label{ass:contraction_global_K} $\|J_K\left(\param\right)\|\le \lambda$ for some $\lambda<1$ for all $\theta \in \tilde D_K$;
	\item \label{ass:contraction_global_M} $\|J_M\left(\param_{-1} \right) \|\le \lambda$ for some $\lambda<1$ for all $\theta_{-1} \in \tilde D_M$.
\end{enumerate}
\end{assumption}
Under this additional assumption, the iterative algorithms are guaranteed to converge to the fixed point. We summarize this result below.
\begin{proposition} \label{prop:contraction_global}Suppose that Assumptions \ref{ass:ivqr} and \ref{ass:global_decentralization} hold.
\begin{enumerate}[label=(\roman*)]
\item Suppose further that Assumption \ref{ass:contraction_global}.\ref{ass:contraction_global_K} holds. Then $K$ is a contraction on the closure of $\tilde{D}_K$. The fixed point $\true\in cl(\tilde{D}_K)$ of $K$ is unique. For any $\theta^{(0)}\in \tilde D_K,$ the sequence $\{\theta^{(s)}\}_{s=0}^\infty$ defined in \eqref{eq:simultaneous_ds} satisfies
$\theta^{(s)}\to \true$ as $s\to\infty$.
\item Suppose further that Assumption \ref{ass:contraction_global}.\ref{ass:contraction_global_M} holds. Then $M$ is a contraction on the closure of $ \tilde{D}_M$. The fixed point $\true_{-1} \in  cl(\tilde{D}_M)$ of $M$ is unique.  For any $\theta^{(0)}_{-1}\in \tilde D_M,$ the sequence $\{\theta^{(s)}_{-1}\}_{s=0}^\infty$ defined in \eqref{eq:sequential_ds} satisfies
$\theta^{(s)}_{-1}\to \true_{-1}$ as $s\to\infty$.
\end{enumerate}
\end{proposition}
In the case of a single endogenous variable,  the Jacobian matrices of $K$ and $M$   are given by
\begin{eqnarray}
J_K(\param)=\begin{pmatrix}0 & J_{L_1}(\param_2)\\J_{L_2}(\param_1)&0 \end{pmatrix}
\qquad\text{and}\qquad
J_{M}(\param_2)=J_{L_2}\left(L_1 (\param_2)\right)J_{L_1}\left(\param_2 \right),
\end{eqnarray}
where
\begin{eqnarray}
J_{L_{-j}}\left( \theta_j\right) =-\left(\frac{\partial\Z_{P,-j}(\theta_j,\theta_{-j})}{\partial\theta_{-j}'}\bigg|_{\theta=(\theta_j,L_{-j}( \theta_j))} \right)^{-1}\frac{\partial\Z_{P,-j}(\theta_j,\theta_{-j})}{\partial\theta_j'}\bigg|_{\theta=(\theta_j,L_{-j}(\theta_j))},~ j=1,2.
\label{eq:def_JL}
\end{eqnarray}
One may therefore check the high-level condition through the Jacobians of the original moment restrictions. In Appendix \ref{app:local_contractions}, we illustrate a simple primitive condition for a local version of Assumption \ref{ass:contraction_global}. In practice, we found that violations of Assumption \ref{ass:contraction_global} lead to explosive behavior of the estimation algorithms and, thus, are very easy to detect numerically.

\subsection{Connections to the Identification Conditions in the Literature} \label{ssec:identification_algorithms}

In view of Proposition \ref{prop:equivalence},
 identification of $\true$ is equivalent to  uniqueness of the fixed points of $K$ and $M$, which is ensured by Proposition \ref{prop:contraction_global}.
Here, we discuss how the conditions required by Proposition \ref{prop:contraction_global} relate to the ones in the literature.

We start with local identification.
The parameter vector $\true$ is said to be locally identified if there is a neighborhood $\mathcal{N}$ of $\true$ such that $\Z_P(\param)\ne 0$ for all $\param\ne \true$ in the neighborhood. Local identification in the IVQR model follows from standard results \citep[e.g.,][]{Rothenberg71,Chen+14}. For example, if $\Z_P(\param)$ is differentiable, \citet[][Section 2.1]{Chen+14} show that full rank of $J_{\Z_P}(\param)$ at $\true$ is sufficient for local identification.

It is interesting to compare this full rank condition to Assumption \ref{ass:contraction_local}.\ref{ass:contraction_local_K} in the appendix, which is a local version of Assumption \ref{ass:contraction_global}.\ref{ass:contraction_global_K}. Assumption \ref{ass:contraction_local}.\ref{ass:contraction_local_K}  requires that $\rho\left(J_K(\true) \right)<1$, where $\rho(A)$ denotes the spectral radius of a square matrix $A$. We highlight the connection in the case with a single endogenous variable. Full rank of $J_{\Z_P}(\true)$ is equivalent to $\det\left(J_{\Z_P}(\true)\right)\ne 0$.  Observe that, for any $\theta$,
\begin{eqnarray*}
\det\left(J_{\Z_P}(\theta)\right)&=&
\det\begin{pmatrix}
\partial\Z_{P,1}(\param_1,\param_2)/\partial\param_1' & \partial\Z_{P,1}(\param_1,\param_2)/\partial\param_2'\\ \partial\Z_{P,2}(\param_1,\param_2)/\partial\param_1' & \partial\Z_{P,2}(\param_1,\param_2)/\partial\param_2'\end{pmatrix}\\
&=&
\det\left(\begin{pmatrix}
 \partial\Z_{P,1}(\param_1,\param_2)/\partial\param_1'&0\\
0&\partial\Z_{P,2}(\param_1,\param_2)/\partial\param_2'
\end{pmatrix}
\begin{pmatrix}I_{d_1} & -J_{L_1}(\param_2)\\
    -J_{L_2}(\param_1)& I_{d_2}
\end{pmatrix}\right)\\
 &=&\det\begin{pmatrix}
 \partial\Z_{P,1}(\param_1,\param_2)/\partial\param_1'&0\\
 0&\partial\Z_{P,2}(\param_1,\param_2)/\partial\param_2'
  \end{pmatrix}
   \det\begin{pmatrix}I_{d_1} & -J_{L_1}(\param_2)\\
   -J_{L_2}(\param_1)&I_{d_2} \end{pmatrix}.
\end{eqnarray*}
If $\partial\Z_{P,j}(\theta)/\partial\theta_j'|_{\theta=\true} $ is invertible for $j=1,2$ (which is true under Assumption \ref{ass:global_decentralization}.\ref{ass:global_decentralization4}), $J_{\Z_P}(\true)$ is full rank if and only if
\begin{eqnarray}
0\ne \det
 \begin{pmatrix}I_{d_1} & -J_{L_1}(\true_2)\\ -J_{L_2}(\true_1)&I_{d_2} \end{pmatrix}=\det(I_{\dparam}-J_{K}(\true)). \label{eq:full_rank_eig_K}
\end{eqnarray}
That is, it requires that none of the eigenvalues of $J_K$ has modulus one.
Therefore, Assumption \ref{ass:contraction_local}.\ref{ass:contraction_local_K} is sufficient but not necessary for condition \eqref{eq:full_rank_eig_K} to hold. Specifically,
Assumption \ref{ass:contraction_local}.\ref{ass:contraction_local_K} requires all eigenvalues of $J_{K}(\true)$ to lie strictly within the unit circle, while  local
identification only requires all eigenvalues not to be on the unit circle. In terms of the dynamical system induced by $K$, the former ensures that the dynamical system has a unique \emph{asymptotically stable fixed point}, while the latter  ensures that the system has a unique \emph{hyperbolic fixed point}, which is a more general class of fixed points \citep[e.g.][]{galor2007discrete}.\footnote{The argument above also applies to settings with multiple endogenous variables. A similar result can also be shown for $M$.}
Under the former condition, iteratively applying the contraction map induces convergence, while the latter generally requires a root finding method to obtain the fixed point.

Now we turn to global identification and compare Proposition \ref{prop:contraction_global} to the global identification result in \citet{CH2006}.

\begin{lemma}[Theorem 2 in \citet{CH2006}] \label{lem:CH}Suppose that Assumption \ref{ass:ivqr} holds. Moreover, suppose that (i) $\Theta$ is compact and convex and $\true$ is in the interior of $\Theta$; (ii) $f_{Y|D,Z,X}$ is uniformly bounded a.s.; (iii) $J_{\Z_P}(\theta)$ is continuous and has full rank uniformly over $\Theta$; and (iv) the image of $\Theta$ under the mapping $\theta \to \Z_P(\theta)$ is simply connected. Then, $\true$ uniquely solves $\Z_P(\param)=0$ over $\Theta$.
\end{lemma}
Under Conditions (i)--(iv), which are substantially stronger than the local identification conditions discussed above, the result in Lemma \ref{lem:CH} follows from an application of Hadamard's global univalence theorem (e.g. Theorem 1.8 in \citet{ambrosetti1995primer}).

Comparing Lemma \ref{lem:CH} to Proposition \ref{prop:contraction_global}, we can see that the result in Lemma \ref{lem:CH} establishes identification over the whole parameter space $\Theta$, while Proposition \ref{prop:contraction_global} establishes identification over the sets $\tilde{D}_K$ and $\tilde{D}_M$, which will generally be subsets of $\Theta$. Regarding the underlying assumptions, Conditions (i) and (ii) in Lemma \ref{lem:CH} correspond to our Assumptions \ref{ass:global_decentralization}.\ref{ass:global_decentralization1} and \ref{ass:global_decentralization}.\ref{ass:global_decentralization3}. Moreover, our Assumption \ref{ass:global_decentralization}.\ref{ass:global_decentralization3} constitutes an easy-to-interpret sufficient condition for continuity of $J_{\Z_P}$ as required in Condition (iii). To apply Hadamard's global univalence theorem, \citet{CH2006} assume the simple connectedness of the image of $\Z_P$ (Condition (iv)). By contrast, we use a different univalence theorem by \citet{gale1965jacobian} (applied to the map $\Xi$ defined in \eqref{eq:xi} that arises from each subsystem), which does not require further conditions. However, when establishing global identification based on the contraction mapping theorem, we need to impose an additional condition on the Jacobian (Assumption \ref{ass:contraction_global}). In sum, our conditions are somewhat stronger in terms of  restrictions on the Jacobian, but they are relatively easy to check and allow us to dispense with an abstract condition (simple connectedness of the image of a certain map) to apply a global univalence theorem.

\subsection{Root-Finding Algorithms and Nesting} \label{sec:root_finder} Assumption \ref{ass:contraction_global} is a sufficient condition for the uniqueness of the fixed point and the convergence of the contraction-based algorithms. Even in settings where this assumption fails to hold, one may still identify $\true$ and design an algorithm that is able to find it under weaker conditions on the Jacobian. This is the case under the assumptions in the general (global) identification result of \citet[][]{CH2006}; see Lemma \ref{lem:CH}.

Note that, for the simultaneous dynamical system, $\true$ solves
\begin{align}
  (I_{\dparam}-K)(\true)=0,
\end{align}
where $I_{\dparam}$ is the identity map. Similarly, in the sequential dynamical system, $\true_{-1}$ solves
\begin{align}
  (I_{d_D}-M)(\true_{-1})=0.
\end{align}
 Therefore, standard root-finding algorithms can be used to compute the fixed point. 

For implementing root-finding algorithms, we find that reducing the dimension of the fixed point problem is often helpful.
Toward this end, we briefly discuss another class of dynamical systems and associated population algorithms which can be used for the purpose of dimension reduction.
Namely, with more than two players, one can construct \emph{nested} dynamical systems, which induce nested fixed point algorithms. Nesting is useful as it allows for transforming any setup with more than two players into a two-player system.

To fix ideas, consider the case of three players ($J=3$). Fix player 3's action $\theta_3\in\Theta_3\subset \mathbb R$ and consider the associated ``sub-game'' between players 1 and 2. To describe the sub-game, define $M_{1,2|3}(\cdot\mid \param_3):\Theta_2\to\Theta_2$ pointwise by
\begin{align}
M_{1,2|3}(\theta_2\mid \param_3)\equiv L_2\left(L_1\left(\theta_2,\param_3\right), \param_3\right).
\end{align}
This map gives the sequential best responses of players 1 and 2, while taking player 3's strategy given.
Define the fixed point $L_{12}:\Theta_3\to \Theta_1\times\Theta_2$ of the sub-game by
\begin{align}
L_{12}(\param_3)\equiv \begin{pmatrix}
\fp_1(\param_3)\\
\fp_2(\param_3)
\end{pmatrix}=\begin{pmatrix}\br_1(\fp_2(\theta_3),\theta_3)\\
M_{1,2|3}(\fp_2(\param_3)\mid \param_3)\\
\end{pmatrix}.
\end{align}
This map then defines a new ``best response'' map. Here, given $\param_3$, the players in the sub-game (i.e., players 1 and 2) collectively respond  by choosing the Nash equilibrium of the sub-game.  The overall dynamical system induced by the nested decentralization is then given by
\begin{align}
M_3(\theta_3)=L_3\left(\br_{12}(\theta_3)\right).\label{eq:nestedbr}
\end{align}
Hence, we can interpret the nested algorithm as a two-player dynamical system where one player solves an internal fixed point problem. The nesting algorithms require existence and uniqueness of the fixed points in the sub-game between players 1 and 2. In Appendix \ref{app:nesting}, we discuss the formal conditions required for the existence and uniqueness of such fixed points.\footnote{For an equilibrium of the sub-game to be well-defined, one may directly assume that \cite{CH2006}'s global identification condition holds within the sub-game, given player 3's action $\theta_3$.   Alternatively, if Assumption \ref{ass:contraction_global} holds, we fix  $\theta_3$ to a value such that $(\theta_2,\theta_3)\in \tilde D_M$ for some $\theta_2\in\Theta_2$; see Appendix \ref{app:nesting}.}

This nesting procedure is generic and can be extended to more than three players by sequentially adding additional layers of nesting.\footnote{In the current example, consider adding player 4 and letting players 1-3 best respond by  returning the fixed point of the sub-game through $M_3$  given $\theta_4$. One can repeat this for additional players.
This procedure can also be applied to the simultaneous dynamical system induced by $K$.} It follows that  any decentralized estimation problem with more than two players can be reformulated as a nested dynamical system with two players: player $J$ and all others $-J$.
The resulting dynamical system $M_J(\theta_J)=L_{J}\left(\br_{-J}(\theta_J)\right)$ is particularly useful when $M_J$ is not necessarily a contraction map since $\theta_J$ is a scalar such that, as we see below,  its fixed point can be computed using univariate root-finding algorithms.

\section{Sample Estimation Algorithms}
\label{sec:sample_algorithms}

\subsection{Sample Estimation Problem}
Let $\left\{(Y_i,D_i',X_i',Z_i')\right\}_{i=1}^N$ be a sample generated from the IVQR model. Our estimators are constructed using the analogy principle.
For this, define the sample payoff functions for the players as
\begin{align}
  Q_{N,1}\left(\theta\right)&\equiv \frac{1}{N}\sum_{i=1}^N\rho_\tau( Y_i-X_i'\theta_1-D_{1,i}\theta_2-\dots-D_{d_D,i}\theta_J),\label{eq:QN1}\\
  Q_{N,j}\left(\theta\right)&\equiv \frac{1}{N}\sum_{i=1}^N\rho_\tau( Y_i-X_i'\theta_1-D_{1,i}\theta_2-\dots-D_{d_D,i}\theta_J)\frac{Z_{j-1,i}}{D_{j-1,i}},\quad j=2,\dots,J.\label{eq:QNj}
\end{align}
For each $j=1,\dots, J$, let the sample BR function $\hat\br_j(\param_{-j})$ be a function such that
\begin{eqnarray}
\hat \br_1\left(\param_{-1} \right)&\in&\arg\min_{\tilde\theta_1 \in \mathbb{R}^{d_X}}Q_{N,1}(\tilde\theta_1,\theta_{-1}),\label{eq:brN1}\\
 \hat \br_j\left(\theta_{-j} \right)&\in&\arg\min_{\tilde\theta_j \in \mathbb{R}} Q_{N,j}(\tilde{\theta_j},\theta_{-j}),\quad j=2,\dots,J.\label{eq:brNj}
\end{eqnarray}
Assuming that the model is parametrized in such a way that $Z_{\ell,i}/D_{\ell,i}$, $\ell=1,\dots,d_D$, is positive, these are convex (weighted) QR problems for which fast solution algorithms exist. In our empirical applications and simulations, we use the \texttt{R}-package \texttt{quantreg} to estimate the QRs \citep{quantreg2018}.
For example, $\hat \br_2(\theta_{-2})$ can be computed by running a QR with weights $Z_{1,i}/D_{1,i}$ in which one regresses $Y_i-X_i'\theta_1-D_{2,i}\theta_3-\dots-D_{d_D,i}\theta_J$ on $D_{1,i}$ without a constant. These sample BR functions also approximately solve the sample analog of the moment restrictions in \eqref{eq:ivqr_moment_br1}--\eqref{eq:ivqr_moment_br2}; see Lemma \ref{lem:sample_BR} in the appendix.

\begin{remark}
The proposed estimators rely on decentralizing the original non-smooth and non-convex IVQR GMM problem into a series of convex QR problems. The quality and the computational performance of our procedures therefore crucially depends on the choice of the underlying QR estimation approach, which deserves some further discussion. The interested reader is referred to \citet{Koenker2017} for an excellent overview on the computational aspects of QR. In this paper, we use the Barrodale and Roberts algorithm which is implemented as the default in the \texttt{quantreg} package and described in detail in \citet{KoenkerdOrey1987,KoenkerdOrey1994}. This algorithm is computationally tractable for problems up several thousand observations. For larger problems, we recommend using interior point methods, potentially after preprocessing; see \citet{Portnoy:1997aa} for a detailed description. These methods are conveniently implemented in the \texttt{quantreg} package. For very large problems, one can resort to first-order gradient descent methods, which are amenable to modern parallelized computation; see Section 5.5 in \citet{Koenker2017} for an introduction and simulation evidence on the performance of such methods. \qed
\end{remark}


We focus on estimators constructed based on the dynamical system $M$. Lemma \ref{lem:asy_equiv} in the appendix shows that the estimators based on $M$ and $K$ are asymptotically equivalent. However, in our simulations, we found that, while the convergence properties of contraction algorithms based on $M$ are excellent, the convergence properties of contraction algorithms based on $K$ can be  sensitive to the choice of starting value. This may be attributed to the fact that the algorithm based on $K$ requires all $d_X+d_D$ components of the starting value  to be in the domain of the contraction map, while the algorithm  based on $M$ only requires that the same condition is satisfied by the starting value for $\theta_{-1}$. Our simulations suggest that it is often easier to satisfy this requirement with $M$ since the number of components in $\theta_{-1}$, $d_D$, is  typically much smaller than $d_X+d_D$.
Also, for root-finding algorithms, we prefer the sequential dynamical system (induced by $M$) because it again leads to a substantial dimension reduction: it reduces the original $(d_X+d_D)$-dimensional GMM estimation problem to a $d_D$-dimensional root-finding problem. 

We construct estimation algorithms by mimicking the population algorithms. Let  $\hat M$ denote a sample analog of $M$:
\begin{align}
\hat M\left(\param_{-1}\right):=
\begin{pmatrix}\hat M_1(\theta_{-1})\\\hat M_2(\theta_{-1})\\ \vdots \\\hat M_{d_D}(\theta_{-1}) \end{pmatrix}=
\begin{pmatrix}\hat L_2\left(\hat L_1(\theta_{-1}),\theta_{-\{1,2\}}\right)\\\hat L_3\left(\hat L_1(\theta_{-1}),\hat L_2( \hat L_1(\theta_{-1}),\theta_{-\{1,2\}}),\theta_{-\{1,2,3\}}\right)\\\vdots \\ \hat L_J\left(\hat L_{1}( \theta_{-1}),\hat L_2(\hat L_1(\theta_{-1}),\theta_{-\{1,2\}}),\cdots\right)\end{pmatrix},
\end{align}
where $\param_1=\hat\br_1\left(\param_{-1}\right)$.
This map induces a sample analog of the sequential dynamical system in Section \ref{sec:decentralization}.
\begin{enumerate}
\item[] \textsc{Sample sequential dynamical system:}
\begin{align}
\param_{-1}^{(s+1)}&=\hat M\left(\param_{-1}^{(s)} \right),~~s=0,1,2,\dots,~~\param_{-1}^{(0)}~\text{given},\label{eq:sample_sequential_ds}
\end{align}
where $\param_1^{(s+1)}=\hat{L}_1\left(\param_{-1}^{(s)}\right)$.
\end{enumerate}

\subsection{Contraction-based Algorithms}\label{ssec:contraction_alg}

The first set of algorithms exploits that, under Assumption \ref{ass:contraction_global},  $\hat M$ is a contraction mapping with probability approaching one. In this case, we iterate the dynamical system or \eqref{eq:sample_sequential_ds} until $\|\param_{-1}^{(s)}-\hat M(\param_{-1}^{(s)})\|$
is within a numerical tolerance $e_N$.\footnote{In the next section, we require $e_N=o(N^{-1/2})$, which ensures that the numerical error does not affect the asymptotic distribution.}  This iterative algorithm is known to converge at least linearly. The approximate sample fixed point $\est=(\hat\theta_{N,1},\hat\theta_{N,-1})$ that meets the convergence criterion  then serves as an estimator for $\param$.

\subsection{Algorithms based on Root-Finders and Optimizers}
\label{sec:algo_root_finder}
We construct an estimator $\est$ of $\true$ as an approximate fixed point to the sample problem:
\begin{align}\big\|\hat \param_{N,-1} - \hat M(\hat \param_{N,-1})\big\|\le e_N,\end{align}
where $\hat\param_{N,1}=\hat\br_1(\hat \param_{N,-1} )$ and $e_N$ is a numerical tolerance. This problem can be solved efficiently using well-established root-finding algorithms since $\hat M$ is easy to evaluate as the composition of standard QRs. When $d_D=1$, one may use Brent's method \citep{Brent1971} whose convergence is superlinear. When $d_D>1$, one could apply the Newton-Raphson method, which achieves quadratic convergence but requires an estimate or a finite difference approximation of the derivative. The corresponding approximation error may affect the performance. Alternatively, one can compute the fixed point by minimizing $\| \hat M(\param) - \param \|^2$. The potential issue with this approach is that translating the root-finding problem into a minimization problem can lead to local minima in the objective function. Therefore, it is important to use global optimization strategies. 

As described in Section \ref{sec:root_finder}, nesting can be used to reduce the dimensionality even further. In particular, the  problem can be reformulated as a one-dimensional fixed point problem, which can be solved using existing methods. We found that Brent's method performs very well in our context. Nesting is suitable when the number of endogenous variables is moderate. While adding nests incurs  additional computational costs, our Monte Carlo experiments suggest that they are not excessive when the number of endogenous variables is moderate.\footnote{See Section \ref{sec:simulation_study} and Appendix \ref{app:additional_simulation}.}

\subsection{Profiling Algorithms} \label{sec:profiling} A key insight underlying our estimation algorithms is that, given $\theta_{-1}$, the estimation problem becomes a standard convex QR problem:
\[
\hat\br_1\left(\param_{-1} \right) \in \arg\min_{\tilde\param_1 \in \mathbb{R}^{d_X}}Q_{N,1}(\tilde\param_1,\param_{-1})
\]
This insight suggests a profiling estimator.\footnote{We thank an anonymous referee for suggesting this alternative estimator.} In particular, $\hat\param_{N,-1}$ can be obtained as the approximate root of the function
\begin{equation}
f(\theta_{-1})=\frac{1}{N}\sum_{i=1}^N\left(1\left\{ Y_i\le X_i'\hat\br_1(\param_{-1})+D_{i}'\param_{-1}\right\}-\tau \right)Z_i\label{eq:profiling}
\end{equation}
and $\hat\param_{N,1}$ can be estimated as $\hat\param_{N,1}\in \arg\min_{\tilde\param_1 \in \mathbb{R}^{d_X}}Q_{N,1}(\tilde\param_1,\hat\param_{N,-1})$. 
When there is only one endogenous variable, $f$ is scalar-valued and univariate root-finders can be used.  With multivariate endogenous variables, one can either directly apply multivariate root-finders or construct nested algorithms as described in Sections \ref{sec:root_finder} and \ref{sec:algo_root_finder}.

Relative to the root-finding methods described in Section \ref{sec:algo_root_finder}, the profiling estimator has the advantage that evaluating $f$ only requires estimating one QR. On the other hand, the root-finding algorithms in Section \ref{sec:algo_root_finder} efficiently exploit the convexity of the subproblems for players $j=2,\dots,J$, and demonstrate a better computational performance than the profiling estimators with multiple endogenous variables (cf.\ Table \ref{tab:cpu2}). 

\section{Asymptotic Theory}\label{sec:theory}

\subsection{Estimators}
We  define an estimator $\est$ of $\true$  as an approximate fixed point of $\hat M$ in the following sense:%
\begin{align}
\|\hat\theta_{N,-1}-\hat M(\hat\theta_{N,-1})\|\le \inf_{\theta'_{-1}\in\Theta_{-1}}\|\theta'_{-1}-\hat M(\theta'_{-1})\|+o_p(N^{-1/2}).\label{eq:hattheta1}
\end{align}
An estimator of $\true_1$ can be constructed by setting
\begin{align}
\hat\theta_{N,1}\equiv \hat L_1(\hat\theta_{N,-1})\label{eq:hattheta2}.
\end{align}
 In what follows, we call $\est=(\hat\theta_{N,1},\hat\theta_{N,-1})$ the \emph{fixed point estimator} of $\true.$
 Under the conditions we introduce below, $\est$ is also an approximate fixed point of $\hat K$; see Lemma \ref{lem:asy_equiv} in the appendix. This turns out to be useful for stating theoretical results in a concise manner. While we mostly focus on algorithms based on $\hat M$ below, some of our theoretical results will be stated using $K$.
$\hat M$ (or $\hat K$) is defined similarly for the nested dynamical system in which one player solves a fixed point problem in a sub-game.

Consistency and parametric convergence rates of $\est$ can be established using existing results.
When  $\hat M$ is asymptotically  a contraction map, one may construct an estimator $\est$ satisfying \eqref{eq:hattheta1}--\eqref{eq:hattheta2} using the contraction algorithm in Section \ref{ssec:contraction_alg} with tolerance $e_N=o(N^{-1/2}).$ One may then apply the result of \cite{Dominitz:2005aa} to obtain the root-$N$ consistency of the estimator.\footnote{Satisfying $e_N=o(N^{-1/2})$ requires the number of iterations to increase as the sample size tends to infinity, which in turn satisfies requirement (ii) in Theorem 2  of \citet{Dominitz:2005aa}.}
For completeness, this result is summarized in Appendix \ref{sec:consistency_contraction}.

More generally, if $\hat M$ is not guaranteed to be a contraction, one may use root-finding algorithms  that solve $\theta_{-1}-\hat M(\theta_{-1})=0$ up to approximation errors of $o(N^{-1/2})$.
The root-$N$ consistency of $\theta_{N,-1}$ then follows from the standard argument for extremum estimators, in which we take $\mathcal Q_N(\theta_{-1})=\|\theta_{-1}-\hat M(\theta_{-1})\|$ as a criterion function.\footnote{The key conditions for these results, uniform convergence (in probability) of $\hat K$ and its stochastic equicontinuity, are established in Lemma \ref{lem:sequi}.}
Since these results are standard, we omit details and focus below on the asymptotic distribution and bootstrap validity of the fixed point estimators. Our contributions are two-fold. First, we establish the asymptotic distribution of the fixed point estimator without assuming that $\hat M$ is an asymptotic contraction map, which therefore allows the practitioner to conduct inference using the estimator based on the general root-finding algorithm and complements the result of   \cite{Dominitz:2005aa}. Second, to our knowledge, the bootstrap validity of the fixed point estimators is new. These results are established by showing that, under regularity conditions, the population fixed point is  Hadamard-differentiable and hence admits the use of the functional $\delta$-method, which may be of independent theoretical interest.

\begin{remark}
To establish the asymptotic properties, one could try to reformulate our estimator as an estimator that approximately solves a GMM problem. Here, instead of relying on another reformulation, which would require establishing a sample analog version of Proposition \ref{prop:equivalence}, we develop and directly apply an asymptotic theory for fixed point estimators. We take this approach as the theory itself contains generic results (Theorem \ref{thm:asynorm} and Lemmas \ref{lem:H_hausdorff}--\ref{lem:H_dot}) surrounding the Hadamard-differentiability of fixed points, which allow for applying the functional $\delta$-method to obtain the asymptotic distribution of $\est$ and bootstrap validity.
These results can potentially be used to analyze decentralized estimators outside the IVQR class. \qed
\end{remark}

\subsection{Asymptotic Theory and Bootstrap Validity}

\label{ssec:asymptotics_bootstrap}
The following theorem gives the limiting distribution of our estimator. For each $w=(y,d',x',z')'$ and $\theta\in\Theta,$ let $f(w;\theta)\in \mathbb R^{d_X+d_D}$ be a vector whose sub-vectors are given by
\begin{align*}
f_1(w;\theta)&=(1\{ y\leq d'\theta_{-1}+x'\theta_1 \} -\tau)x,\\
f_j(w;\theta)&=(1\{ y\leq d'\theta_{-1}+x'\theta_1 \} -\tau)z_{j-1},\quad j=2,\dots,J,
\end{align*}
and let $g(w;\theta)=(g_1(w;\theta)',\dots,g_J(w;\theta))'$ be a vector such that
\begin{align}
	g_j(w;\theta)=\left(\frac{\partial^2}{\partial\theta_j\partial\theta_j'}Q_{P,j}(L_j(\theta_{-j}),\theta_{-j})\right)^{-1}f_j(w;L_j(\theta_{-j}),\theta_{-j}),~j=1,\dots,J.
\end{align}

\begin{theorem}\label{thm:asynorm}
Suppose that Assumptions \ref{ass:ivqr} and \ref{ass:global_decentralization} hold. Let $\left\{W_i\right\}_{i=1}^N$ be an i.i.d.\ sample generated from the IVQR model, where $W_i=(Y_i,D_i',X_i',Z_i')$.   Then,
	\begin{align}
	\sqrt N(\est-\true)	\stackrel{L}{\to}N(0,V)~,\label{eq:asynorm0a}
	\end{align}
	with
	\begin{align}
	V=(I_{d_X+d_D}-J_K(\true))^{-1}E[\mathbb W(\true)\mathbb W(\true)'][(I_{d_X+d_D}-J_K({\true}))^{-1}]',\label{eq:asyvar}
	\end{align}
	where $\mathbb W$ is a tight Gaussian process in $\ell^\infty(\Theta)^\dparam$ with mean zero and the covariance kernel
\begin{align}
	\emph{Cov}(\mathbb W(\theta),\mathbb W(\tilde\theta))=E_P\big[(g(W;\theta)- E_P[g(W;\theta)])(g(w;\tilde\theta)-E_P[g(w;\tilde\theta)])'\big].\label{eq:cov_kernel}
\end{align}
\end{theorem}

In the theorem above, the asymptotic variance of $\hat\theta_N$ is characterized in terms of the Jacobian of $K$ and a Gaussian process $\mathbb W$. This can be reformulated to show its asymptotic relationship to a GMM estimator. Let $\tilde\theta_N$ be an estimator that  solves the following estimating equations
\begin{align}
\frac{1}{N}\sum_{i=1}^N1\{Y_i\le X_i'\tilde\theta_{N,1}+D_{1,i}\tilde\theta_{N,2}+\dots+D_{d_{D},i}\tilde\theta_{N,J}\}-\tau)\begin{pmatrix}
	X_i\\
	Z_i
\end{pmatrix}=o_p(N^{-1/2}).\label{eq:est_eq_tilde}
\end{align}
Let $\Psi(\tau)=(X',Z')'.$ As shown in \cite{CH2006} (Theorem 3 and Remark 3), $\sqrt N(\tilde \theta_N-\true)$ converges weakly to a mean zero multivariate normal distribution with variance
\begin{align}
	\tilde V=\tau(1-\tau)J_{\Psi_P}(\theta^*)^{-1}E[\Psi(\tau)\Psi(\tau)'][J_{\Psi_P}(\theta^*)^{-1}]',\label{eq:Vtilde}
\end{align}
where $J_{\Psi_P}(\theta^*)=E\left[f_{\varepsilon(\tau)|X,D,Z}(0)\Psi(\tau) (X',D')\right],$ $\varepsilon(\tau)=Y-X'\theta^*_1-D'\theta^*_{-1}$, and $f_{\varepsilon(\tau)|X,D,Z}$ is $\varepsilon(\tau)$'s conditional density given $(X,D,Z)$. The following corollary shows that the fixed point estimator  $\hat\theta_N$ is asymptotically equivalent to $\tilde\theta_N$ in terms of its limiting distribution.
\begin{corollary}\label{cor:asy_equivalence}
Suppose the conditions of Theorem \ref{thm:asynorm} hold. Then,  $V=\tilde V$.
\end{corollary}

To conduct inference on $\true$,  one may employ a natural bootstrap procedure. For this, use in \eqref{eq:hattheta1}--\eqref{eq:hattheta2} the bootstrap sample instead of the original sample to define the bootstrap analogs $\hat M^*$ and $\hat\theta^*_{N,-1}$ of $\hat M$ and $\hat\theta_{N,-1}$.
 In practice, the bootstrap can be implemented using the following steps.

\begin{enumerate}
\item Compute the fixed point estimator $\est$ using the original sample.
\item Draw a bootstrap sample $\{W^*_i\}_{i=1}^N$ randomly with replacement from $P_N.$ Use the sequential  dynamical system based on $\hat M^*$ combined with a contraction or root-finding algorithm to compute $\est^*$.
\item Repeat Step 2 across bootstrap replications $b=1,\dots, B$. Let
\begin{align}
F_B(x)\equiv	\frac{1}{B}\sum_{b=1}^B1\Big\{\sqrt N(\est^{*,b}-\est)\le x\Big\},~x\in\mathbb R.
\end{align}
Use $F_B$ as an approximation to the sampling distribution of the root $\sqrt N(\est-\true)$.
\end{enumerate}
The bootstrap is particularly attractive in conjunction with our new and computationally efficient estimation algorithms. By contrast, directly bootstrapping for instance the IQR estimator of \citet{CH2006} is computationally very costly. Alternative methods (either an asymptotic approximation or a score-based bootstrap)  require estimation of the influence function, which involves nonparametric estimation of a certain conditional density. Directly bootstrapping our fixed point estimators avoids the use of any smoothing and tuning parameters.\footnote{The use of the bootstrap here is for consistently estimating the law of the estimator. Whether one may obtain higher-order refinements through a version of the bootstrap, e.g., the $m$ out of $n$ bootstrap with extrapolation \citep{Sakov:2000aa}, is an interesting question which we leave for future research.}

The following theorem establishes the consistency of the bootstrap procedure. For this, let $\stackrel{L^*}{\leadsto}$ denote the weak convergence of the bootstrap law in outer probability, conditional on the sample path $\{W_i\}_{i=1}^\infty.$
\begin{theorem}\label{thm:bootstrap}
Suppose that Assumptions \ref{ass:ivqr} and \ref{ass:global_decentralization} hold. Let $\left\{W_i\right\}_{i=1}^N$ be an i.i.d.\ sample generated from the IVQR model.
   Then,
\begin{align*}
\sqrt N(\est^{*}-\est) \stackrel{L^*}{\leadsto} N(0,V),
\end{align*}
where $V$ is as in \eqref{eq:asyvar}. 
\end{theorem}

\section{Empirical Example}
\label{sec:empirical_example}

In this section, we illustrate the proposed estimators by reanalyzing the effect of 401(k) plans on savings behavior as in \citet{CH2004}. This empirical example constitutes the basis for our Monte Carlo simulations in Section \ref{sec:simulation_study}. As explained by \citet{CH2004}, 401(k) plans are tax-deferred savings options that allow for deducting contributions from taxable income and accruing tax-free interest. These plans are provided by employers and were introduced in the United States in the early 1980s to increase individual savings. To estimate the effect of 401(k) plans ($D$) on accumulated assets ($Y$), one has to deal with the potential endogeneity of the actual participation status. \citet{CH2004} propose an instrumental variables approach to overcome this problem. They use 401(k) eligibility as an instrument ($Z$) for the participation in 401(k) plans. The argument behind this strategy, which is due to \citet{Poterbaetal1994,Poterbaetal1995,Poterbaetal1998} and \citet{Benjamin2003},  is that eligibility is exogenous after conditioning on income and other observable factors. We use the same identification strategy here but note that there are also papers which argue that 401(k) eligibility is not conditionally exogenous \citep[e.g.,][]{Engenetal1996}.

We use the same dataset as in \citet{CH2004}. The dataset contains information about 9913 observations from a sample of households from the 1991 Survey of Income and Program Participation.\footnote{The dataset analyzed by \citet{CH2004} has 9,915 observations. Here we delete the two observations with negative income.} We refer to \citet{CH2004} for more information about the data and to their Tables 1 and 2 for descriptive statistics. Here we focus on net financial assets as our outcome of interest.\footnote{\citet{CH2004} also consider total wealth and net non-financial assets.}

We consider the following linear model for the conditional potential outcome quantiles
\begin{eqnarray}
q(D,X,\tau)&=&X^{\prime }\theta_X(\tau)+D\theta_D(\tau).\label{eq:model_wo_interactions}
\end{eqnarray}
The vector of covariates $X$ includes seven dummies for income categories, five dummies for age categories, family size, four dummies for education categories, indicators for marital status, two-earner status, defined benefit pension status, individual retirement account participation status and homeownership, and a constant. Because $P(D=0)>0$, we re-parametrize the model by replacing $D$ by $D^{\star}=D+1$ to ensure that $Z/D^{\star}$ is well-defined and positive.

Below, we briefly describe the construction of the sequential response map for this application.
First, we allocate  $\theta_X$ to player 1 and hence denote this subvector by $\theta_1$. Similarly, we allocate $\theta_D$ to player 2 and denote it by $\theta_2$.  For each $\tau$, the sequential response map can be constructed by taking the following steps. 
\begin{enumerate}
	\item For any given $\theta_2$,  compute $\theta_1=\hat L_1(\theta_2)$ by running an ordinary QR (for the $\tau$-th quantile) using $Y_i-D_i^\star\theta_2$ as the dependent variable and $X_i$ as a regressor vector.
	\item Given $\theta_1=\hat L_1(\theta_2)$ from the previous step, compute $\hat L_2(\theta_1)$ by running an ordinary QR with weights $Z_i/D_i^\star$ using $Y_i-X_i'\theta_1$ as the dependent variable, $D_i^\star$ as the regressor, and omitting the constant.
\end{enumerate}
Combining these two steps yields the sequential response map $\hat M(\cdot)=\hat L_2(\hat L_1(\cdot))$.
Figure \ref{fig:fixed_point_graphical} shows, for each $\tau \in \{0.25,0.50,0.75\}$, the graph of $\theta_2\mapsto\hat{M}(\theta_2)$. For each $\tau$, the intersection between $\hat{M}$ and the 45-degree line (i.e.\ the identity map) is our fixed point estimator $\hat\theta_D(\tau)$. Figure \ref{fig:fixed_point_graphical} further provides a straightforward graphical way to check the validity of the sample analog of Assumption \ref{ass:contraction_global}. We can see that the sample analog of $J_M$ (i.e. the slope of $M$) is smaller than one at any $\theta_2$. This suggests that the contraction-based sequential algorithm converges at all three quantile levels, which is indeed what we find.

\begin{figure}[htbp]
\begin{center}
\caption{Illustration Fixed Point Estimator}
\includegraphics[width=0.7\textwidth,trim=0 2cm 0 2cm]{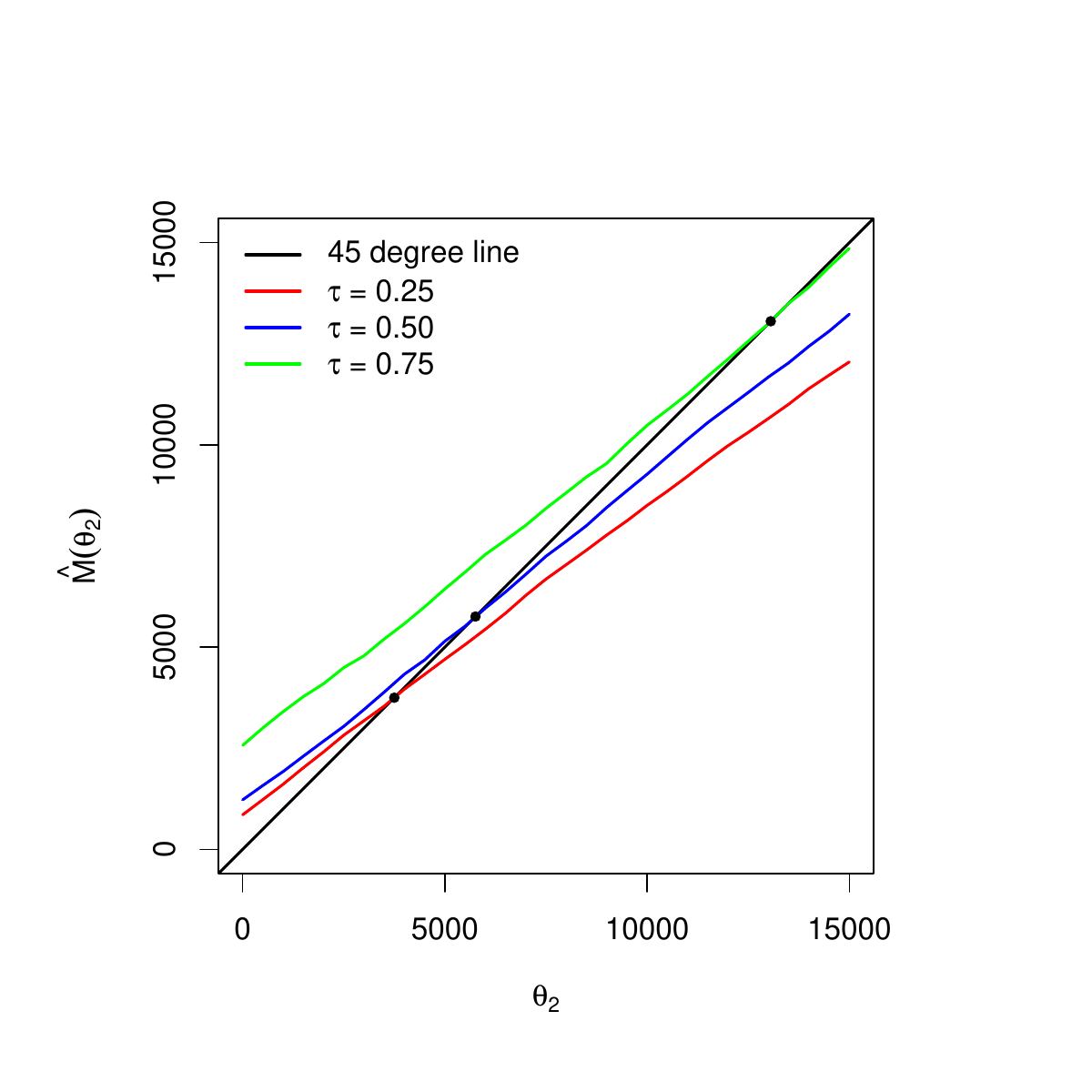}
\label{fig:fixed_point_graphical}
\end{center}
\vspace{0.2in}

\begin{minipage}{0.65\textwidth}
{\footnotesize Notes: The sequential response map $\theta_2\mapsto \hat M(\theta_2)$ is plotted for three quantile levels: $\tau=0.25$ (red), $\tau=0.5$ (blue), and $\tau=0.75$ (green). The intersection of each map and the 45-degree line (black) gives the fixed point estimator of $ \theta_{D}(\tau)$.\par}
\end{minipage}
\end{figure}

For each $\tau$, the steps for numerically calculating the fixed point estimator are as follows.
\begin{enumerate}\setcounter{enumi}{2}
	\item Find the fixed point of $\hat M$ using one of the following algorithms.
\begin{itemize}
	\item Contraction algorithm: Set the tolerance $e_N$.\footnote{In our application and simulations, we set  $e_N$ equal to the square-root of the machine precision in \texttt{R}: $e_N\approx 1.5\cdot 10^{-8}$.} We set the initial value $\param_2^{(0)}$ to the 2SLS estimate of the coefficient on $D^\star_i$. Iterate  $\param_{2}^{(s+1)}=\hat M(\param_{2}^{(s)}),$ for $s=0,1,2,\dots,$ until $|\param_{2}^{(s+1)}-\param_{2}^{(s)}|\le e_N$. Report $\param_{2}^{(s+1)}$ as the estimate  $\hat\theta_{D}$ of $\theta_D$.
	\item Root-finding algorithm: Find the solution to the equation $\theta_2-\hat M(\theta_2)=0$ by applying Brent's method as implemented by the \texttt{R}-package \texttt{uniroot}. Report the solution as the estimate $\hat\theta_{D}$ of $\theta_D$.
\end{itemize}
	\item (Optional) Compute $\hat\theta_{X}=\hat L_1(\hat\theta_{D})$ (as in Step (2)) if one is also interested in the coefficients on $X$.  Recover estimates of the original coefficients following the discussion in Appendix \ref{app_sec:reparametrization}.
\end{enumerate}

We compare our estimators to the IQR estimator of \citet{CH2006} based on a grid search over 500 points. IQR provides a very robust benchmark. However, due to the use of grid search, it is computationally expensive. In our empirical Monte Carlo study in Section \ref{sec:simulation_study}, we find that, with 10,000 observations and one endogenous regressor, IQR is nine times slower than the contraction algorithm and 23 times slower than the root-finding algorithm. With two endogenous regressors, the computational advantages of our procedures are even more pronounced. IQR based on a grid search over 100$\times$100 points is 149 times slower than the contraction algorithm and 84 times slower than the nested root-finding algorithm.

Figure \ref{fig:qte_401k} displays the estimates of $\theta_D(\tau)$ for $\tau \in \{0.15,0.20,\dots,0.85\}$. We can see that all estimation algorithms yield very similar results. We also note that the contraction algorithm converges for all quantile levels considered.
\begin{figure}[htbp]
\begin{center}
\caption{Comparison Point Estimates}
\includegraphics[width=0.7\textwidth, trim = 0 1cm 0 1cm]{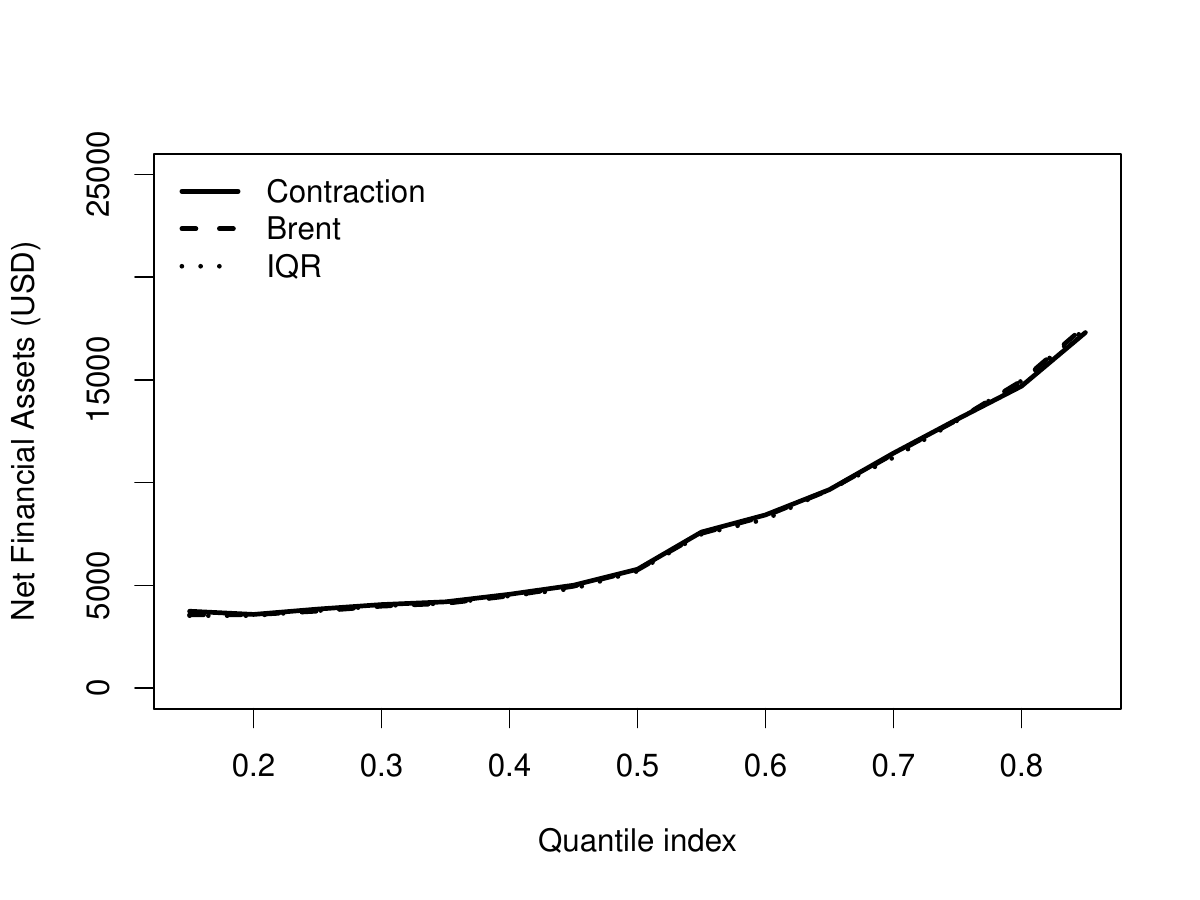}
\label{fig:qte_401k}
\end{center}
\end{figure}

Figures \ref{fig:qte_401k_ci} depicts pointwise 95\% confidence intervals for the proposed estimators obtained using the empirical bootstrap described in Section \ref{ssec:asymptotics_bootstrap}
with 500 replications. We can see that the resulting confidence intervals are very similar for both algorithms and do not include zero at all quantile levels considered.

\begin{figure}[htbp]
\begin{center}
\caption{Pointwise 95\% Bootstrap Confidence Intervals}
\includegraphics[width=0.7\textwidth, trim = 0 1cm 0 1cm]{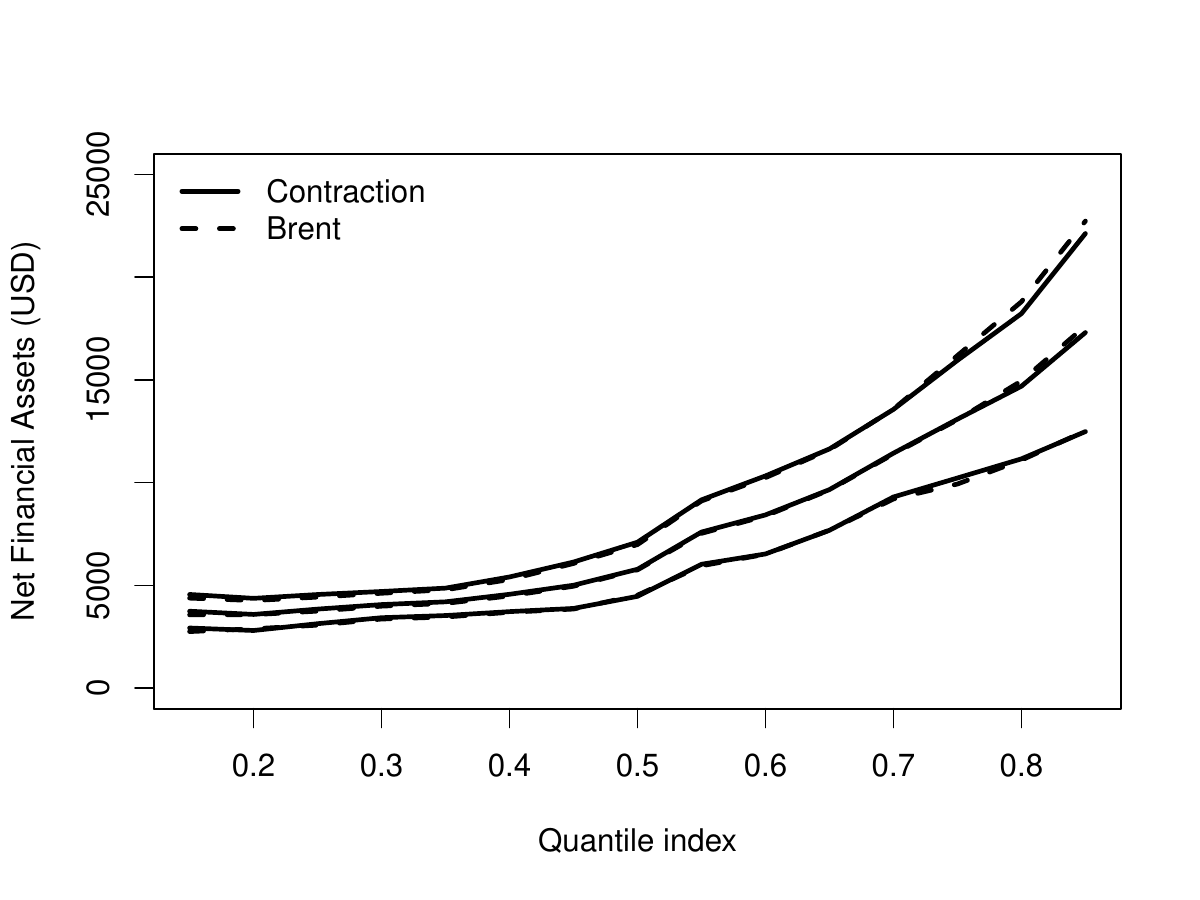}
\label{fig:qte_401k_ci}
\end{center}
\end{figure}

\section{Empirical Monte Carlo Study}
\label{sec:simulation_study}

In this section, we assess the practical performance of our estimation algorithms in an empirical Monte Carlo study based on the application in Section \ref{sec:empirical_example}. 

\subsection{An Application-based DGP}
\label{sec:application_based_dgp}
We consider DGPs which are based on the empirical application of Section \ref{sec:empirical_example}.\footnote{The construction of our DGPs is inspired by the application-based DGPs in \citet{KaplanSun2017}.}  We focus on a simplified setting with only two exogenous covariates: income and age. The covariates are drawn from their joint empirical distribution. The instrument $Z_i$ is generated as $\text{Bernoulli}\left(\bar{Z} \right)$, where $\bar{Z}$ is the mean of the instrument in the empirical application. We then generate the endogenous variable as $D_i=Z_i\cdot 1\left\{ 0.6\cdot  V_i < U_i\right\},$
where $U_i\sim \text{Uniform}(0,1)$ and $V_i\sim \text{Uniform}(0,1)$ are independent disturbances. The DGP for $D_i$ is chosen to roughly match the joint empirical distribution of $(D_i,Z_i)$. The outcome variable $Y_i$ is generated as
\begin{eqnarray}
 Y_i = X_i'\theta_X(U_i)+D_i\theta_{D}(U_i)+ G^{-1}(U_i).\label{eq:dgp_1endog}
\end{eqnarray}
The coefficient $\theta_X(\cdot)$ is constant and equal to the IQR median estimate in the empirical application. $\theta_D(U_i)=5000+U_i\cdot 10000$ is chosen to match the increasing shape of the estimated conditional quantile treatment effects in Figure \ref{fig:qte_401k}. $G^{-1}(\cdot)$ is the quantile function of a re-centered Gamma distribution, estimated to match the distribution of the IQR residuals at the median. To investigate the performance of our procedure with more than one endogenous variable, we add a second endogenous regressor:
\begin{eqnarray}
Y_i = X_i'\theta_X(U_i)+D_i\theta_{D}(U_i)+D_{2,i}\theta_{D,2}(U_i) + G^{-1}(U_i),
\end{eqnarray}
where we set $\theta_{D,2}(U_i)=10000$. The second endogenous variable is generated as
\[
D_{2,i}=0.8 \cdot Z_{2,i}+0.2 \cdot \Phi^{-1}(U_i)
\]
and the second instrument is generated as $Z_{2,i}\sim N(0,1)$. We set $N=9913$ as in the empirical application. 

\subsection{Estimation Algorithms}
We assess and compare several different algorithms, all of which are based on the dynamical system $\hat M$.\footnote{We do not explore  algorithms based on $\hat K$ because, as discussed in Section \ref{sec:sample_algorithms}, the algorithms based on $\hat M$ have advantages in terms of the choice of starting values (for the contraction algorithm) and dimension reduction (for the root-finding algorithm).}
For models with one endogenous variable, we consider a contraction algorithm and a root-finding algorithm based on Brent's method. For models with two endogenous variables, we consider a contraction algorithm, a nested root-finding algorithm based on Brent's method, and a root-finding algorithm implemented as a minimization problem based on simulated annealing (SA).\footnote{We have also explored algorithms based on Newton-Raphson-type root-finders. These algorithms are, in theory, up to an order of magnitude faster than the contraction algorithm and the nested algorithms, but, unlike the other algorithms considered here, require an approximation to the Jacobian and are not very robust to the choice of starting values. We therefore do not report the results here.}  For all estimators, we use 2SLS estimates as starting values. We compare the results of our algorithms to (nested) profiling estimators based on Brent's method, and to IQR with a grid search over 500 points (one endogenous regressor) and 100$\times$100 points (two endogenous regressors), which serves as a slow but very robust benchmark. Table \ref{tab:algorithms} presents an overview of the algorithms.

\begin{table}[ht]
\begin{scriptsize}
\centering
\caption{Algorithms}
  \begin{tabular}{l l }
\toprule
    \multicolumn{2}{c}{\textbf{One endogenous variable}}\\
\midrule
    Algorithm &\texttt{R}-Package \\
    \midrule
    Contraction algorithm &    \\
    Root-finding algorithm & \texttt{uniroot} \citep{uniroot2019} \\
     Profiling algorithm & \texttt{uniroot} \citep{uniroot2019} \\
    IQR &  \\
        \midrule
    \multicolumn{2}{c}{\textbf{Two endogenous variables}}\\
    \midrule
        Algorithm &\texttt{R}-Package \\

\midrule

    Contraction algorithm &  \\

    Nested root-finding algorithm & \texttt{uniroot} \citep{uniroot2019}  \\
    
        Root-finding algorithm  (implemented as optimizer) &  \texttt{optim\_sa} \citep{optimsa2018}\\
        Nested profiling algorithm & \texttt{uniroot} \citep{uniroot2019}  \\

     IQR (implementation: p.132 in&   \\
       \citet{Chernozhukov+17handbook})\\
    \bottomrule
	  \end{tabular}
\label{tab:algorithms}
\end{scriptsize}
\end{table}

\subsection{Results} Here we describe the computational performance of the different procedures and the finite sample performance of our bootstrap inference procedure. In Appendix \ref{app:bias_rmse}, we further investigate the finite sample bias and root mean squared error (RMSE) of the different methods. We find that the proposed estimation algorithms perform well and exhibit a similar bias and RMSE. Moreover, the finite sample properties of the estimators based on the contraction and root-finding algorithms are comparable to the profiling estimators and IQR. This shows that the computational advantages of our algorithms do not come at a cost in terms of the finite sample performance. Appendix \ref{app:sims_3} presents additional simulation evidence, demonstrating that our algorithms perform well and remain computationally tractable with more than two endogenous regressors.\footnote{Specifically, we present simulation results based on the DGP used in this section, augmented with an additional endogenous regressor, generated as $D_{3,i}=0.8 \cdot Z_{3,i}+0.2 \cdot \Phi^{-1}(U_i)$, where $Z_{3,i}\sim N(0,1)$.}

Tables \ref{tab:cpu1}--\ref{tab:cpu2} show the average computation time (in seconds) for estimating the model with one and two endogenous variables for different sample sizes. All computations were carried out on a standard desktop computer with a 3.2 GHz Intel Core i5 processor and 8GB RAM.

With one endogenous regressor, the contraction algorithm and the root-finding algorithm based on Brent's method  are computationally more efficient than IQR.  Specifically, the root-finding algorithm based on Brent's method is 11 to 23 times faster than IQR and the contraction algorithm is 4 to 9 times faster. The root-finding algorithm is as fast as the profiling estimator and about twice as fast as the contraction algorithm.\footnote{Note that the computational speed of the contraction algorithm depends on $|\hat J_M|$ and, thus, will differ across applications.}

The computational advantages of our algorithms become more pronounced with two endogenous variables. Table \ref{tab:cpu2} shows that IQR's average computation times are around two orders of magnitude slower than those of our preferred procedures. Specifically, the nested root-finding algorithm is 84 to 134 times faster than IQR, while the contraction algorithm is 149 to 308 times faster. This is as expected since, due to the use of grids, IQR's becomes computationally impractical in our setting whenever the number of endogenous variables exceeds two or three.\footnote{Our implementation of IQR with two endogenous variables is inherently slower than the implementation with one endogenous variable, even when the number of grid points is the same. First, there is an additional covariate in the underlying QRs (the second instrument). Second, with one endogenous variable, we choose the grid value that minimizes the absolute value of the coefficient on the instrument. By contrast, with two endogenous regressors, we choose the grid point which minimizes a quadratic form based on the inverse of the estimated QR variance covariance matrix as suggested in \citet{Chernozhukov+17handbook}, which requires an additional computational step.} The contraction algorithm is almost twice as fast as the nested algorithm, which, in turn, is almost twice as fast as the profiling estimator. Finally, the minimization algorithm based on SA is about one order of magnitude slower than the contraction algorithm, while still being an order of magnitude faster than IQR.

\begin{table}[ht]
\linespread{1.15}\parskip .05in
\begin{footnotesize}
\centering
\caption{Computation time, 401(k) DGP with one endogenous regressor}
\begin{tabular}{lcccc}
\toprule

$N$& Contr& Brent & Profil &InvQR \\
 \midrule
  \\
1000 & 0.10 & 0.03 & 0.03 & 0.34 \\ 
  5000 & 0.44 & 0.13 & 0.12 & 2.44 \\ 
  10000 & 0.75 & 0.28 & 0.30 & 6.48 \\ 
\\
  \bottomrule
\multicolumn{5}{p{6.2cm}}{\footnotesize {\it Notes:}  The table reports average computation time in seconds at $\tau=0.5$ over 100 simulation repetitions based on the DGP described in the main text. Contr: contraction algorithm; Brent: root-finding algorithm based on Brent's method; Profil: profiling estimator based on Brent's method; InvQR: inverse QR with grid search over 500 grid points. We use 2SLS estimates as starting values.}
\end{tabular}
\label{tab:cpu1}
\end{footnotesize}
\linespread{1.5}\parskip .05in
\end{table}

\begin{table}[ht]
\linespread{1.15}\parskip .05in
\begin{footnotesize}
\centering
\caption{Computation time, 401(k) DGP with two endogenous regressor}
\begin{tabular}{lccccc}
\toprule

$N$& Contr & NestBr & SimAnn  &Profil& InvQR \\
 \midrule
  \\
1000 & 0.18 & 0.41 & 3.13 & 0.77 & 55.40 \\ 
  5000 & 1.71 & 2.67 & 19.05 & 4.39 & 282.64 \\ 
  10000 & 4.50 & 7.99 & 54.95 & 12.50 & 672.00 \\ 
\\
  \bottomrule
\multicolumn{6}{p{8cm}}{\footnotesize {\it Notes:} The table reports average computation time in seconds at $\tau=0.5$ over 100 simulation repetitions based on the DGP described in the main text.  Contr: contraction algorithm; NestBr: nested algorithm based on Brent's method; SimAnn: simulated annealing based optimization algorithm; Profil: nested profiling estimator based on Brent's method; InvQR: inverse QR with grid search over 100$\times 100$ grid points. We use 2SLS estimates as starting values.}
\end{tabular}
\label{tab:cpu2}
\end{footnotesize}
\linespread{1.5}\parskip .05in
\end{table}

Finally, we analyze the finite sample properties of our bootstrap inference procedure for making inference on $\theta_D$ in the model with a single endogenous variable. Table \ref{tab:size_401k} shows the empirical coverage probabilities of bootstrap confidence intervals based on the contraction algorithm and the root-finding algorithm based on Brent's method. Both methods demonstrate an excellent performance and exhibit coverage rates that are very close to the respective nominal levels.

\begin{table}[ht]
\linespread{1.15}\parskip .05in
\begin{small}
\centering
\caption{Coverage, 401(k) DGP with one endogenous regressor}
\begin{tabular}{lcccc}
\toprule

& \multicolumn{2}{c}{$1-\alpha=0.95$} & \multicolumn{2}{c}{$1-\alpha=0.9$}\\
\cmidrule(l{5pt}r{5pt}){2-3}  \cmidrule(l{5pt}r{5pt}){4-5} \

$\tau$& Contr & Brent&  Contr & Brent \\
 \midrule
  \\
0.15 & 0.96 & 0.96 & 0.93 & 0.91 \\ 
  0.25 & 0.95 & 0.95 & 0.91 & 0.92 \\ 
  0.50 & 0.95 & 0.95 & 0.90 & 0.91 \\ 
  0.75 & 0.95 & 0.95 & 0.90 & 0.89 \\ 
  0.85 & 0.94 & 0.94 & 0.89 & 0.89 \\ 
\\
  \bottomrule
\multicolumn{5}{p{6cm}}{\small {\it Notes:} Monte Carlo simulation with 1000 repetitions as described in the main text. Contr: contraction algorithm; Brent: root-finding algorithm based on Brent's method. We use 2SLS estimates as starting values.}
\end{tabular}
\label{tab:size_401k}
\end{small}
\linespread{1.5}\parskip .05in
\end{table}

\section{Conclusion}
\label{sec:conclusion}

The main contribution of this paper is to develop computationally convenient and easy-to-implement estimation algorithms for IVQR models. Our key insight is that the non-smooth and non-convex IVQR estimation problem can be decomposed into a sequence of much more tractable convex QR problems, which can be solved very quickly using well-established methods. The proposed algorithms are particularly well-suited if the number of exogenous variables is large and the number of endogenous variables is moderate as in many empirical applications.

An interesting avenue for further research is to investigate weak identification robust inference within the decentralized model.  One may, for example, write the (re-scaled) sample fixed point restriction as $\sqrt N(I_\dparam-\hat K)(\theta)=s_N(\theta)+\mathbb W(\theta)+r_N(\theta)$, where $s_N(\theta)=\sqrt N(I_\dparam-K)(\theta),$ $\mathbb W$ is a Gaussian process, and $r_N$ is an error that tends to zero uniformly. This paper assumes that  $s_N(\true)=0$ uniquely, and outside  $N^{-1/2}$-neighborhoods of $\true$, $s_N(\theta)$ diverges and dominates $\mathbb W.$  For a one-dimensional fixed point problem,  this requires the BR map to be bounded away from the 45-degree line outside any $N^{-1/2}$-neighborhood of the fixed point. However if $s_N$ fails to dominate $\mathbb W$ over a substantial part of the parameter space, one would end up with weak identification.\footnote{\cite{AndrewsMikusheva2016} study weak identification robust inference methods in models characterized by moment restrictions.} How to conduct robust inference in such settings is an interesting question, which we leave for future research.

Finally, we note that, while we study the performance of the proposed algorithms separately, our reformulation and the resulting algorithms are potentially very useful when combined with other existing procedures. For instance, one could choose starting values using an initial grid search over a coarse grid and then apply a fast contraction algorithm.

\newpage

\appendix

\clearpage

\small

\setcounter{page}{1}

\thispagestyle{plain}
\begin{center}
{\textbf{APPENDIX (FOR ONLINE PUBLICATION)}}\\
\vspace{0.5cm}

{\sc Hiroaki Kaido and Kaspar W\"uthrich}\\

\sc \today
\end{center}

\section*{Table of Contents}

\startcontents[sections]
\printcontents[sections]{l}{1}{\setcounter{tocdepth}{1}}

\section{Overidentification and Efficiency}
\label{app:overid}

In the main text, we focus on just-identified moment restrictions with $d_Z=d_D$, for which the construction of an estimator is straightforward. If the model is overidentified (i.e.\ if $d_Z>d_D$), we can transform the original moment conditions
\begin{equation*}
E_P\left[ \left( 1\left\{ Y\leq X^\prime \param_1(\tau)+D_1 \param_2(\tau)+\dots+D_{d_D} \param_J(\tau)\right\} -\tau\right)\begin{pmatrix}X\\Z\end{pmatrix} \right]=0
\end{equation*}
into a set of just-identified moment conditions
\begin{equation}
E_P\left[ \left( 1\left\{ Y\leq X^\prime \param_1(\tau)+D_1 \param_2(\tau)+\dots+D_{d_D} \param_J(\tau)\right\} -\tau\right)\begin{pmatrix}X\\\tilde{Z}\end{pmatrix}\right] =0, \label{eq:just_id}
\end{equation}
where $\tilde{Z}$ is a $d_D\times 1$ vector of transformations of $(X,Z)$. A practical choice is to construct $\tilde{Z}$ using a least squares projection of $D$ on $Z$ and $X$.

To achieve pointwise (in $\tau$) efficiency, we can employ the following two-step procedure which is based on Remark 5 in \citet[][]{CH2006}:

\textbf{Step 1:} We first obtain an initial consistent estimate of $\true$ using one of our estimators based on a set of just-identified moment conditions such as \eqref{eq:just_id}. We then use nonparametric estimators to estimate the conditional densities $V(\tau)=f_{\varepsilon(\tau)|X,Z}(0)$ and $v(\tau)=f_{\varepsilon(\tau)|D,X,Z}(0)$, where $\varepsilon(\tau)=Y-X'\true_1(\tau)-D_1 \true_2(\tau)-\dots-D_{d_D} \true_J(\tau)$, and the conditional expectation function $E_P\left[D v(\tau) \mid X,Z\right]$.

\textbf{Step 2:}
We apply our procedure to obtain a solution to the following moment conditions:
\begin{equation}
E_P\left[ \left( 1\left\{ Y\leq X^\prime \param_1(\tau)+D_1 \param_2(\tau)+\dots+D_{d_D} \param_J(\tau)\right\} -\tau\right)\begin{pmatrix}V(\tau)X\\E_P\left[D v(\tau) \mid X,Z\right]\end{pmatrix}\right] =0.\label{eq:eff_equations}
\end{equation}
Consider players $j=1,\dots,J$ solving the following optimization problems:
\begin{align}
&\min_{\tilde\theta_1 \in \mathbb{R}^{d_X}}Q_{P,1}\left( \tilde\theta_1,\theta_{-1}\right)\label{eq:ivqr_br1}\\
&\min_{\tilde\theta_{j} \in \mathbb{R}}Q_{P,j}\left(\tilde\theta_{j} ,\theta_{-j}\right),~ j=2,\dots,J\label{eq:ivqr_br2},
\end{align}
where
\begin{align*}
Q_{P,1}\left(\theta(\tau)\right)&:= E_P\left[\rho_\tau\left( Y-X'\theta_1(\tau)-D_1\theta_2(\tau)-\dots-D_{d_D}\theta_J(\tau)\right)V(\tau)\right],\\
Q_{P,j}\left(\theta(\tau)\right)&:= E_P\left[\rho_\tau\left( Y-X'\theta_1(\tau)-D_1\theta_2(\tau)-\dots-D_{d_D}\theta_J(\tau)\right)\frac{E_P\left[Dv(\tau)  \mid X,Z\right]_{j-1}}{D_{j-1}}\right], ~j=2,\dots,J,
\end{align*}
and $E_P\left[D v(\tau) \mid X,Z\right]_{j-1}$ is the $j$-th element of $E_P\left[D v(\tau) \mid X,Z\right]$. For each $j,$ the BR function $L_j(\theta_{-j}(\tau))$, defined as a member of the set of minimizers of $Q_{P,j}(\cdot,\theta_{-j})$, solves a suitable subset of the moment conditions in \eqref{eq:eff_equations}.
The optimization problems in \eqref{eq:ivqr_br1}-\eqref{eq:ivqr_br2} are convex population QR problems provided that the model is parametrized such that $E_P\left[D v(\tau) \mid X,Z\right]_{j-1}/D_{j-1}$, $j=2,\dots,J$, is positive. Estimation can then proceed by replacing the population QR problems by their sample analogues and applying one of the estimation algorithms discussed in the main text.  By Corollary \ref{cor:asy_equivalence}, the resulting estimator is asymptotically equivalent to a GMM estimator that uses the optimal instrumental variables  and thus achieves pointwise (in $\tau$) efficiency \citep[e.g.,][]{Chamberlain1987}.\footnote{Corollary \ref{cor:asy_equivalence} can be applied to the current setting by replacing the original set of covariates and instruments $\Psi(\tau)=(X',Z')'$ in \eqref{eq:Vtilde} with the optimal instrumental variables $\Psi(\tau)=(V(\tau)',E_P[D_1v(\tau)  \mid X,Z]/D_1,\dots, E_P[D_{d_D}v(\tau)  \mid X,Z]/D_{d_D})'$.}

\section{Reparametrization}
\label{app_sec:reparametrization}

In the main text, we assume that the model is reparametrized such that $Z_\ell/D_\ell$ is positive for all $\ell=1,\dots,d_D$. This ensures that the weights are well-defined and that the weighted QR problems are convex. However, in empirical applications, the weights may not be well-defined (e.g., if $D_\ell$ is an indicator variable with $P(D_\ell=0)>0$) or negative in some instances. Assuming that $Z_\ell$ is positive, a simple way to reparametrize the model is to add a large enough constant $c$ to $D_\ell$.\footnote{Since the IVQR model is characterized by conditional moments (as in \eqref{eq:cond_mr}), one may choose transformations of instruments to generate unconditional moment conditions. In case $Z_\ell$ is not positive $a.s.$, one can use a positive transformation (e.g., a logistic function) of $Z_\ell$ instead of $Z_\ell$ itself. The decentralization and identification results then hold with the transformed instruments as long as they satisfy our assumptions.} This transformation is theoretically justified by the compactness of the support of $D_\ell$ (Assumption \ref{ass:global_decentralization}.\ref{ass:global_decentralization2}). To fix ideas, suppose that one is interested in estimating the following linear-in-parameters model with a single endogenous variable:
\begin{align*}
q(D,X,\tau)= \theta_{11}+\tilde{X}^{\prime}\theta_{12}+D\theta_2,
\end{align*}
where $\theta_1=(\theta_{11},\theta_{12}')'$ and $X=\left(1,\tilde{X}^{\prime}\right)^{\prime}$. Suppose further that the support of $D$ is a compact interval, $\left[d_{\min},d_{\max}\right]\subset \mathbb{R}$, with $d_{\min}<0$. In this case, we can apply the transformation $D^{\star}=D+c$, where $c>|d_{\min}|$. The transformed model reads
\begin{align*}
q(D,X,\tau)= \theta^\star_{11}+\tilde{X}^{\prime}\theta_{12}+D^\star\theta_2,
\end{align*}
where $\theta^{\star}_{11}=\theta_{11}-c\theta_2$. Importantly, one can always back out the original parameters, $\theta=(\theta_{11},\theta_{12}',\theta_2)'$, from the parameters in the reparametrized model,
$\theta^\star=(\theta_{11}^\star,\theta_{12}',\theta_2)'$.

\section{Decentralization}\label{app_sec:decentralization}
\subsection{The Domains of $M_j$-Maps}
Recall that, in \eqref{eq:tildeR1}, we defined the set
\begin{align*}
	\tilde R_1\equiv \big\{\theta_{-1}\in\Theta_{-1}:&\Z_{P,1}(\theta_1,\theta_{-1})=0,\\
	&\Z_{P,2}(\theta_1,\theta_2,\pi_{-\{1,2\}}\theta_{-1})=0,~\exists (\theta_1,\theta_2)\in\Theta_1\times\Theta_2\big\}.
\end{align*}
Similarly, for $k=2,\dots,d_D-1$, define
\begin{align*}
	\tilde R_k\equiv \big\{\theta_{-1}\in\Theta_{-1}:&\Z_{P,1}(\theta_1,\theta_{-1})=0,\\
	&\Z_{P,2}(\theta_1,\theta_2,\pi_{-\{1,2\}}\theta_{-1})=0,\\
	&\vdots\\
	&\Z_{P,k}(\theta_1,\dots,\theta_k,\pi_{-\{1,\dots,k\}}\theta_{-1})=0,~\exists (\theta_1,\dots,\theta_k)\in\prod_{j=1}^k \Theta_j\big\}.
\end{align*}
For $k=d_D$, let
\begin{align*}
	\tilde R_{d_D}\equiv \big\{\theta_{-1}\in\Theta_{-1}:&\Z_{P,1}(\theta_1,\theta_{-1})=0,\\
	&\Z_{P,2}(\theta_1,\theta_2,\pi_{-\{1,2\}}\theta_{-1})=0,\\
	&\vdots\\
	&\Z_{P,J}(\theta_1,\dots,\theta_J)=0,~\exists (\theta_1,\dots,\theta_J)\in\prod_{j=1}^J \Theta_j\big\}.
\end{align*}
For each $k$, the map $M_k$ is well-defined on $\tilde R_{k}$.
Note also that $\tilde R_{d_D}\subset \tilde R_j$ for all $j\le d_D$.

\subsection{Local Decentralization and Local Contractions}\label{app_sec:local_decentralization}
 We say that an estimation problem admits \emph{local decentralization} if the BR functions $L_j$, $j=1,\dots,J$, and the maps $K$ and $M$ are well-defined over a local neighborhood of $\true$. The following weak conditions are sufficient for local decentralization of the IVQR estimation problem.

\begin{assumption}
\label{ass:local_decentralization}
The following conditions hold.
\begin{enumerate}
	\item The conditional cdf $y\mapsto F_{Y|D,X,Z}(y)$ is continuously differentiable at $y^*=d'\true_{-1}+x'\true_1$ for almost all $(d,x,z)$. The conditional density $f_{Y|D,Z,X}$ is   bounded  on a neighborhood of $y^*$ $a.s.$; \label{ass:local_decentralization3}
	\item The matrices
	\[
	E_P[f_{Y|D,X,Z}\left(D'\true_{-1}+X^\prime \true_1\right)  X X']
	\]
	and
	\[
	E_P[f_{Y|D,X,Z}\left(D'\true_{-1} +X^\prime \true_1\right)D_{\ell}Z_{\ell} ],\quad \ell=1,\dots,d_D,
	\]
	 are positive definite. \label{ass:local_decentralization4}
\end{enumerate}
\end{assumption}

Assumption \ref{ass:local_decentralization} is weaker than Assumption \ref{ass:global_decentralization}.\ref{ass:global_decentralization3}--\ref{ass:global_decentralization}.\ref{ass:global_decentralization4} as it only requires the conditions, e.g., continuous differentiability of the conditional CDF, at a particular point, e.g. $y^*$.
Under this condition, we can study the local properties of our population algorithms. For this, the following lemma ensures that the BR maps are well-defined locally.

\begin{lemma}
\label{lemma:local_decentralization}
Suppose that Assumptions \ref{ass:ivqr}, \ref{ass:global_decentralization}.\ref{ass:global_decentralization1}--\ref{ass:global_decentralization}.\ref{ass:global_decentralization2}, and \ref{ass:local_decentralization} hold. Then, there exist open neighborhoods $\mathcal{N}_{L_{-j}},j= 1,\dots, J$, $\mathcal{N}_K$, $\mathcal{N}_M$ of $\true_{-j}$, $\true$, and $\true_{-1}$ such that
\begin{enumerate}[label=(\roman*)]
	\item There exist maps $L_j:\mathcal{N}_{-j}\to \mathbb{R}^{d_j}$, $j=1,\dots,J$ such that, for $j=1,\dots,J$,
\begin{eqnarray*}
\Z_{P,j}\left(\br_j(\param_{-j}),\param_{-j}\right)=0,\quad \text{for all}~\param_{-j}\in \mathcal{N}_{-j}
\end{eqnarray*}
Further, $L_j$ is continuously differentiable for all $j=1,\dots,J$.
\item The maps $K:\mathcal{N}_K\to \mathbb{R}^{\dparam}$ and $M:\mathcal{N}_M\to \mathbb{R}^{d_{D}}$ are continuously differentiable.
\end{enumerate}
\end{lemma}
\begin{proof}
  (i) The proof is similar to that of Lemma \ref{lem:global_decentralization} (see Appendix \ref{sec:appendix_proof_sec3}). Therefore, we sketch the argument below for $j=1$.
  By Assumptions \ref{ass:global_decentralization}.\ref{ass:global_decentralization2} and \ref{ass:local_decentralization}.\ref{ass:local_decentralization3}, $\Z_{P,1}$ is continuously differentiable on a neighborhood $V$ of $\true$.   By Assumption  \ref{ass:local_decentralization}.\ref{ass:local_decentralization4} and the continuity of $\det(\partial\Z_{P,1}(\theta)/\partial\theta_1')$, one may choose $V$ so that $\det(\partial\Z_{P,1}(\theta)/\partial\theta_1')\ne 0$ for all $\theta=(\theta_1,\theta_{-1})\in V$.
  By the implicit function theorem, there is a continuously differentiable function $L_1$ and an open set $\mathcal N_{-1}$ containing $\theta_{-1}$ such that
  \begin{align*}
  \Z_{P,1}(L_1(\theta_{-1}),\theta_{-1})=0,\text{ for all }\theta_{-1}\in \mathcal N_{-1}.
  \end{align*}
  The arguments for $L_j$, $j\ne 1$ are similar. 
  
  (ii) Let $\mathcal N_K=\{\theta\in\Theta: \pi_{-j}\theta\in \mathcal N_{-j},~ j=1,\dots,J\}$ and let $\mathcal N_M$ be defined by mimicking \eqref{eq:tildeR1}, while replacing $\Theta_j$ with $\mathcal N_j$ in the definition of $\tilde R_j$ for $j=1,\dots,J$. The continuous differentiability of $K$ and $M$ follows from that of $L_j$, $j=1,\dots,J$.
\end{proof}

\subsubsection{Local Contractions}
\label{app:local_contractions}
Recall that  $\rho(A)$ denotes the spectral radius of a square matrix $A$.
The following assumption ensures that $K$ and $M$ are local contractions.
\begin{assumption} \label{ass:contraction_local}
\text{ }
\begin{enumerate}
	\item	 \label{ass:contraction_local_K}
$\rho\left(J_K\left(\true\right) \right)<1$
	\item \label{ass:contraction_local_M}
$\rho\left(J_M\left(\true_2 \right) \right)<1$\
\end{enumerate}
\end{assumption}
Here, we illustrate a primitive condition for Assumption \ref{ass:contraction_local}. Consider a simple setup without covariates (i.e. $X=1$), a binary $D$, and a binary $Z$. We only analyze Assumption \ref{ass:contraction_local}.\ref{ass:contraction_local_K}. A similar result can be derived for Assumption \ref{ass:contraction_local}.\ref{ass:contraction_local_M}. In this setting, the Jacobian of $K$ evaluated at $\true$ is given by
\begin{eqnarray*}
J_K(\true)=\begin{pmatrix}0 &-\frac{E_P\left[f_{Y|D,Z}\left(D \true_2 +\true_1\right) D \right]}{E_P\left[f_{Y|D,Z}\left(D \true_2 + \true_1\right)\right]}\\ -\frac{E_P\left[f_{Y|D,Z}\left(D \true_2 +\true_1\right) Z \right]}{E_P\left[f_{Y|D,Z}\left(D \true_2 +\true_1\right) ZD \right]}&0 \end{pmatrix}.
\end{eqnarray*}
The characteristic polynomial is then given by
\begin{eqnarray*}
p_K(\lambda)&=&\lambda^2-\frac{E_P\left[f_{Y|D,Z}\left(D \true_2 +\true_1\right) D\right]}{E_P\left[f_{Y|D,Z}\left(D \true_2 + \true_1\right)\right]}\frac{E_P\left[f_{Y|D,Z}\left(D \true_2 +\true_1\right) Z \right]}{E_P\left[f_{Y|D,Z}\left(D \true_2 +\true_1\right) ZD \right]}.
\end{eqnarray*}
Hence, Assumption \ref{ass:contraction_global}.\ref{ass:contraction_global_K} holds if all eigenvalues (i.e.\ the roots $\lambda_K$ of $p_K(\lambda)=0$) have modulus less than one, which holds when
\begin{eqnarray*}
\bigg|\frac{E_P\left[f_{Y|D,Z}\left(D \true_2 +\true_1\right) D \right]}{E_P\left[f_{Y|D,Z}\left(D \true_2 + \true_1\right)\right]}\frac{E_P\left[f_{Y|D,Z}\left(D \true_2 +\true_1\right) Z \right]}{E_P\left[f_{Y|D,Z}\left(D \true_2 +\true_1\right) ZD \right]}\bigg|<1 .
\end{eqnarray*}
This condition can be simplified to
\begin{eqnarray}
f_{Y|0,1}(\true_1)p(0|1)f_{Y|1,0}(\true_2+\true_1)p(1|0)< f_{Y|1,1}(\true_2+\true_1)p(1|1)f_{Y|0,0}(\true_1)p(0|0) \label{eq:contraction_K_late},
\end{eqnarray}
where $f_{Y|d,z}(y):=f_{Y|D=d,Z=z}(y)$ and $p(d|z):=P(D=d\mid Z=z)$.
It is instructive to interpret condition \eqref{eq:contraction_K_late} under the local average treatment effects framework of \citet{AngristImbens1994}. Specifically, condition \eqref{eq:contraction_K_late} holds if (i) their monotonicity assumption is such that there are compliers but no defiers and (ii) the complier potential outcome density functions are strictly positive. Conversely, the condition is violated if there are defiers but no compliers.

Under the local contraction conditions in Assumption \ref{ass:contraction_local}, we have the following results.

\begin{proposition}\label{prop:local_contraction} Suppose that Assumptions \ref{ass:ivqr}, \ref{ass:global_decentralization}.\ref{ass:global_decentralization1}, \ref{ass:global_decentralization}.\ref{ass:global_decentralization2}, and \ref{ass:local_decentralization} hold.
\begin{enumerate}[label=(\roman*)]
\item Suppose further that Assumption \ref{ass:contraction_local}.\ref{ass:contraction_local_K}  holds. Then there exists a closed neighborhood $\bar{\mathcal{N}}_K$ of $\true$ such that $K(\bar{\mathcal{N}}_K)\subset \bar{\mathcal{N}}_K$ and $K$ is a contraction on $\bar{\mathcal{N}}_K$ with respect to an adapted norm.
\item Suppose further that Assumption \ref{ass:contraction_local}.\ref{ass:contraction_local_M}  holds. Then there exists a closed neighborhood $\bar{\mathcal{N}}_M$ of $\true_2$ such that $M(\bar{\mathcal{N}}_M)\subset \bar{\mathcal{N}}_M$ and $M$ is a contraction on $\bar{\mathcal{N}}_M$ with respect to an adapted norm.
\end{enumerate}
\end{proposition}
\begin{proof}
We only prove the result for $K$, the proof for $M$ is similar. By Lemma \ref{lemma:local_decentralization}, $L_j$ is continuously differentiable at $\true$. Note that $J_K$ is given by
\begin{align}
  J_K(\theta)=\begin{bmatrix}
    0&\frac{\partial L_1(\theta_{-1})}{\partial \theta_2'}&\dots&\dots&\frac{\partial L_1(\theta_{-1})}{\partial \theta_J'}\\
    \frac{\partial L_2(\theta_{-2})}{\partial \theta_1'}&0&\frac{\partial L_2(\theta_{-2})}{\partial \theta_3'}&\dots&\frac{\partial L_2(\theta_{-2})}{\partial \theta_J'}\\
    \vdots&\vdots& \vdots &\vdots & \vdots\\
     \frac{\partial L_J(\theta_{-J})}{\partial \theta_1'} &\cdots & \cdots & \frac{\partial L_J(\theta_{-J})}{\partial \theta_{J-1}'}&0
\end{bmatrix},
\end{align}
which is continuous at $\true$. The desired result now follows, for instance, from Proposition 2.2.19 in \citet{HasselblattKatok2003}.
\end{proof}

\subsection{Nested Algorithms: Existence and Uniqueness of Fixed Points in Sub-games}\label{app:nesting}
Here we discuss two different sets of conditions which ensure that the nested fixed point algorithms in Section \ref{sec:root_finder} are well-defined. Specifically, we present conditions for the existence and uniqueness of fixed points in the sub-games. Section \ref{app:nested_M} considers contraction-based identification conditions. In Section \ref{app:nested_global}, we briefly discuss global identification conditions. To illustrate, we consider the case with three players ($J=3$). 

\subsubsection{Contraction-based Identification.}\label{app:nested_M} Suppose that Assumptions \ref{ass:ivqr}--\ref{ass:contraction_global} hold. Consider a sub-game formed by players 1 and 2 given some $\tilde\theta_3$. We assume that $\tilde\theta_3$ is chosen so that $\mathsf{D}_{M_{1,2|3}}$ (defined below) is non-empty.  Let the moment conditions for the sub-game be defined as
	\begin{align*}
	\mathsf{\Z}_{P,1}(\theta_1,\theta_2)\equiv\Z_{P,1}(\theta_1,\theta_{2},\tilde\theta_3)&= E_P\left[ \left( 1\left\{ Y\leq X^\prime \theta_1
	+D_1\param_{2}+D_2\tilde\theta_3\right\} -\tau\right)X \right], \\
	\mathsf{\Z}_{P,2}(\theta_1,\theta_2)\equiv\Z_{P,2}(\theta_1,\theta_{2},\tilde\theta_3)&= E_P\left[ \left( 1\left\{ Y\leq X^\prime \theta_1
	+D_1\param_{2}+D_2\tilde\theta_3
	\right\} -\tau\right)Z_{1} \right].
	\end{align*}
Note that $\mathsf{\Z}_{P,1}$ and $\mathsf{\Z}_{P,2}$ and other objects below depend on $\tilde\theta_3$. We will often suppress this dependence to alleviate the notation. Moreover, define
	\begin{align*}
		\mathsf{R}_1&:=\{\theta_1\in\Theta_1:\mathsf{\Psi}_{P,2}(\theta_1,\theta_2)=0,\text{ for some }\theta_2\in\Theta_2\},\\
		\mathsf{R}_2&:=\{\theta_2\in\Theta_2:\mathsf{\Psi}_{P,1}(\theta_1,\theta_2)=0,\text{ for some }\theta_1\in\Theta_1\}.
	\end{align*}
Assumptions \ref{ass:ivqr}--\ref{ass:global_decentralization} and Lemma \ref{lem:global_decentralization} guarantee that BR functions $\mathsf{L}_1$ and $\mathsf{L}_2$,  where
	\begin{align}
		&\mathsf{\Psi}_{P,1}(\mathsf{L}_1(\theta_2),\theta_2)=0,\quad \text{for all } \theta_{2}\in \mathsf{R}_2,\label{eq:subgameL1}\\
		&\mathsf{\Psi}_{P,2}(\theta_1,\mathsf{L}_2(\theta_1))=0,\quad \text{for all } \theta_{1}\in \mathsf{R}_1,\label{eq:subgameL2}
	\end{align}
are well-defined. 
The $M$ map for the sub-system is 
	\begin{align}
	M_{1,2|3}(\theta_2\mid \tilde\theta_3)=\mathsf{L}_2(\mathsf{L}_1(\theta_2)),
	\end{align}
This map exists and is well-defined on 
\begin{align*}
	\mathsf{D}_{M_{1,2|3}}:=\{\theta_2\in\Theta_2:&\mathsf{\Psi}_{P,1}(\theta_1,\theta_2)=0,\mathsf{\Psi}_{P,2}(\theta_1,\tilde\theta_2)=0,\text{ for some }(\theta_1,\tilde\theta_2)\in\Theta_1\times \Theta_2,(\theta_2,\tilde\theta_3)\in \tilde{D}_M\}.
\end{align*}
Note that, by Assumption \ref{ass:contraction_global}, for any  $(\theta_2,\theta_3)\in \tilde{D}_M$,
\[
 \|J_M(\theta_2,\theta_3)\|\le \lambda<1.
\]

Given these definitions, we can now investigate the sub-system.
Observe that the derivative $J_{M_{1,2|3}}(\theta_2)$ of $M_{1,2|3}$ with respect to $\theta_2$  is a component of $J_M.$ In particular, for any $(\theta_2,\theta_3)$, $J_M$ may be written as
\begin{align}
J_M(\theta_2,\theta_3)=\begin{bmatrix}
	\frac{\partial M_1(\theta_2,\theta_3)}{\partial\theta_2}&\frac{\partial M_1(\theta_2,\theta_3)}{\partial\theta_3}\\
	\frac{\partial M_2(\theta_2,\theta_3)}{\partial\theta_2}&\frac{\partial M_2(\theta_2,\theta_3)}{\partial\theta_3}
\end{bmatrix}	=
\begin{bmatrix}
	J_{M_{1,2|3}}(\theta_2)&\frac{\partial M_1(\theta_2,\theta_3)}{\partial\theta_3}\\
	\frac{\partial M_2(\theta_2,\theta_3)}{\partial\theta_2}&\frac{\partial M_2(\theta_2,\theta_3)}{\partial\theta_3}
\end{bmatrix}.
\end{align}
Let $V_{12}\equiv \{x\in \mathbb R^{2}:x=(x_1,0)',x_1\in\mathbb R\}$. Then, by the definition of the operator norm \citep[see e.g.][]{bhatia2013matrix},
\begin{align*}
	\|J_M(\theta_2,\theta_3)\|&\equiv \sup_{x,y\in \mathbb R^{2}:\|x\|=\|y\|=1}|x'J_M(\theta_2,\theta_3)y|\\
	&\ge \sup_{x,y\in V_{12}:\|x\|=\|y\|=1}|x'J_M(\theta_2,\theta_2)y|\\
	&=\sup_{\|x_1\|=\|y_1\|=1}|x_1'J_{M_{1,2|3}}(\theta_2)y_1|=\|J_{M_{1,2|3}}(\theta_2)\|.
\end{align*}
Hence, it follows that
\begin{align*}
	\|J_{M_{1,2|3}}(\theta_2)\| \le \|J_{M}(\theta_2,\tilde\theta_3)\| \le \lambda<1,~\text{ for all }\theta_2 \in \mathsf{D}_{K_{1,2|3}},
\end{align*}
which is an analog of Assumption \ref{ass:contraction_global}  for the sub-game. Then, Proposition \ref{prop:contraction_global} ensures the existence and uniqueness of the fixed point in the sub-game. In this section, we focused on identification conditions based on the dynamical system $M$. Similar arguments can be used to establish identification based on the dynamical system $K$.

\subsubsection{Global Identification Conditions.}\label{app:nested_global}

To ensure the existence and uniqueness of the fixed point in the sub-game, we can alternatively rely on the existing global identification conditions in \citet{CH2006} (cf.\ Section \ref{ssec:identification_algorithms} and Lemma \ref{lem:CH}) and Proposition \ref{prop:equivalence}. For every value $\tilde\theta_3\in\Theta_3$, this requires analogues of the conditions in Lemma \ref{lem:CH} to hold for the sub-game between players 1 and 2.

\section{Additional Simulation Results} \label{app:additional_simulation}

\subsection{Bias and RMSE Application-based DGP}
\label{app:bias_rmse}

Here we report simulation evidence on the finite sample bias and RMSE of the different IVQR estimators based on the application-based DGPs in Section \ref{sec:simulation_study}. Tables \ref{tab:401k_1}--\ref{tab:401k_2} present the results. We find that all the proposed algorithms perform well and exhibit comparable bias and RMSE properties.  In particular, the finite sample performances of our preferred estimators are comparable to IQR and the profiling estimator, which shows that their computational advantages do not come at a cost in terms of the finite sample performance.

\begin{table}[ht]
\linespread{1.15}\parskip .05in
\begin{scriptsize}
\centering
\caption{Bias and RMSE, 401(k) DGP with one endogenous regressor}
\begin{tabular}{lcccccccc}
\toprule

& \multicolumn{4}{c}{Bias/$10^2$} & \multicolumn{4}{c}{RMSE/$10^3$}\\
\cmidrule(l{5pt}r{5pt}){2-5}  \cmidrule(l{5pt}r{5pt}){6-9} \

$\tau$& Contr & Brent& Profil & InvQR &  Contr & Brent & Profil& InvQR \\
 \midrule
  \\
0.15 & -4.43 & -6.84 & -6.71 & -7.26 & 7.58 & 7.71 & 7.64 & 7.87 \\ 
  0.25 & -0.28 & -1.86 & -1.90 & -1.67 & 4.04 & 4.10 & 4.11 & 4.10 \\ 
  0.50 & -1.04 & -1.46 & -1.49 & -1.23 & 2.06 & 2.07 & 2.07 & 2.08 \\ 
  0.75 & -1.60 & -1.31 & -1.48 & -1.14 & 1.85 & 1.85 & 1.86 & 1.85 \\ 
  0.85 & 0.61 & 1.17 & 0.85 & 1.24 & 2.06 & 2.07 & 2.07 & 2.08 \\ 
\\
  \bottomrule
\multicolumn{9}{p{10cm}}{\scriptsize {\it Notes:} Monte Carlo simulation with 500 repetitions as described in the main text. Contr: contraction algorithm; Brent: root-finding algorithm based on Brent's method; Profil: profiling estimator based on Brent's method; InvQR: inverse quantile regression. We use 2SLS estimates as starting values.}
\end{tabular}
\label{tab:401k_1}
\end{scriptsize}
\linespread{1.5}\parskip .05in
\end{table}

\begin{table}[ht]
\linespread{1.15}\parskip .05in
\begin{scriptsize}
\centering
\caption{Bias and RMSE, 401(k) DGP with two endogenous regressors}
\begin{tabular}{lcccccccccc}
\toprule
& \multicolumn{5}{c}{Bias/$10^2$} & \multicolumn{5}{c}{RMSE/$10^3$}\\
\cmidrule(l{5pt}r{5pt}){2-6}  \cmidrule(l{5pt}r{5pt}){7-11} \

$\tau$& Contr & NestBr&SimAnn &Profil& InvQR & Contr & NestBr&SimAnn&Profil& InvQR \\
 \midrule
 \\
  \multicolumn{11}{c}{\emph{Coefficient on binary endogenous variable}}\\
  \\
0.15 & -4.81 & -2.73 & 3.77 & -3.65 & -7.77 & 8.29 & 7.84 & 6.74 & 7.95 & 8.60 \\ 
  0.25 & -3.47 & -3.75 & -3.17 & -3.83 & -3.42 & 4.32 & 4.31 & 4.25 & 4.31 & 4.32 \\ 
  0.50 & 0.84 & 0.56 & 0.68 & 0.71 & 0.74 & 1.93 & 1.95 & 1.95 & 1.95 & 1.98 \\ 
  0.75 & -0.56 & -0.33 & -0.37 & -0.57 & -0.25 & 1.75 & 1.74 & 1.74 & 1.74 & 1.78 \\ 
  0.85 & -1.03 & -0.72 & -0.74 & -1.28 & -0.61 & 2.18 & 2.19 & 2.20 & 2.19 & 2.22 \\ 
  
\\

  \multicolumn{11}{c}{\emph{Coefficient on continuous endogenous variable}}\\
  \\
0.15 & 2.00 & 0.48 & 4.72 & -0.03 & 0.23 & 1.07 & 1.07 & 2.20 & 1.07 & 1.19 \\ 
  0.25 & 1.72 & -0.09 & 0.25 & -0.48 & -0.10 & 1.00 & 1.02 & 1.13 & 1.03 & 1.13 \\ 
  0.50 & 0.75 & -0.47 & -0.45 & -0.54 & -0.68 & 0.89 & 0.97 & 0.97 & 0.97 & 1.11 \\ 
  0.75 & -1.43 & -0.49 & -0.44 & -1.57 & -0.15 & 0.99 & 1.10 & 1.11 & 1.12 & 1.24 \\ 
  0.85 & -2.66 & -0.89 & -0.91 & -2.60 & -0.40 & 1.12 & 1.27 & 1.27 & 1.28 & 1.32 \\
\\

  \bottomrule
\multicolumn{11}{p{13.5cm}}{\scriptsize {\it Notes:} Monte Carlo simulation with 500 repetitions as described in the main text. Contr: contraction algorithm; NestBr: nested algorithm based Brent's method; SimAnn: simulated annealing based optimization algorithm; Profil: nested profiling estimator based on Brent's method; IQR: inverse quantile regression. We use 2SLS estimates as starting values.}
\end{tabular}
\label{tab:401k_2}
\end{scriptsize}
\linespread{1.5}\parskip .05in
\end{table}

\subsection{Three Endogenous Variables}
\label{app:sims_3}

Here we present additional simulation evidence with three endogenous variables. We consider the application-based DGP of Section \ref{sec:simulation_study}, augmented with an additional endogenous variable:
\begin{eqnarray}
Y_i = X_i'\theta_X(U_i)+D_{i}\theta_{D}(U_i)+D_{2,i}\theta_{D,2}(U_i) + D_{3,i}\theta_{D,3}(U_i) + G^{-1}(U_i),
\end{eqnarray}
where $\theta_{D,3}(U_i)=10000$, $D_{3,i}=0.8 \cdot Z_{3,i}+0.2 \cdot \Phi^{-1}(U_i)$, and $Z_{3,i}\sim N(0,1)$. We only report the results based on the contraction algorithm and the nested fixed point algorithm. We do not report results for IQR, which we found to be computationally prohibitive with three endogenous regressors. Table \ref{tab:401k_3} shows that both methods exhibit similar performances in terms of bias and RMSE, which are comparable to their respective performances with two endogenous regressors. Table \ref{tab:cpu3} displays average computation times. As expected, the computational advantages of the contraction algorithm relative to the nested fixed point algorithm are more pronounced than with two endogenous variables.

\begin{table}[ht]
\linespread{1.15}\parskip .05in
\begin{scriptsize}
\centering
\caption{Bias and RMSE, 401(k) DGP with three endogenous regressors}
\begin{tabular}{lcccc}
\toprule
& \multicolumn{2}{c}{Bias/$10^2$} & \multicolumn{2}{c}{RMSE/$10^3$}\\
\cmidrule(l{5pt}r{5pt}){2-3}  \cmidrule(l{5pt}r{5pt}){4-5} \

$\tau$& Contr &Nested & Contr & Nested \\
 \midrule
 \\
  \multicolumn{5}{c}{\emph{Coefficient on $D$}}\\
  \\
0.15 & -3.40 & -5.02 & 7.52 & 7.62 \\ 
  0.25 & -0.98 & -1.52 & 4.11 & 4.12 \\ 
  0.50 & -1.04 & -1.42 & 2.03 & 2.06 \\ 
  0.75 & -1.67 & -1.50 & 1.86 & 1.86 \\ 
  0.85 & 0.86 & 1.03 & 2.05 & 2.05 \\ 
\\

  \multicolumn{5}{c}{\emph{Coefficient on $D_2$}}\\
  \\
0.15 & 1.12 & -0.13 & 1.01 & 1.03 \\ 
  0.25 & 2.11 & 0.74 & 1.01 & 1.00 \\ 
  0.50 & 0.80 & -0.28 & 0.94 & 0.99 \\ 
  0.75 & -0.39 & 0.48 & 1.00 & 1.09 \\ 
  0.85 & -2.64 & -0.83 & 1.13 & 1.25 \\ 
  \\
  \multicolumn{5}{c}{\emph{Coefficient on $D_3$}}\\
  \\
0.15 & 1.70 & -0.25 & 1.08 & 1.13 \\ 
  0.25 & 1.57 & -0.21 & 0.99 & 1.01 \\ 
  0.50 & 0.95 & -0.06 & 0.92 & 0.96 \\ 
  0.75 & -1.15 & -0.33 & 1.02 & 1.11 \\ 
  0.85 & -1.01 & 0.23 & 1.22 & 1.37 \\ 
\\
  \bottomrule
\multicolumn{5}{p{5.5cm}}{\scriptsize {\it Notes:} Monte Carlo simulation with 500 repetitions as described in the main text. Contr: contraction algorithm; Nested: nested algorithm based on Brent's method. We use 2SLS estimates as starting values.}
\end{tabular}
\label{tab:401k_3}
\end{scriptsize}
\linespread{1.5}\parskip .05in
\end{table}

\begin{table}[ht]
\linespread{1.15}\parskip .05in
\begin{footnotesize}
\centering
\caption{Computation time, 401(k) DGP with three endogenous regressor}
\begin{tabular}{lcc}
\toprule

$N$& Contr & Nested\\
 \midrule
  \\
1000 & 0.36 & 6.29 \\ 
  5000 & 4.47 & 42.93 \\ 
  10000 & 10.11 & 145.58 \\ 

\\
  \bottomrule
\multicolumn{3}{p{4cm}}{\footnotesize {\it Notes:}  The table reports average computation time in seconds at $\tau=0.5$ over 20 simulation repetitions based on the DGP described in the main text. Contr: contraction algorithm; Nested: nested algorithm based on Brent's method. We use 2SLS estimates as starting values.}
\end{tabular}
\label{tab:cpu3}
\end{footnotesize}
\linespread{1.5}\parskip .05in
\end{table}

\subsection{Additional Simulations Simple Location Scale DGP}
This section presents some additional simulation evidence based on the following location-scale shift model:
\begin{equation}
Y_i= \gamma_1+ \gamma_2X_i+\gamma_3D_{1,i} + \gamma_4D_{2,i}+\left(\gamma_5+\gamma_6 D_{1,i}+\gamma_7D_{2,i}\right)U_i
\end{equation}
Here $D_{1,i}$ and $D_{2,i}$ are the endogenous variables of interest and $X_{i}$ is an exogenous covariate. In addition, we have access to two instruments $Z_{1,i}$ and $Z_{2,i}$. For $\gamma_2=\gamma_4=\gamma_7=0$, this model reduces to the model considered in Section 6.1 of \citet{AndrewsMikusheva2016}. We set $\gamma_1=\dots = \gamma_7=1$. To evaluate the performance of our algorithms with one endogenous variable, we set $\gamma_4=\gamma_7=0$ and use $Z_{1i}$ as the  instrument. Following \citet{AndrewsMikusheva2016}, we consider a symmetric as well as an asymmetric DGP:
\begin{eqnarray*}
(U_i,D_{1,i},D_{2,i},Z_{1,i},Z_{2,i},X_i)&=&\left(\Phi(\xi_{U,i}),\Phi(\xi_{D_1,i}),\Phi(\xi_{D_2,i}),\Phi(\xi_{Z_1,i}),\Phi(\xi_{Z_2,i}), \Phi(\xi_{X,i})\right),~ \text{(symmetric)}\\
(U_i,D_{1,i},D_{2,i},Z_{1,i},Z_{2,i},X_i)&=&\left(\xi_{U,i},\exp(2\xi_{D_1,i}),\exp(2\xi_{D_2,i}),\xi_{Z_1,i},\xi_{Z_2,i}, \xi_{X,i}\right),~ \text{(asymmetric)}
\end{eqnarray*}
where $\left(\xi_{U,i}, \xi_{D_1,i},\xi_{D_2,i},\xi_{Z_1,i},\xi_{Z_2,i}, \xi_{X,i}\right)$ is a Gaussian vector with mean zero, all variances are set equal to one, $Cov(\xi_U,\xi_{D_1})=Cov(\xi_U,\xi_{D_2})=0.5$, $Cov(\xi_{D_1},\xi_{Z_1})=0.8$, $Cov(\xi_{D_2},\xi_{Z_2})=0.4$, which allows us to investigate the impact of instrument strength, all other covariances are equal to zero, and $\Phi$ is the cumulative distribution function of the standard normal distribution.

We first investigate the bias and RMSE of the different methods. Tables \ref{tab:sy1}--\ref{tab:asy2} present the results. With one endogenous variable, the performances of the root-finding algorithm using Brent's method, the profiling estimator, and IQR are similar both in terms of bias and RMSE. The contraction algorithm performs well, but exhibits some bias at the tail quantiles. Turning to the results with two endogenous variables, we can see that the nested algorithm exhibits the best overall performance, both in terms of bias and RMSE. The performances of the SA-based optimization algorithm, IQR, and the profiling estimator are similar and only slightly worse than that of the nested algorithm. The contraction algorithm tends to exhibit some bias at the tail quantiles. However, this bias decreases substantially as the sample size gets larger. 
Finally, comparing the results for the coefficients on $D_1$ and $D_2$, we can see that the instrument strength matters for the performance of all estimators (including IQR), suggesting that weak identification can have implications for the estimation of IVQR models.

\begin{table}[H]
\linespread{1.15}\parskip .05in
\begin{scriptsize}
\centering
\caption{Bias and RMSE, symmetric design with one endogenous regressor}
\begin{tabular}{lcccccccc}
\toprule

& \multicolumn{8}{c}{$N=500$}\\
 \cmidrule(l{5pt}r{5pt}){2-9}

& \multicolumn{4}{c}{Bias} & \multicolumn{4}{c}{RMSE}\\
\cmidrule(l{5pt}r{5pt}){2-5}  \cmidrule(l{5pt}r{5pt}){6-9} \

$\tau$& Contr & Brent& Profil &InvQR &  Contr & Brent&Profil& InvQR \\
 \midrule
  \\
0.15 & 0.03 & -0.00 & -0.02 & -0.00 & 0.10 & 0.10 & 0.11 & 0.10 \\ 
  0.25 & 0.03 & 0.00 & -0.01 & 0.00 & 0.12 & 0.12 & 0.12 & 0.12 \\ 
  0.50 & -0.00 & -0.00 & -0.02 & -0.00 & 0.12 & 0.14 & 0.14 & 0.14 \\ 
  0.75 & -0.04 & -0.01 & -0.03 & -0.01 & 0.13 & 0.12 & 0.12 & 0.12 \\ 
  0.85 & -0.04 & -0.00 & -0.02 & -0.00 & 0.11 & 0.11 & 0.11 & 0.11 \\ 
\\

\midrule
& \multicolumn{8}{c}{$N=1000$}\\
 \cmidrule(l{5pt}r{5pt}){2-9}

& \multicolumn{4}{c}{Bias} & \multicolumn{4}{c}{RMSE}\\
\cmidrule(l{5pt}r{5pt}){2-5}  \cmidrule(l{5pt}r{5pt}){6-9} \

$\tau$& Contr & Brent& Profil &InvQR &  Contr & Brent&Profil& InvQR \\
 \midrule
  \\
0.15 & 0.02 & 0.00 & -0.00 & 0.00 & 0.07 & 0.07 & 0.07 & 0.07 \\ 
  0.25 & 0.01 & -0.00 & -0.01 & -0.00 & 0.08 & 0.08 & 0.08 & 0.08 \\ 
  0.50 & -0.01 & -0.01 & -0.01 & -0.01 & 0.09 & 0.10 & 0.10 & 0.10 \\ 
  0.75 & -0.02 & -0.00 & -0.01 & -0.00 & 0.09 & 0.08 & 0.08 & 0.08 \\ 
  0.85 & -0.02 & -0.00 & -0.01 & -0.00 & 0.08 & 0.08 & 0.08 & 0.08 \\ 
\\
  \bottomrule
\multicolumn{9}{p{10cm}}{\scriptsize {\it Notes:} Monte Carlo simulation with 500 repetitions as described in the main text. Contr: contraction algorithm; Brent: root-finding algorithm based on Brent's method; Profil: profiling estimator based on Brent's method; InvQR: inverse quantile regression. We use 2SLS estimates as starting values.}
\end{tabular}
\label{tab:sy1}
\end{scriptsize}
\linespread{1.5}\parskip .05in
\end{table}

\begin{table}[H]
\linespread{1.15}\parskip .05in
\begin{scriptsize}
\centering
\caption{Bias and RMSE, asymmetric design with one endogenous regressor}
\begin{tabular}{lcccccccc}
\toprule

& \multicolumn{8}{c}{$N=500$}\\
 \cmidrule(l{5pt}r{5pt}){2-9}

& \multicolumn{4}{c}{Bias} & \multicolumn{4}{c}{RMSE}\\
\cmidrule(l{5pt}r{5pt}){2-5}  \cmidrule(l{5pt}r{5pt}){6-9} \

$\tau$& Contr & Brent& Profil &InvQR &  Contr & Brent&Profil& InvQR \\
 \midrule
  \\
0.15 & 0.11 & 0.01 & -0.04 & -0.00 & 0.22 & 0.20 & 0.21 & 0.20 \\ 
  0.25 & 0.07 & 0.00 & -0.02 & -0.00 & 0.17 & 0.16 & 0.17 & 0.16 \\ 
  0.50 & 0.04 & -0.00 & -0.02 & -0.00 & 0.13 & 0.12 & 0.12 & 0.12 \\ 
  0.75 & 0.03 & 0.00 & -0.01 & 0.00 & 0.11 & 0.11 & 0.11 & 0.11 \\ 
  0.85 & -0.03 & -0.01 & -0.03 & -0.00 & 0.12 & 0.11 & 0.12 & 0.11 \\ 
\\
\midrule
& \multicolumn{8}{c}{$N=1000$}\\
 \cmidrule(l{5pt}r{5pt}){2-9}

& \multicolumn{4}{c}{Bias} & \multicolumn{4}{c}{RMSE}\\
\cmidrule(l{5pt}r{5pt}){2-5}  \cmidrule(l{5pt}r{5pt}){6-9} \

$\tau$& Contr & Brent& Profil &InvQR &  Contr & Brent&Profil& InvQR \\
 \midrule
  \\
0.15 & 0.05 & -0.01 & -0.03 & -0.01 & 0.16 & 0.15 & 0.15 & 0.15 \\ 
  0.25 & 0.04 & 0.00 & -0.01 & 0.00 & 0.11 & 0.11 & 0.11 & 0.11 \\ 
  0.50 & 0.03 & 0.00 & -0.01 & 0.00 & 0.08 & 0.08 & 0.08 & 0.08 \\ 
  0.75 & 0.01 & -0.01 & -0.02 & -0.01 & 0.08 & 0.08 & 0.08 & 0.08 \\ 
  0.85 & -0.03 & -0.01 & -0.02 & -0.01 & 0.09 & 0.09 & 0.09 & 0.09 \\ 
\\
  \bottomrule
\multicolumn{9}{p{10cm}}{\scriptsize {\it Notes:} Monte Carlo simulation with 500 repetitions as described in the main text. Contr: contraction algorithm; Brent: root-finding algorithm based on Brent's method; Profil: profiling estimator based on Brent's method; InvQR: inverse quantile regression. We use 2SLS estimates as starting values.}
\end{tabular}
\label{tab:asy1}
\end{scriptsize}
\linespread{1.5}\parskip .05in
\end{table}

\begin{table}[H]
\linespread{1.15}\parskip .05in
\begin{scriptsize}
\centering
\caption{Bias and RMSE, symmetric design with two endogenous regressors}
\begin{tabular}{lcccccccccc}
\toprule

& \multicolumn{10}{c}{$N=500$}\\
 \cmidrule(l{5pt}r{5pt}){2-11}
& \multicolumn{5}{c}{Bias} & \multicolumn{5}{c}{RMSE}\\
\cmidrule(l{5pt}r{5pt}){2-6}  \cmidrule(l{5pt}r{5pt}){7-11} \
$\tau$& Contr & NestBr&SimAnn &Profil& InvQR & Contr & NestBr&SimAnn&Profil& InvQR \\
 \midrule
  \\
  \multicolumn{11}{c}{\emph{Coefficient on $D_1$}}\\
  \\
0.15 & 0.00 & -0.00 & 0.00 & -0.02 & -0.01 & 0.11 & 0.12 & 0.14 & 0.13 & 0.13 \\ 
  0.25 & 0.01 & -0.00 & -0.01 & -0.02 & -0.01 & 0.15 & 0.16 & 0.17 & 0.16 & 0.16 \\ 
  0.50 & -0.02 & -0.02 & -0.02 & -0.04 & -0.02 & 0.17 & 0.19 & 0.19 & 0.19 & 0.20 \\ 
  0.75 & -0.04 & -0.03 & -0.03 & -0.05 & -0.03 & 0.21 & 0.20 & 0.21 & 0.20 & 0.20 \\ 
  0.85 & -0.05 & -0.03 & -0.03 & -0.05 & -0.03 & 0.18 & 0.17 & 0.18 & 0.18 & 0.17 \\ 
\\
  \multicolumn{11}{c}{\emph{Coefficient on $D_2$}}\\
  \\
0.15 & 0.10 & -0.01 & -0.01 & -0.05 & -0.02 & 0.27 & 0.27 & 0.29 & 0.28 & 0.31 \\ 
  0.25 & 0.10 & -0.00 & -0.02 & -0.04 & -0.02 & 0.29 & 0.29 & 0.30 & 0.30 & 0.30 \\ 
  0.50 & -0.01 & -0.02 & -0.02 & -0.06 & -0.02 & 0.33 & 0.38 & 0.39 & 0.40 & 0.39 \\ 
  0.75 & -0.15 & -0.04 & -0.06 & -0.08 & -0.05 & 0.40 & 0.40 & 0.41 & 0.41 & 0.41 \\ 
  0.85 & -0.19 & -0.05 & -0.06 & -0.10 & -0.07 & 0.39 & 0.36 & 0.40 & 0.38 & 0.43 \\
  
\\

\midrule

& \multicolumn{10}{c}{$N=1000$}\\
 \cmidrule(l{5pt}r{5pt}){2-11}
& \multicolumn{5}{c}{Bias} & \multicolumn{5}{c}{RMSE}\\
\cmidrule(l{5pt}r{5pt}){2-6}  \cmidrule(l{5pt}r{5pt}){7-11} \
$\tau$& Contr & NestBr&SimAnn &Profil& InvQR & Contr & NestBr&SimAnn&Profil& InvQR \\

 \midrule
  \\
  \multicolumn{11}{c}{\emph{Coefficient on $D_1$}}\\
  \\
0.15 & -0.00 & -0.00 & -0.01 & -0.01 & -0.00 & 0.08 & 0.09 & 0.10 & 0.09 & 0.10 \\ 
  0.25 & -0.00 & -0.00 & -0.01 & -0.01 & -0.01 & 0.10 & 0.11 & 0.12 & 0.11 & 0.13 \\ 
  0.50 & -0.01 & -0.01 & -0.01 & -0.02 & -0.01 & 0.12 & 0.13 & 0.13 & 0.13 & 0.16 \\ 
  0.75 & -0.01 & -0.01 & -0.01 & -0.01 & -0.00 & 0.13 & 0.13 & 0.14 & 0.13 & 0.14 \\ 
  0.85 & -0.02 & -0.01 & -0.02 & -0.02 & -0.02 & 0.12 & 0.12 & 0.13 & 0.12 & 0.13 \\ 
\\
  \multicolumn{11}{c}{\emph{Coefficient on $D_2$}}\\
  \\
0.15 & 0.05 & -0.01 & -0.01 & -0.02 & -0.02 & 0.19 & 0.19 & 0.21 & 0.19 & 0.20 \\ 
  0.25 & 0.05 & -0.00 & -0.01 & -0.02 & -0.01 & 0.22 & 0.21 & 0.23 & 0.22 & 0.23 \\ 
  0.50 & -0.02 & -0.02 & -0.02 & -0.04 & -0.03 & 0.25 & 0.27 & 0.27 & 0.28 & 0.29 \\ 
  0.75 & -0.09 & -0.02 & -0.02 & -0.04 & -0.03 & 0.27 & 0.25 & 0.28 & 0.25 & 0.26 \\ 
  0.85 & -0.09 & -0.01 & -0.03 & -0.04 & -0.03 & 0.26 & 0.23 & 0.25 & 0.24 & 0.24 \\ 
\\
  \bottomrule
\multicolumn{11}{p{13.5cm}}{\scriptsize {\it Notes:} Monte Carlo simulation with 500 repetitions as described in the main text. Contr: contraction algorithm; NestBr: nested algorithm based on Brent's method; SimAnn: simulated annealing based optimization algorithm; Profil: nested profiling estimator based on Brent's method; InvQR: inverse quantile regression. We use 2SLS estimates as starting values.}
\end{tabular}
\label{tab:sy2}
\end{scriptsize}
\linespread{1.5}\parskip .05in
\end{table}

\begin{table}[H]
\linespread{1.15}\parskip .05in
\begin{scriptsize}
\centering
\caption{Bias and RMSE, asymmetric design with two endogenous regressors}
\begin{tabular}{lcccccccccc}
\toprule

& \multicolumn{10}{c}{$N=500$}\\
 \cmidrule(l{5pt}r{5pt}){2-11}
& \multicolumn{5}{c}{Bias} & \multicolumn{5}{c}{RMSE}\\
\cmidrule(l{5pt}r{5pt}){2-6}  \cmidrule(l{5pt}r{5pt}){7-11} \
$\tau$& Contr & NestBr&SimAnn &Profil& InvQR & Contr & NestBr&SimAnn&Profil& InvQR \\
 \midrule
  \\
  \multicolumn{11}{c}{\emph{Coefficient on $D_1$}}\\
  \\
0.15 & -0.02 & 0.02 & 0.00 & -0.03 & 0.01 & 0.25 & 0.26 & 0.28 & 0.27 & 0.26 \\ 
  0.25 & -0.05 & 0.01 & 0.00 & -0.01 & -0.00 & 0.20 & 0.20 & 0.21 & 0.20 & 0.21 \\ 
  0.50 & -0.04 & -0.00 & 0.00 & -0.01 & 0.00 & 0.16 & 0.17 & 0.21 & 0.17 & 0.19 \\ 
  0.75 & -0.02 & -0.02 & -0.01 & -0.02 & -0.02 & 0.17 & 0.17 & 0.18 & 0.18 & 0.19 \\ 
  0.85 & -0.01 & -0.01 & -0.02 & -0.02 & -0.02 & 0.20 & 0.19 & 0.19 & 0.20 & 0.19 \\ 
\\
  \multicolumn{11}{c}{\emph{Coefficient on $D_2$}}\\
  \\
0.15 & 0.26 & -0.06 & -0.11 & -0.16 & -0.13 & 0.57 & 0.52 & 0.58 & 0.53 & 0.59 \\ 
  0.25 & 0.23 & -0.01 & -0.02 & -0.06 & -0.01 & 0.45 & 0.41 & 0.43 & 0.40 & 0.44 \\ 
  0.50 & 0.12 & -0.03 & -0.04 & -0.07 & -0.07 & 0.34 & 0.32 & 0.48 & 0.33 & 0.73 \\ 
  0.75 & 0.04 & -0.05 & -0.06 & -0.11 & -0.05 & 0.32 & 0.31 & 0.34 & 0.33 & 0.34 \\ 
  0.85 & -0.13 & -0.03 & -0.01 & -0.08 & 0.01 & 0.40 & 0.34 & 0.38 & 0.34 & 0.36 \\ 
\\

\midrule

& \multicolumn{10}{c}{$N=1000$}\\
 \cmidrule(l{5pt}r{5pt}){2-11}
& \multicolumn{5}{c}{Bias} & \multicolumn{5}{c}{RMSE}\\
\cmidrule(l{5pt}r{5pt}){2-6}  \cmidrule(l{5pt}r{5pt}){7-11} \
$\tau$& Contr & NestBr&SimAnn &Profil& InvQR & Contr & NestBr&SimAnn&Profil& InvQR \\

 \midrule
  \\
  \multicolumn{11}{c}{\emph{Coefficient on $D_1$}}\\
  \\
0.15 & -0.03 & 0.01 & -0.00 & -0.02 & -0.01 & 0.18 & 0.19 & 0.19 & 0.18 & 0.19 \\ 
  0.25 & -0.04 & -0.00 & -0.01 & -0.02 & -0.01 & 0.15 & 0.15 & 0.16 & 0.15 & 0.16 \\ 
  0.50 & -0.03 & -0.01 & -0.01 & -0.01 & -0.01 & 0.13 & 0.13 & 0.14 & 0.13 & 0.14 \\ 
  0.75 & -0.03 & -0.01 & -0.01 & -0.02 & -0.01 & 0.12 & 0.12 & 0.13 & 0.13 & 0.14 \\ 
  0.85 & 0.01 & 0.00 & 0.00 & -0.01 & -0.00 & 0.14 & 0.13 & 0.15 & 0.14 & 0.15 \\ 
\\
  \multicolumn{11}{c}{\emph{Coefficient on $D_2$}}\\
  \\
0.15 & 0.15 & -0.03 & -0.03 & -0.07 & -0.04 & 0.37 & 0.37 & 0.38 & 0.37 & 0.39 \\ 
  0.25 & 0.10 & -0.01 & -0.01 & -0.05 & -0.02 & 0.28 & 0.28 & 0.30 & 0.28 & 0.28 \\ 
  0.50 & 0.05 & -0.02 & -0.03 & -0.04 & -0.03 & 0.22 & 0.22 & 0.23 & 0.22 & 0.24 \\ 
  0.75 & 0.06 & -0.01 & -0.02 & -0.03 & -0.01 & 0.24 & 0.22 & 0.24 & 0.23 & 0.24 \\ 
  0.85 & -0.08 & -0.03 & -0.04 & -0.07 & -0.03 & 0.27 & 0.24 & 0.26 & 0.25 & 0.24 \\ 
\\
  \bottomrule
\multicolumn{11}{p{13.5cm}}{\scriptsize {\it Notes:} Monte Carlo simulation with 500 repetitions as described in the main text. Contr: contraction algorithm; NestBr: nested algorithm based on Brent's method; SimAnn: simulated annealing based optimization algorithm; Profil: nested profiling estimator based on Brent's method; InvQR: inverse quantile regression. We use 2SLS estimates as starting values.}
\end{tabular}
\label{tab:asy2}
\end{scriptsize}
\linespread{1.5}\parskip .05in
\end{table}

Table \ref{tab:size_locscale} displays the empirical coverage probabilities of the bootstrap confidence intervals. The results show that the our bootstrap procedure exhibits excellent size properties.

\begin{table}[ht]
\linespread{1.15}\parskip .05in
\begin{scriptsize}
\centering
\caption{Coverage, location-scale DGP with one endogenous regressor}
\begin{tabular}{lcccccccc}
\toprule

\multicolumn{9}{c}{$N=500$} \\
 \cmidrule(l{5pt}r{5pt}){2-9}
& \multicolumn{4}{c}{Symmetric Design}& \multicolumn{4}{c}{Asymmetric Design}  \\
\cmidrule(l{5pt}r{5pt}){2-5} \cmidrule(l{5pt}r{5pt}){6-9}
& \multicolumn{2}{c}{$1-\alpha=0.95$} & \multicolumn{2}{c}{$1-\alpha=0.9$}& \multicolumn{2}{c}{$1-\alpha=0.95$} & \multicolumn{2}{c}{$1-\alpha=0.9$}\\
\cmidrule(l{5pt}r{5pt}){2-3}  \cmidrule(l{5pt}r{5pt}){4-5} \cmidrule(l{5pt}r{5pt}){6-7}  \cmidrule(l{5pt}r{5pt}){8-9}

$\tau$& Contr & Brent&  Contr & Brent & Contr & Brent&  Contr & Brent\\
 \midrule
  \\
0.15 & 0.93 & 0.98 & 0.88 & 0.94 & 0.89 & 0.97 & 0.85 & 0.95 \\ 
  0.25 & 0.94 & 0.96 & 0.89 & 0.93 & 0.93 & 0.97 & 0.88 & 0.95 \\ 
  0.50 & 0.96 & 0.96 & 0.91 & 0.91 & 0.95 & 0.97 & 0.91 & 0.93 \\ 
  0.75 & 0.94 & 0.97 & 0.90 & 0.94 & 0.95 & 0.97 & 0.91 & 0.93 \\ 
  0.85 & 0.95 & 0.98 & 0.90 & 0.96 & 0.96 & 0.98 & 0.92 & 0.96 \\ 
\\
\midrule
\multicolumn{9}{c}{$N=1000$} \\
 \cmidrule(l{5pt}r{5pt}){2-9}
& \multicolumn{4}{c}{Symmetric Design}& \multicolumn{4}{c}{Asymmetric Design}  \\
\cmidrule(l{5pt}r{5pt}){2-5} \cmidrule(l{5pt}r{5pt}){6-9}
& \multicolumn{2}{c}{$1-\alpha=0.95$} & \multicolumn{2}{c}{$1-\alpha=0.9$}& \multicolumn{2}{c}{$1-\alpha=0.95$} & \multicolumn{2}{c}{$1-\alpha=0.9$}\\
\cmidrule(l{5pt}r{5pt}){2-3}  \cmidrule(l{5pt}r{5pt}){4-5} \cmidrule(l{5pt}r{5pt}){6-7}  \cmidrule(l{5pt}r{5pt}){8-9}

$\tau$& Contr & Brent&  Contr & Brent & Contr & Brent&  Contr & Brent\\
 \midrule
  \\
0.15 & 0.95 & 0.96 & 0.90 & 0.92 & 0.92 & 0.97 & 0.85 & 0.93 \\ 
  0.25 & 0.95 & 0.96 & 0.91 & 0.91 & 0.92 & 0.95 & 0.87 & 0.91 \\ 
  0.50 & 0.96 & 0.96 & 0.90 & 0.90 & 0.94 & 0.96 & 0.90 & 0.93 \\ 
  0.75 & 0.95 & 0.96 & 0.90 & 0.91 & 0.95 & 0.95 & 0.91 & 0.91 \\ 
  0.85 & 0.96 & 0.97 & 0.93 & 0.94 & 0.96 & 0.95 & 0.92 & 0.91 \\ 
\\
  \bottomrule
\multicolumn{9}{p{10cm}}{\scriptsize {\it Notes:} Monte Carlo simulation with 1000 repetitions as described in the main text. Contr: contraction algorithm; Brent: root-finding algorithm based on Brent's method. We use 2SLS estimates as starting values.}
\end{tabular}
\label{tab:size_locscale}
\end{scriptsize}
\linespread{1.5}\parskip .05in
\end{table}

\section{Proofs of Theoretical Results in Section \ref{sec:decentralization}}\label{sec:appendix_proof_sec3}
\begin{proof}[\rm \textbf{Proof of Lemma \ref{lem:global_decentralization}}]

(i)	We first show that $L_1$ is well-defined. For a given $\theta_{-1}\in\mathbb R^{d_D}$, let $\theta_1^*\in\arg\min_{\tilde\theta_1 \in \mathbb{R}^{d_X}}Q_{P,1}( \tilde\theta_1,\theta_{-1})$. Under Assumption \ref{ass:global_decentralization}, the objective function is convex and differentiable with respect to $\tilde\theta_1$. Therefore, by the necessary and sufficient condition of minimization, $\theta_1^*$ solves
\begin{align*}
E_P[(1\{Y\le D'\theta_{-1}+X'\theta^*_1\})X]=0.
\end{align*}
In what follows, we show that the map $L_1:\theta_{-1}\mapsto\theta^*_1$ is well-defined on $R_{-1}$ using a global inverse function theorem.  Recall that
\begin{align}
\Z_{P,1}(\theta)= E_P[(1\{Y\le D'\theta_{-1}+X'\theta_1\})X].
\end{align}
This function  is continuously differentiable with respect to $\theta$. The Jacobian is given by
\begin{align}
J_{\Z_{P,1}}(\theta)=\frac{\partial }{\partial \theta'}   E_P[F_{Y|D,X,Z}(D'\theta_{-1}+X'\theta_1)X]=E_P[f_{Y|D,X,Z}(D'\theta_{-1}+X'\theta_1)X(X',D')],\label{eq:Jacobian_Psi}
\end{align}
where the second equality follows from Assumption \ref{ass:global_decentralization} and the dominated convergence theorem.
Define a transform $\Xi:\Theta\to\mathbb{R}^{d_X+d_D}$ by
\begin{align}
\Xi(\theta)
\equiv(\Z_{P,1}(\theta)',\theta_{-1}{}')'. \label{eq:xi}
\end{align}
We follow \cite{Krantz:2012aa} (Section 3.3) to obtain an implicit function $L_1$ on a suitable domain such that
$\theta_1=L_1(\theta_2)$ if and only if $\Z_{P,1}(\theta)=0.$
The key is to apply a global inverse function theorem to $\Xi$.
Toward this end, we analyze the Jacobian of $\Xi$, which is given as
\begin{align}
J_\Xi(\theta)&=\begin{bmatrix}
	\partial\Z_{P,1}(\theta_1,\theta_{-1})/\partial\theta_{1}'& \partial\Z_{P,1}(\theta_1,\theta_{-1})/\partial\theta_{-1}'\\
	0_{d_{-1}\times d_{1}}& I_{d_{-1}}
\end{bmatrix}\notag\\
&=\begin{bmatrix}
	E_P[f_{Y|D,X,Z}(D'\theta_{-1}+X'\theta_1)XX']&E_P[f_{Y|D,X,Z}(D'\theta_{-1}+X'\theta_1)XD']\\
	0_{d_{-1}\times d_{1}}& I_{d_{-1}}
\end{bmatrix},\label{eq:Jacobian1}
\end{align}
where, for any $d\in\mathbb N$, $I_d$ denotes the $d\times d$ identity matrix.
Let $I\subset\{1,\dots \dparam\}$.

For any matrix $A$, let $[A]_{I,I}$ denote a principal minor of $A$, which collects the rows and columns of $A$ whose indices belong to the index set $I$.
By \eqref{eq:Jacobian1}, if $I\subset \{1,\dots,d_1\}$,
\begin{align}
	[J_\Xi(\theta)]_{I,I}=E_P[f_{Y|D,X,Z}(D'\theta_{-1}+X'\theta_1)\tilde X\tilde X']
\end{align}
for a subvector $\tilde X$ of $X$, which is positive definite by Assumption \ref{ass:global_decentralization} and Lemma \ref{lem:pd}. If $I\subset \{d_1+1,\dots,\dparam\}$, $[J_\Xi(\theta)]_{I,I}=I_{\ell}$ for some $1\le\ell\le d_D$ and is hence positive definite. Otherwise,  any principal minor is of the following form:
\begin{align}
[J_\Xi(\theta)]_{I,I}=\begin{bmatrix}
	E_P[f_{Y|D,X,Z}(D'\theta_{-1}+X'\theta_1)\tilde X\tilde X']&B\\
	0_{\ell \times m}& I_{\ell}
\end{bmatrix}
\end{align}
for some subvector $\tilde X$ of $X$ and a $m\times\ell$ matrix $B.$
Note that
\begin{multline}
\det([J_\Xi(\theta)]_{I,I})=\det(E_P[f_{Y|D,X,Z}(D'\theta_{-1}+X'\theta_1)\tilde X\tilde X']-B I_\ell^{-1}\times 0_{\ell \times m})\det(I_\ell)\\
=\det(E_P[f_{Y|D,X,Z}(D'\theta_{-1}+X'\theta_1)\tilde X\tilde X'])>0,
\end{multline}
where the last inequality follows again from Assumption \ref{ass:global_decentralization} and Lemma \ref{lem:pd}. 
Hence, $J_\Xi(\theta)$ is a $P$-matrix. Note that $\Theta$ is a closed rectangle. By Theorem 4 in \citet{gale1965jacobian}, $\Xi$ is univalent, and hence the inverse map $\Xi^{-1}$ is well defined.

Let
\begin{align*}
R_{-1}=\{\theta_{-1}\in\mathbb R^{d_{-1}}:(0,\theta_{-1})\in \Xi(\Theta)\}=\{\theta_{-1}\in\mathbb R^{d_{-1}}:\Z_{P,1}(\theta_1,\theta_{-1})=0,\text{ for some }(\theta_1,\theta_{-1})\in\Theta\},
\end{align*}
which coincides with the definition in \eqref{eq:defRj} with $j=1$.
Let $F_1=[I_{d_1},0_{d_1\times d_{-1}}]$. For each $\theta_{-1}\in R_{-1}$, define
\begin{align*}
L_1(\theta_{-1})\equiv F_1\Xi^{-1}(0,\theta_{-1}).
\end{align*}
Then, for any $\theta\in\Theta$, $\Z_{P,1}(\theta)=0$ if and only if $\theta_{-1}\in R_{-1}$ and $\Xi(\theta)=(0,\theta_{-1})$. By the univalence of $\Xi$, this is true if and only if
$\theta=\Xi^{-1}(0,\theta_{-1})$, and the first $d_1$ components extracted by applying $F_1$ is $\theta_1$. This ensures $L_1$ is well-defined on $R_{-1}.$

Below, for any set $A,$ let $A^o$ denote the interior of $A$. Let $R_{-1}^o=\{\theta_{-1}\in\mathbb R^{d_{-1}}:(0,\theta_{-1})\in \Xi(\Theta^o)\}.$
Note that $\Z_{P,1}$ is $\mathcal C^1$ on $\Theta^o$ and, for each $\theta=(\theta_1,\theta_{-1})\in \Theta$ with $\theta_{-1}\in R_{d_{-1}}^o$, $\det(\partial\Z_{P,1}(\theta)/\partial\theta_1')\ne 0$.
Therefore, by the implicit function theorem, there is a $\mathcal C^1$-function $\tilde L_1$ and an open set $V$ containing $\theta_{-1}$ such that
\begin{align*}
\Z_{P,1}(\tilde L_1(\theta_{-1}),\theta_{-1})=0,\text{ for all }\theta_{-1}\in V.
\end{align*}
However, such a local implicit function must coincide with the unique global map $L_1$ on $V$. Hence, $L_1|_V=\tilde L_1$, and therefore $L_1$ is continuously differentiable at $\theta_{-1}$. Since the choice of $\theta_{-1}$ is arbitrary, $L_1$ is continuously differentiable for all $\theta_{-1}\in R_2^o.$

Showing that the conclusion holds for any other $L_j$ for $j=2,\dots,J$ is similar, and hence we omit the proof.
\end{proof}

\begin{lemma}\label{lem:pd} Suppose	$E_P[f_{Y|D,X,Z}\left(D'\param_{-1}+X^\prime \param_1\right) X X']$ is positive definite. Then, for any subvector $\tilde X$ of $X$ with dimension $\tilde d_X\le d_X$, $E_P[f_{Y|D,X,Z}\left(D'\param_{-1}+X^\prime \param_1\right) \tilde X \tilde X']$ is positive definite.
\end{lemma}

\begin{proof}
  In what follows, let $W=f_{Y|D,X,Z}\left(D'\param_{-1}+X^\prime \param_1\right)$ and let
  \begin{align}
  A\equiv E_P[f_{Y|D,X,Z}\left(D'\param_{-1}+X^\prime \param_1\right) X X']=E[WXX'].\label{eq:pd1}
  \end{align}
Let $\tilde X$ be a subvector of $X$ with $\tilde d_X$ components. Then, there exists a $d_X\times d_X$ permutation matrix $P_\pi$ such that the first $\tilde d_X$ components of $P_\pi X$ is $\tilde X$.

Let $B\equiv E[WP_\pi XX'P_\pi']$ and note that
\begin{align}
  B=P_\pi E[WXX']P_\pi'=P_\pi AP_\pi',\label{eq:pd3}
\end{align}
by the linearity of the expectation operator and $W$ being a scalar. Let $\lambda$ be an eigenvalue of $B$ such that
\begin{align}
  Bz = \lambda z,\label{eq:pd4}
\end{align}
for the corresponding eigenvector $z\in\mathbb R^{d_X}$. By \eqref{eq:pd3}-\eqref{eq:pd4},
\begin{align}
  P_\pi AP_\pi'z=\lambda z~\Leftrightarrow~ AP_\pi'z = \lambda P_\pi^{-1} z.
\end{align}
Note that $P_{\pi}^{-1}=P_\pi'$ due to $P_\pi$ being a permutation matrix. Letting $y\equiv P_\pi'z$ then yields
\begin{align}
  A y =\lambda y,
\end{align}
which in turn shows that $\lambda$ is an eigenvalue of $A$. For any eigenvalue of $A$, the argument above can be reversed to show that it is also an eigenvalue of $B$.
Since the choice of the eigenvalue is arbitrary, $A$ and $B$ share the same eigenvalues.

Now let $C\equiv E[W\tilde X\tilde X']$ and note that it is a leading principal submatrix of $B$. Then, by the eigenvalue inclusion principle \citep[][Theorem 4.3.28]{Horn1990},
\begin{align}
  \lambda_{\text{min}}(C)\ge   \lambda_{\text{min}}(B)=  \lambda_{\text{min}}(A)>0,
\end{align}
where the last inequality follows from the positive definiteness of $A$. This completes the claim of the lemma.
\end{proof}

\begin{proof}[{\rm\textbf{Proof of Corollary \ref{cor:global_decentralization}}}]
	The existence of $K$ and its continuous differentiability follows immediately from Lemma \ref{lem:global_decentralization}. For $M$, by  the definition of $\tilde R_{1}$, for any $\theta_{-1}\in \tilde R_{j}$, there exists $(\theta_1,\theta_2)\in\Theta_1\times\Theta_2$ such that
	\begin{align}
	&\Z_{P,1}(\theta_1,\theta_{-1})=0,\\
	&\Z_{P,2}(\theta_1,\theta_2,\pi_{-\{1,2\}}\theta_{-1})=0.
	\end{align}
	By (i), one may then write $\theta_1=L_1(\theta_{-1})$ and $\theta_2=L_2(L_1(\theta_{-1}),\pi_{-\{1,2\}}\theta_{-1})$.
	Hence, the map $M_1:\tilde R_1\to\Theta_2$  below is well-defined:
		\begin{align}
		M_1(\theta_{-1})=L_2\big(L_1(\theta_{-1}),\pi_{-\{1,2\}}\theta_{-1}\big).
	\end{align}
	Recursively, arguing in the same way, the maps
	\begin{align}
		M_2(\theta_{-1})&=L_3\big(L_1(\theta_{-1}),M_1(\theta_{-1}),\pi_{-\{1,2,3\}}\theta_{-1}\big)\\
		\vdots&\notag\\
		M_j(\theta_{-1})&=L_{j+1}\big(L_1(\theta_{-1}),M_1(\theta_{-1}),\dots, M_{j-1}(\theta_{-1}),\pi_{-\{1,\dots,j+1\}}\theta_{-1}\big)\\
		\vdots&\notag\\
		M_{d_D}(\theta_{-1})&=L_J\big(L_1(\theta_{-1}),M_1(\theta_{-1}),\dots, M_{d_D-1}(\theta_{-1})\big)
	\end{align}
	are well-defined on $\tilde R_2,\dots,\tilde R_{d_D}$ respectively. The continuous differentiability of $M$ follows from that of $L_j$s and the chain rule.
\end{proof}

\begin{proof}[\rm \textbf{Proof of Proposition \ref{prop:equivalence}}]
$\Longrightarrow$: For every solution, $\Z_{P}(\true)=0$, $\true_j=L_{j}\left(\true_{-j}\right)$ by construction under Assumptions \ref{ass:ivqr} and \ref{ass:global_decentralization}. It follows that $K\left(\true \right)=\true$ and $M\left(\true_{-1}\right)=\true_{-1}$. \\
$\Longleftarrow$: For the simultaneous response note that $K\left(\fp\right)=\fp$ implies that $\fp_j=L_{j}\left(\fp_{-j}\right)$ for all $j\in \{1,\dots,J\}$. Thus, $\fp$ solves $\Z_{P}(\fp)=0$ by Lemma \ref{lem:global_decentralization}. Consider next the sequential response. Let $\tilde\theta,\bar\theta\in\Theta$ be such that $\tilde\theta_j=L_j(\bar\theta_{-j})$ for $j=1,\dots,J$. By Lemma \ref{lem:global_decentralization}, they satisfy
\begin{eqnarray*}
\Z_{P,1}\left(\tilde\theta_1,\bar\theta_{2},\cdots, \bar\theta_{J}\right)&=&0\\
\Z_{P,2}\left(\tilde\theta_1,\tilde\theta_{2},\cdots, \tilde\theta_{J}\right)&=&0\\
&\vdots&\\
\Z_{P,J}\left(\tilde\theta_1,\tilde\theta_{2},\cdots, \tilde\theta_{J}\right)&=&0
\end{eqnarray*}
Thus, a fixed point $\tilde\theta=\fp$ satisfies $\Z_{P}\left(\fp\right)=0$. 
\end{proof}

\section{Proofs of Theoretical Results in Section \ref{sec:population_algorithms}}
\begin{proof}[\rm \textbf{Proof of Proposition \ref{prop:contraction_global}}]
We prove the result for $K$. By Assumption \ref{ass:contraction_global}, there exists a strictly convex set $\tilde D_K$ on which the spectral norm of the Jacobian of $K$ is uniformly bounded by 1. This ensures that $K$ is a contraction map on $cl(\tilde D_K)$, and the claim of the proposition now follows from Theorem 2.2.16 in \citet{HasselblattKatok2003}.
 \end{proof}

\section{Proofs of Theoretical Results in Section \ref{sec:theory}}
\label{sec:proofs_asydist_bootstrap}
To state and prove results in a concise manner, we use the population and sample simultaneous response maps $K$ and $\hat K$ below to define our estimand $\theta^*$ and estimator $\est$. Namely, $\theta^*$ is the fixed point of $K$, and $\est$ solves 
\begin{align}
\|\est-\hat K(\est)\|\le \inf_{\theta'\in\Theta}\|\theta'-\hat K(\theta')\|+o_p(N^{-1/2}).
\end{align}	
Note that the fixed point estimator defined in \eqref{eq:hattheta1}-\eqref{eq:hattheta2} is asymptotically equivalent to the  estimator above due to Lemma \ref{lem:asy_equiv}.
\begin{proof}[\rm \textbf{Proof of Theorem \ref{thm:asynorm}}]
Let $H\equiv I_\dparam-K$.  A fixed point $\true$ of $K$ then satisfies
\begin{align*}
H(\true)=0.
\end{align*}
Similarly, let $\hat H\equiv I_\dparam-\hat K$. The estimator $\est$  satisfies
\begin{align}
\|\hat H(\est)\|^2\le \inf_{\theta'\in\Theta}\|\hat H(\theta)\|^2+r_N^2,
\end{align}
where $r_N=o_p(N^{-1/2})$. Let $\varphi:\ell^\infty(\Theta)^\dparam\times\mathbb R\to \mathbb{R}^{d_X+d_D}$ be a map such that, for each $(H,r)\in \ell^\infty(\Theta)^\dparam\times\mathbb R$, $\tilde \theta=\varphi(H,r)$ is an $r$-approximate solution, which  satisfies
\begin{align}
\|H(\tilde\theta)\|^2\le \inf_{\theta'\in\Theta}\|H(\theta')\|^2+r^2.
\end{align}
One may then write
\begin{align}
\sqrt N(\est-\true)=\sqrt N(\varphi(\hat H,\hat r)-\varphi(H,0)).
\end{align}
By Lemma \ref{cor:LMdist}, $	\sqrt N (\hat K-K)
	\leadsto \mathbb W$ in $\ell^\infty(\Theta)^\dparam$, where $\mathbb W$ is a Gaussian process defined in Lemma \ref{cor:LMdist}.
By Lemmas \ref{lem:H_hausdorff}-\ref{lem:H_dot}, Condition $Z$ in \cite{Chernozhukov+13}(CFM henceforth) holds, which in turn ensures that one may apply Lemmas E.2 and E.3 in CFM. This ensures
\begin{align}
\sqrt N(\varphi(\hat H,\hat r)-\varphi(H,0))\leadsto \varphi'_{H,0}(\mathbb W,0)=-\dot H_{\true}^{-1}\mathbb W(\true).
\end{align}
Hence, we obtain \eqref{eq:asynorm0a} with
\begin{align}
V=\dot H_{\true}^{-1}E[\mathbb W(\true)\mathbb W(\true)']\dot H_{\true}^{-1}.
\end{align}
Finally, note that $\dot H_{\true}=I_\dparam-J_K(\true)$ by Lemma \ref{lem:H_dot}.
This establishes the theorem.
\end{proof}

\begin{proof}[\rm \textbf{Proof of Theorem \ref{thm:bootstrap}}]
Recall that $\hat H= I_\dparam-\hat K$. The estimator $\est$  satisfies
\begin{align}
\|\hat H(\est)\|^2\le \inf_{\theta'\in\Theta}\|\hat H(\theta')\|^2+r_N^2,\label{eq:boot_cons1}
\end{align}
where $r_N=o_p(N^{-1/2})$. Similarly, let $\hat H^*= I_\dparam-\hat K^*$. Let $P^*$ denote the law of $\hat H^*$ conditional on $\{W_i\}_{i=1}^\infty.$ The bootstrap estimator $\est^*$  satisfies
\begin{align}
\|\hat H^*(\est^*)\|^2\le \inf_{\theta'\in\Theta}\|\hat H^*(\theta')\|^2+(r^*_N)^2,\label{eq:boot_cons2}
\end{align}
where $r^*_N=o_{P^*}(N^{-1/2})$ conditional on $\{W_i\}_{i=1}^\infty$.

Using the $r$-approximation, one may therefore write
\begin{align}
\sqrt N(\est^*-\est)&=\sqrt N(\varphi(\hat H^*,r^*_N)-\varphi(\hat H,r_N)).\label{eq:boot_cons3}
\end{align}
Let $E_{P^*}$ denote the conditional expectation with respect to $P^*$. Let $BL_1$ denote the space of bounded Lipschitz functions on $\mathbb{R}^{d_X+d_D}$ with Lipschitz constant 1. Then, for any $\epsilon>0$,
\begin{multline}
\sup_{h\in BL_1}\Big|E_{P^*}h\big(\sqrt N\big[\varphi(\hat H^*,r^*_N)-\varphi(\hat H,r_N)\big]\big)-E_{P^*}h\big(\varphi'_{H,0}\big(\sqrt N\big[(\hat H^*,r^*_N)'-(\hat H,r_N)'\big]\big)\big)\Big|\\
\le \epsilon +2P^*\Big(\big\|\sqrt N\big[\varphi(\hat H^*,r^*_N)-\varphi(\hat H,r_N)\big]-\varphi'_{H,0}\big(\sqrt N\big[(\hat H^*,r^*_N)-(\hat H,r_N)\big]\big)\big\|>\epsilon\Big).\label{eq:boot_cons4}
\end{multline}
By Lemma \ref{cor:LMdist}, $\sqrt N(\hat H^*-\hat H)=-\sqrt N(\hat K^*-\hat K)\stackrel{L^*}{\leadsto} -\mathbb W\stackrel{d}{=}\mathbb W$. Noting that $h\circ\varphi'_{H,0}\in BL_1(\ell^\infty(\Theta)\times\mathbb R)$ and $r_N=o_p(N^{-1/2})$, it follows that
\begin{align}
\sup_{h\in BL_1}\Big|E_{P^*}h\big(\varphi'_{H,0}\big(\sqrt N\big[(\hat H^*,r^*_N)-(\hat H,r_N)\big]\big)-E_{P^*}h\circ \varphi'_{H,0}(\mathbb W,0)\Big|\to 0,\label{eq:boot_cons5}
\end{align}
with probability approaching 1 due to $r_N=o_P(N^{-1/2})$. Hence, for the conclusion of the theorem, it suffices to show that the right hand side of \eqref{eq:boot_cons4} tends to 0 in probability.

For this, as shown in the proof of Theorem \ref{thm:asynorm}, $\varphi$ is Hadamard differentiable at $(H,0)$. Hence, by Theorem 3.9.4 in \cite{VanderVaartWellner1996},
\begin{align*}
\sqrt N\big[\varphi(\hat H^*,r^*_N)-\varphi(H,0)\big]&=\varphi'_{H,0}(\sqrt N[(\hat H^*,r^*_N)-(H,0)])+o_{P^*}(1) \\
\sqrt N\big[\varphi(\hat H,r_N)-\varphi(H,0)]&=\varphi'_{H,0}(\sqrt N[(\hat H,r_N)-(H,0)])+o_{P}(1),
\end{align*}
Take the difference of the left and right hand sides respectively and note that $\varphi'_{H,0}$ is linear. This implies the right hand side of \eqref{eq:boot_cons4} tends to 0 in probability. This ensures
\begin{align}
\sqrt N(\varphi(\hat H,r^*_N)-\varphi(\hat H,r_N))\stackrel{L^*}{\leadsto} \varphi'_{H,0}(\mathbb W,0)=-\dot H_{\true}^{-1}\mathbb W(\true).
\end{align}
\end{proof}

\begin{proof}[\rm \textbf{Proof of Corollary \ref{cor:asy_equivalence}}]
Note that $V$ and $g$ may be written as
\begin{align}
	V&=(I_\dparam-J_K(\true))^{-1}E[g(W;\true)g(W;\true)'][(I_\dparam-J_K(\true))^{-1}]'\\
	g(w;\theta^*)&=R^{-1}(\theta^*)f(w;\theta^*),\label{eq:defg}
\end{align}
where $R(\theta^*)$ is a $d_X+d_D$-by-$d_X+d_D$ matrix given  by
\begin{align}
	R(\theta^*)&=\begin{pmatrix}
		\frac{\partial^2}{\partial\theta_1\partial\theta_1'}Q_{P,1}(\theta^*)&0&\cdots&\cdots&0\\
		0&\frac{\partial^2}{\partial\theta_2\partial\theta_2'}Q_{P,2}(\theta^*)&0&\cdots&0\\
		\vdots&0&\ddots&&\vdots\\
		0 & \cdots & \cdots & & \frac{\partial^2}{\partial\theta_J\partial\theta_J'}Q_{P,J}(\theta^*)
	\end{pmatrix}\\
	&=
	\begin{pmatrix}
	\frac{\partial}{\partial\theta_1'}\Psi_{P,1}(\theta^*)&0&\cdots&\cdots&0\\
	0&\frac{\partial}{\partial\theta_2'}\Psi_{P,2}(\theta^*)&0&\cdots&0\\
	\vdots&0&\ddots&&\vdots\\
	0 & \cdots & \cdots & & \frac{\partial}{\partial\theta_J'}\Psi_{P,J}(\theta^*)	
	\end{pmatrix}.
\end{align}
Further, by Lemma \ref{lem:global_decentralization} and the form of $J_{L_{-j}}(\theta^*_{j})$ given in \eqref{eq:def_JL},
\begin{align}
	J_K(\param^*)
	&=
	\begin{pmatrix}
		0_{d_X\times d_X}&\frac{\partial L_1(\true)}{\partial\theta_2'}&\cdots &\frac{\partial L_1(\true)}{\partial\theta_J'}\\
		\frac{\partial L_2(\true)}{\partial\theta_1'}&0&\cdots&\frac{\partial L_2(\true)}{\partial\theta_J'}\\
		\vdots&&\ddots&\vdots\\
		\frac{\partial L_J(\true)}{\partial\theta_1'}&\frac{\partial L_J(\true)}{\partial\theta_2'}&\cdots&0
	\end{pmatrix}\\
	&=\begin{pmatrix}
		0_{d_X\times d_X}& -\left(\frac{\partial\Z_{P,1}(\theta^*)}{\partial\theta_{1}'} \right)^{-1}\frac{\partial\Z_{P,1}(\theta*)}{\partial\theta_2'}&\cdots&-\left(\frac{\partial\Z_{P,1}(\theta^*)}{\partial\theta_{1}'} \right)^{-1}\frac{\partial\Z_{P,1}(\theta^*)}{\partial\theta_J'}\\
		-\left(\frac{\partial\Z_{P,2}(\theta^*)}{\partial\theta_{2}'} \right)^{-1}\frac{\partial\Z_{P,2}(\theta^*)}{\partial\theta_1'} &0&\cdots&-\left(\frac{\partial\Z_{P,2}(\theta^*)}{\partial\theta_{2}'} \right)^{-1}\frac{\partial\Z_{P,2}(\theta^*)}{\partial\theta_J'}\\
\vdots&&\ddots&\vdots\\
		-\left(\frac{\partial\Z_{P,J}(\theta^*)}{\partial\theta_{J}'} \right)^{-1}\frac{\partial\Z_{P,J}(\theta^*)}{\partial\theta_1'} &-\left(\frac{\partial\Z_{P,J}(\theta^*)}{\partial\theta_{J}'} \right)^{-1}\frac{\partial\Z_{P,2}(\theta^*)}{\partial\theta_2'}	&\cdots&0
	\end{pmatrix}.
\end{align}
The form of $R(\true)$ and $J_K(\true)$ imply
\begin{align}
	R(\true)(I_\dparam-J_K(\true))=J_{\Psi_P}(\theta^*),\label{eq:recover_Jpsi}
\end{align}
where $J_{\Psi_P}$ is the Jacobian of the estimating equations. Eq. \eqref{eq:asyvar}-\eqref{eq:defg} and \eqref{eq:recover_Jpsi} ensure that
one may also write
\begin{align}
	V=J_{\Psi_P}(\theta^*)^{-1}E[f(W;\theta^*)f(W;\theta^*)'][J_{\Psi_P}(\theta^*)^{-1}]'.\label{eq:asyvar1}
\end{align}
As shown in \cite{CH2006}, the Jacobian of $\Psi_P$ is given by
\begin{align}
	J_{\Psi_P}(\theta^*)=E[f_{\varepsilon(\tau)|X,D,Z}(0)\Psi(\tau) [X',D']],\label{eq:asyvar2}
\end{align}
 where $\Psi(\tau)=(X',Z')'.$ Furthermore,
 \begin{align}
 	E[f(W;\theta^*)f(W;\theta^*)']=\tau(1-\tau)E[\Psi(\tau)\Psi(\tau)'].\label{eq:asyvar3}
 \end{align}
Hence, \eqref{eq:asyvar1}-\eqref{eq:asyvar3} show that $V$ coincides with the asymptotic variance of the estimator that solves the estimating equations in \eqref{eq:est_eq_tilde}.
\end{proof}

\begin{lemma}\label{lem:asy_equiv}
	Suppose Assumptions \ref{ass:ivqr}-\ref{ass:global_decentralization} hold. (i) Let $\est$ be an estimator of $\true$ that satisfies \eqref{eq:hattheta}. Then, it also satisfies \eqref{eq:hattheta1}-\eqref{eq:hattheta2};  (ii) Let $\est$ be an estimator of $\true$ that satisfies \eqref{eq:hattheta1}-\eqref{eq:hattheta2}. Then, it also satisfies 
\begin{align}
\|\est-\hat K(\est)\|\le \inf_{\theta'\in\Theta}\|\theta'-\hat K(\theta')\|+o_p(N^{-1/2}).\label{eq:hattheta}
\end{align}	
\end{lemma}

\begin{proof}
(i) Consider the case  $j=2$. Note that, by \eqref{eq:hattheta},
\begin{align}
\hat\theta_{N,2}-\hat \br_2(\hat\br_1(\hat\theta_{N,-1}),\hat\theta_{N,3},\dots,\hat\theta_{N,\dpl})&=\hat\theta_{N,2}-\hat\br_2(\hat\theta_{N,1}+r_{N,1},\hat\theta_{N,3},\dots,\hat\theta_{N,\dpl})\label{eq:asy_equiv1}\\
&=\hat\br_2(\hat\theta_{N,1},\hat\theta_{N,3},\dots,\hat\theta_{N,\dpl})-\hat\br_2(\hat\theta_{N,1}+r_{N,1},\hat\theta_{N,3},\dots,\hat\theta_{N,\dpl}),\label{eq:asy_equiv2}
\end{align}
where $r_{N,1}=o_p(N^{-1/2})$, and the second equality follows from the definition of $\hat\theta_{N,2}$. \eqref{eq:asy_equiv2} can be written as
\begin{align}
\hat\br_2(\hat\theta_{N,1},\hat\theta_{N,3},\dots,\hat\theta_{N,\dpl})&-\hat\br_2(\hat\theta_{N,1}+r_{N,1},\hat\theta_{N,3},\dots,\hat\theta_{N,\dpl})\notag\\
&=\Big([\hat\br_2(\hat\theta_{N,1},\hat\theta_{N,3},\dots,\hat\theta_{N,\dpl})-\br_2(\hat\theta_{N,1},\hat\theta_{N,3},\dots,\hat\theta_{N,\dpl})]\notag\\
&\qquad-[\hat\br_2(\hat\theta_{N,1}+r_{N,1},\hat\theta_{N,3},\dots,\hat\theta_{N,\dpl})-\br_2(\hat\theta_{N,1}+r_{N,1},\hat\theta_{N,3},\dots,\hat\theta_{N,\dpl})]\Big)\notag\\
&\qquad+[\br_2(\hat\theta_{N,1}+r_{N,1},\hat\theta_{N,3},\dots,\hat\theta_{N,\dpl})-\br_2(\hat\theta_{N,1},\hat\theta_{N,3},\dots,\hat\theta_{N,\dpl})]\notag\\
&=o_p(N^{-1/2})+O_P(r_{N,1}),\label{eq:asy_equiv3}
\end{align}
where the last equality follows from the stochastic equicontinuity of $\mathcal L_N$ shown in the proof of Lemma \ref{lem:sequi} and $L_2$ being Lipschitz  since $L_2$ is continuously differentiable with a derivative that is uniformly bounded on the compact set $\Theta$. By \eqref{eq:asy_equiv1}-\eqref{eq:asy_equiv3}, it holds that $\hat\theta_{N,j}=M_j(\hat\theta_{N,-1})+o_p(N^{-1/2})$ for $j=2$. Repeat the same argument sequentially for $j=3,\dots,\dpl$. The first conclusion of the lemma then follows.

(ii) Suppose now that $r_{N,1}\equiv \hat\theta_{N,1}-\hat\br_1(\hat \theta_{N,-1})\ne o_P(N^{-1/2}).$ Then, there is a subsequence $k_N$ along which, for any $\eta>0$, $\sqrt k_Nr_{k_N,1}>\eta$ for all $k_N$  with positive probability.
Then, the $O_P(r_{k_N,1})$-term in \eqref{eq:asy_equiv3} is not $o_p(k_N^{-1/2})$, which therefore implies  $\hat\theta_{N,j}\ne M_j(\hat\theta_{N,-1})+o_p(N^{-1/2})$ for $j=2$. The second conclusion of the lemma then follows.
\end{proof}

\begin{lemma}\label{lem:H_hausdorff}
Let $\Lambda\subset\mathbb R^p$ be a compact set, and let $\EL:\Lambda\to\mathbb R^p$ be a map that has a unique fixed point $\lambda_0\in\Lambda$.
let $\G:\Lambda\to\mathbb R^p$ be defined by $\G(\lambda)\equiv\lambda-\EL(\lambda)$.
 Then $\G^{-1}(x)=\{\lambda\in\Lambda:\G(\lambda)=x\}$ is continuous at $x=0$ in Hausdorff distance.
\end{lemma}

\begin{proof}
For any $x$, write
\begin{align*}
\G^{-1}(x)=\{\lambda:\lambda-\EL(\lambda)=x\}.
\end{align*}
Let $x_n\to 0$.
Since $\lambda_0$ is the  unique fixed point of $\EL$, $\G^{-1}(0)=\{\lambda_0\}$. Therefore,
\begin{align*}
d_H(\G^{-1}(0),\G^{-1}(x_n))&=\max\left\{\inf_{\lambda\in \G^{-1}(x_n)}\|\lambda-\lambda_0\|,\sup_{\lambda\in \G^{-1}(x_n) }\|\lambda-\lambda_0\|\right\}\\
&=\sup_{\lambda\in \G^{-1}(x_n) }\|\lambda-\lambda_0\|.
\end{align*}
Hence, it suffices to show that $\sup_{\lambda\in \G^{-1}(x_n) }\|\lambda-\lambda_0\|=o(1)$. We show this by contradiction. Suppose that there is a sequence $\{\lambda_n\}\subset\Lambda$ and $\delta>0$ such that $\lambda_n\in \G^{-1}(x_n)$ for all $n$ and $\{\lambda_n\}$ has a subsequence $\{\lambda_{k_n}\}$ such that
$\|\lambda_{k_n}-\lambda_0\|>\delta$ for all $n$. $\lambda_{k_n}\in\Lambda$ is a sequence in a compact space, and hence there is a further subsequence $\lambda_{h_n}$ such that $\lambda_{h_n}\to \lambda^*$ for some $\lambda^*\in\Lambda$ with $\lambda^*\ne \lambda_0$. By the continuity of $\EL$, one then has
\begin{align*}
\lambda_{h_n}-\EL(\lambda_{h_n})\to \lambda^*-\EL(\lambda^*).
\end{align*}
By $\lambda_{h_n}-\EL(\lambda_{h_n})=x_n$ and $x_n\to0$, it must hold that
\begin{align*}
\lambda^*-\EL(\lambda^*)=0.
\end{align*}
However this contradicts the fact that $\lambda_0$ is the unique fixed point, and hence the conclusion follows.
\end{proof}

\begin{lemma}\label{lem:H_dot}
Suppose $\G=I-\EL$ and $\EL:\mathbb R^p\to\mathbb R^p$ is continuously differentiable at $\lambda_0$.
Suppose further that $\det(I-J_{\EL}(\lambda_0))\ne 0$. Let $\dot \G_{\lambda_0}\equiv I-J_{\EL}(\lambda_0)$.
Then,
\begin{align*}
\lim_{t\downarrow 0}\sup_{h:\|h\|=1}\|t^{-1}[\G(\lambda_0+th)-\G(\lambda_0)]-\dot \G_{\lambda_0}h\|=0,
\end{align*}
and
\begin{align*}
\inf_{h:\|h\|=1}\|\dot \G_{\lambda_0}h\|>0.
\end{align*}
\end{lemma}

\begin{proof}
Let $\{h_n\}\subset\mathbb S^p$ be a sequence on the unit sphere.  Then,
\begin{align*}
t^{-1}[\G(\lambda_0+th_n)-\G(\lambda_0)]-\dot \G_{\lambda_0}h_n&=t^{-1}[\lambda_0+th_n+\EL(\lambda_0+th_n)-\lambda_0-\EL(\lambda_0)]-h_n-J_{\EL}(\lambda_0) h_n\\
&=t^{-1}[\EL(\lambda_0+th_n)-\EL(\lambda_0)]-J_{\EL}(\lambda_0) h_n\\
&=(J_{\EL}(\bar\lambda_n)-J_{\EL}(\lambda_0))h_n,
\end{align*}
where $\bar\lambda_n$ is a mean value between $\lambda_0+th_n$ and $\lambda_0$. Therefore, by the Cauchy-Schwarz inequality,
\begin{align*}
\|(J_{\EL}(\bar\lambda_n)-J_{\EL}(\lambda_0))h_n\|\le \|J_{\EL}(\bar\lambda_n)-J_{\EL}(\lambda_0)\|\|h_n\|\to0,
\end{align*}
where we used $\|h_n\|=1$,  $\bar\lambda_n\to \lambda_0$, and the continuity of the Jacobian.

For the second claim, note that
\begin{align*}
\|\dot \G_{\lambda_0}h\|=\|(I-J_{\EL}(\lambda_0)) h\|,
\end{align*}
and $h\mapsto \|(I-J_{\EL}(\lambda_0)) h\|$ is continuous. Since the domain of $h$ is compact, there is $h^*\in \mathbb S^p$ such that $\inf_{\|h\|=1}\|\dot \G_{\lambda_0}h\|=\|(I-J_{\EL}(\lambda_0))h^*\|.$ Let $q=(I-J_{\EL}(\lambda_0))h^*$ and note that $I-J_{\EL}(\lambda_0)$ is linearly independent (due to $\det(I-J_{\EL}(\lambda_0))\ne 0$), and hence $q\ne 0$. Hence $\inf_{\|h\|=1}\|\dot \G_{\lambda_0}h\|=\|q\|>0$. Hence, the second conclusion follows.
\end{proof}

The following result is a slight extension of Lemma E.1 in CFM.
\begin{lemma}\label{lem:psi_cond}
	Suppose that $\Lambda\subset\mathbb R^p$ and $\mathcal U$ is a compact and convex set in $\mathbb R^q$. Let $\mathcal I$ be an open set containing $\mathcal U$. Suppose that $\Z:\Lambda\times\mathcal I\to\mathbb R^p$ is continuous and $\lambda\mapsto\Z(\lambda,u)$ is the gradient of a convex function in $\lambda$ for each $u\in\mathcal U$; (b) for each $u\in\mathcal U$, $\Z(\lambda_0(u),u)=0$; (c) $\frac{\partial}{\partial (\lambda',u')}\Z(\lambda,u)$ exists at $(\lambda_0(u),u)$ and is continuous at $(\lambda_0(u),u)$ for each $u\in\mathcal U$ and $\dot\Z_{\lambda_0(u),u}:=\frac{\partial}{\partial\lambda'}\Z(\lambda,u)|_{\lambda_0(u)}$ obeys $\inf_{u\in\mathcal U}\inf_{\|h\|=1}\|\dot\Z_{\lambda_0(u),u}h\|>c_0>0$. Then, Condition $Z$ in CFM holds and $u\mapsto\lambda_0(u)$ is continuously differentiable with derivative $J_{\lambda_0}(u)=-\dot \Z^{-1}_{\lambda_0(u)u}\frac{\partial}{\partial u'}\Z(\lambda_0(u),u)$.
\end{lemma}

\begin{proof}
	The proof is the same as that of Lemma E.1 in CFM, in which $\mathcal U$ is a compact interval in $\mathbb R$.
	A slight modification is needed when one computes the derivative of $\lambda_0(u)$ with respect to $u$. Since $u$ is allowed to be multidimensional, the implicit function theorem gives
	\begin{align}
	J_{\lambda_0}(u)=-\dot \Z^{-1}_{\lambda_0(u)u}\frac{\partial}{\partial u'}\Z(\lambda_0(u),u),
	\end{align}
which is uniformly bounded and continuous in $u$ by condition (c), which ensures continuous differentiability of $u\mapsto\lambda_0(u)$.
Note that for any $\delta>0$ and $\lambda\in B_{\delta}(\lambda_0(u))$, there is $\eta>0$  and $u'$ such that $\|u'-u\|\le \eta$ so that
\begin{align}
\|\lambda-\lambda_0(u')\|\le \|\lambda-\lambda_0(u)\|+\|\lambda_0(u)-\lambda_0(u')\|\le 2\delta.
\end{align}
Since $\mathcal U$ is compact (and hence totally bounded), there is a finite set $\{u_j\}_{j=1}^J\subset\mathcal U$ such that $\mathcal U\subset\bigcup_j B_{\eta}(u_j)$. The argument above then shows that $\mathcal N=\bigcup_{u\in\mathcal U}B_\delta(\lambda_0(u))\subset\bigcup_{j}B_{2\delta}(\lambda_0(u_j))$, which ensures that $\mathcal N$ is totally bounded. Since $\mathcal N$ is a subset of a Euclidean space (equipped with a complete metric), it follows that $\mathcal N$ is compact. This ensures condition $Z$ (i) in CFM. The rest of the proof is essentially the same as the case, in which $\mathcal U$ being a compact interval.
\end{proof}

\begin{lemma}\label{lem:f_donsker}
Suppose Assumption \ref{ass:global_decentralization} holds.	Let $w=(y,d',x',z')$ and let $\tau\in(0,1)$. Define
	\begin{multline}
	\mathcal M\equiv\Big\{f:f(w;\theta)=\big((1\{ y\leq d'\theta_{-1}+x'\theta_1 \} -\tau)x,\\
	(1\{ y\leq d'\theta_{-1}+x'\theta_1 \} -\tau)z_1,\dots,(1\{ u\leq d'\theta_{-1}+x'\theta_{1} \} -\tau )z_{d_D}\big),\theta\in\Theta\Big\}.
	\end{multline}
Then, $\mathcal M$ is a Donsker-class.
\end{lemma}

\begin{proof}
	The proof is standard, and hence we give a brief sketch for the first component of $f$, $f_1(w;\theta)=(1\{ y\leq d'\theta_{-1}+x'\theta_1 \} -\tau)x$.
Note that $w\mapsto 1\{ y\leq d'\theta_{-1}+x'\theta_1 \}-\tau$ belongs to  Type I-class in \cite{Andrews:1994aa}, and the map $w\mapsto x$ does not depend on the parameter. By Theorems 2 and 3 in \cite{Andrews:1994aa}, this function then satisfies the uniform entropy condition with the envelope function $\bar M(w)=x$, which is square integrable by assumption. Similar arguments apply to the other components of $f$.
By Theorem 1 in \cite{Andrews:1994aa}, the empirical process: $\mathbb G_nf$ is stochastically equicontinuous, and $\mathbb G_nf(\cdot,\theta)$ obeys the classical central limit theorem for each $\theta\in\Theta$.  Hence, we conclude that $\mathcal M$ is Donsker.
\end{proof}

Below, let $g(w;\theta)=(g_1(w;\theta)',\dots,g_J(w;\theta))'$ be a vector such that
\begin{align}
	g_j(w;\theta)=\left(\frac{\partial^2}{\partial\theta_j\partial\theta_j'}Q_{P,j}(L_j(\theta_{-j}),\theta_{-j})\right)^{-1}f_j(w;L_j(\theta_{-j}),\theta_{-j}),~j=1,\dots,J.
\end{align}
Let  $\rho(\theta,\tilde\theta)\equiv \big\|\text{diag}\big( E_P\big[(g(W;\theta)-E_P[g(W;\theta)])(g(w;\tilde\theta)- E_P[g(w;\tilde\theta)])'\big]\big)\big\|$ be the variance semimetric.
Let $W_i=(Y_i,D_i',X_i',Z_i'),i=1,\dots,N$ be an i.i.d. sample generated from the IVQR model.  Define
\begin{align}
\mathcal L_{N,j}(\theta_{-j})\equiv \sqrt N(\hat \br_{j}(\theta_{-j})-\br_j(\theta_{-j}))~,j=1,\dots,J.\label{eq:br_emp_process}
\end{align}
Similarly, let $W_i^*=(Y_i^*,D_i^{*\prime},X^{*\prime}_i,Z^{*\prime}_i)',i=1,\dots,N$ be an bootstrap sample from the empirical distribution $\hat P_N$ of $\{W_i\}$. Define
\begin{align}
\mathcal L_{N,j}^*(\theta_{-j})\equiv \sqrt N(\hat \br^*_{j}(\theta_{-j})-\hat \br_j(\theta_{-j}))~,j=1,\dots,J,
\end{align}
where $\hat \br^*_{j}$ is the sample best response map of player $j$, which is defined as in \eqref{eq:brN1}-\eqref{eq:brNj} while replacing $W_i$ with the bootstrap sample $W_i^*$
in \eqref{eq:QN1}-\eqref{eq:QNj}.

Lemma \ref{lem:sample_BR} below shows that the sample BR functions approximately solve sample estimating equations and Lemma \ref{lem:sequi} characterizes the limiting distributions of $\mathcal L_N$ and $\mathcal L_N^*$.

\begin{lemma}\label{lem:sample_BR}
Let the sample BR functions be $\hat L_j(\theta_{-j})\in \textrm{argmin}_{\tilde\theta_j}Q_{N,j}(\tilde\theta_j,\theta_{-j}),j=1,\dots,J$. Let $\hat L_j(\theta_{-j})^*$ be an analog of $\hat L_j(\theta_{-j})$ for the bootstrap sample. Then, (i) the sample BR functions satisfy
\begin{multline}
	\Big|\frac{1}{N}\sum_{i=1}^N\big(1\{Y_i\le D_{i}'\theta_{-1}+X_i'\hat L_1(\theta_{-1})\}-\tau\big)X_{i}\Big|^2\\
\le \inf_{\theta_1\in \Theta_{1}}\Big|\frac{1}{N}\sum_{i=1}^N\big(1\{Y_i\le D_{i}'\theta_{-1}+X_i'\theta_1\}-\tau\big)X_{i}\Big|^2+r^2_{N,1}(\theta_{-1}),\label{eq:sampleBR1}
\end{multline}
and
\begin{multline}
	\Big|\frac{1}{N}\sum_{i=1}^N\big(1\{Y_i\le (X'_i,D'_{i,-(j-1)})'\theta_{-j}+D_{i,j}'\hat L_j(\theta_{-j})\}-\tau\big)Z_{i,j}\Big|^2\\
\le \inf_{\theta_j\in \Theta_{j}}\Big|\frac{1}{N}\sum_{i=1}^N\big(1\{Y_i\le (X'_i,D'_{i,-(j-1)})'\theta_{-j}+D_{i,j}'\hat L_j(\theta_{-j})\}-\tau\big)Z_{i,j}\Big|^2+r^2_{N,j}(\theta_{-j}),~j=2,\dots,J,\label{eq:sampleBR2}
\end{multline}
where $\sup_{\theta_{-j}\in\Theta_{-j}}|r_{N,j}(\theta_{-j})|=o_P(N^{-1/2})$ for all $j$;
(ii) the sample BR functions $\hat L_j^*(\theta_{-j}),j=1,\dots,J$ in the bootstrap sample satisfy \eqref{eq:sampleBR1}-\eqref{eq:sampleBR2} while replacing $(Y_i,D_i,X_i,Z_i)$ with a bootstrap sample $(Y_i^*,D_i^*,X_i^*,Z_i^*)$, each $\hat L_j$ with $\hat L_j^*$, and  each $r_{N,j}$ with $r^*_{N,j}$ such that $\sup_{\theta_{-j}\in\Theta_{-j}}|r^*_{N,j}(\theta_{-j})|=o_{P^*}(N^{-1/2})$.
\end{lemma}
\begin{proof}
(i)  For $j\ge 2,$ the subgradient of $Q_{N,j}$ is
\begin{align}
	\xi_j=\frac{1}{N}\sum_{i=1}^N\big(1\{Y_i\le(X'_i,D'_{i,-(j-1)})'\theta_{-j}+ D_{i,j-1}'\hat L_j(\theta_{-j})\}-\tau\big)Z_{i,j-1},
\end{align}
and hence by the property of the subgradient, for any $v\in\mathbb R$, one has
\begin{align}
	\xi_j v\le \nabla_{\theta_j} Q_{N,j}(\hat L_j(\theta_{-j}),\theta_{-j},v),
\end{align}
where $\nabla_{\theta_j} Q_{N,j}(\hat L_j(\theta_{-j}),\theta_{-j},v)$ is the directional derivative of  $Q_{N,j}(\theta_j,\theta_{-j})$ with respect to $\theta_j$ toward direction $v\in\mathbb R$ evaluated at $(\hat L_j(\theta_{-j}),\theta_{-j}).$ Note that the directional derivative is given by
\begin{align}
\nabla_{\theta_j} Q_{N,j}(\hat L_j(\theta_{-j}),\theta_{-j},v)=	-\frac{1}{N}\sum_{i=1}^N\psi^*_\tau(Y_i-(X'_i,D'_{i,-(j-1)})'\theta_{-j}- D_{i,j-1}'\hat L_j(\theta_{-j}),-Z_{i,j-1}v)Z_{i,j-1}v,
\end{align}
where 
\begin{align}\psi^*_\tau(u,w)=\begin{cases}
	\tau-1\{u<0\}& u\ne 0\\
	\tau-1\{w<0\}& u=0.
\end{cases}
\end{align}
Observe that $-\nabla_{\theta_j} Q_{N,j}(\hat L_j(\theta_{-j}),\theta_{-j},-v)\le \xi v\le \nabla_{\theta_j} Q_{N,j}(\hat L_j(\theta_{-j}),\theta_{-j},v)$. This implies
\begin{align}
	|\xi_j v|&\le \nabla_{\theta_j} Q_{N,j}(\hat L_j(\theta_{-j}),\theta_{-j},v)-(-\nabla_{\theta_j} Q_{N,j}(\hat L_j(\theta_{-j}),\theta_{-j},-v))\notag\\
	&= \frac{1}{N}\sum_{i=1}^N\Big(-\psi^*_\tau(Y_i-(X'_i,D'_{i,-(j-1)})'\theta_{-j} -D_{i,j}'\hat L_j(\theta_{-j}),-Z_{i,j-1}v)\notag\\
	&\qquad\qquad +\psi^*_\tau(Y_i-(X'_i,D'_{i,-(j-1)})'\theta_{-j}- D_{i,j}'\hat L_j(\theta_{-j}),Z_{i,j-1}v)\Big)Z_{i,j-1}v\notag\\
	&= \frac{1}{N}\sum_{i=1}^N1\{Y_i=(X'_i,D'_{i,-(j-1)})'\theta_{-j}+D_{i,j}'\hat L_j(\theta_{-j})\}\text{sgn}(Z_{i,j-1}v)Z_{i,j-1}v\notag\\
	&= \frac{1}{N}\sum_{i=1}^N1\{Y_i=(X'_i,D'_{i,-(j-1)})'\theta_{-j}+D_{i,j}'\hat L_j(\theta_{-j})\}|Z_{i,j-1}v|\notag\\
	&\le \Big(\sum_{i=1}^N1\{Y_i=(X'_i,D'_{i,-(j-1)})'\theta_{-j}+D_{i,j}'\hat L_j(\theta_{-j})\}\Big)\max_{i=1,\dots,N}\frac{|Z_{i,j-1}v|}{N}.
\end{align}
Noting that $\sum_{i=1}^N1\{Y_i=(X'_i,D'_{i,-(j-1)})'\theta_{-j}+D_{i,j}'\hat L_j(\theta_{-j})\}=\text{dim}(\theta_j)=1$ and taking $v=1$, we obtain,
\begin{align}
	\Big|\frac{1}{N}\sum_{i=1}^N\big(1\{Y_i\le (X'_i,D'_{i,-(j-1)})'\theta_{-j}+D_{i,j}'\hat L_j(\theta_{-j})\}-\tau\big)Z_{i,j}\Big|
	\le \max_{i=1,\dots,N}\frac{|Z_{i,j-1}|}{N}=o_P(N^{-1/2}),
\end{align}
uniformly in $\theta_{-j}$, where the last equality is due to $E[|Z_{i,j-1}|^2]<\infty$ by Assumption 2.2. Therefore, for some $r_{N,j}$ satisfying the assumption of the lemma, we may write
\begin{multline}
	\Big|\frac{1}{N}\sum_{i=1}^N\big(1\{Y_i\le (X'_i,D'_{i,-(j-1)})'\theta_{-j}+D_{i,j}'\hat L_j(\theta_{-j})\}-\tau\big)Z_{i,j}\Big|^2\le r^2_{N,j}(\theta_{-j})\\
\le \inf_{\theta_j\in \Theta_{j}}\Big|\frac{1}{N}\sum_{i=1}^N\big(1\{Y_i\le (X'_i,D'_{i,-(j-1)})'\theta_{-j}+D_{i,j}'\theta_j\}-\tau\big)Z_{i,j}\Big|^2+r^2_{N,j}(\theta_{-j}).
\end{multline}
 The proof for $j=1$ is similar. Also, (ii) can be shown by mimicking the argument above.	
\end{proof}

\begin{lemma}\label{lem:sequi}
Suppose that Assumptions \ref{ass:ivqr} and \ref{ass:global_decentralization} hold. Then,  (i)
 $\mathcal L_N\equiv(\mathcal L_{N,1},\dots,\mathcal L_{N,J})$ defined in \eqref{eq:br_emp_process} satisfies
\begin{align}
\mathcal L_{N}(\cdot)\leadsto \mathbb W,
\end{align}
where $\mathbb W$ is a tight Gaussian process in $\ell^\infty(\Theta)^\dparam$ with the covariance kernel
\begin{align}
	\emph{Cov}(\mathbb W(\theta),\mathbb W(\tilde\theta))=E_P\big[(g(W;\theta)- E_P[g(W;\theta)])(g(W;\tilde\theta)-E_P[g(W;\tilde\theta)])'\big];\label{eq:cov_kernel}
\end{align} $\mathcal L_N$
is stochastically equicontinuous with respect to the variance semimetric $\rho$;  (ii) $\mathcal L^*_N\equiv(\mathcal L^*_{N,1},\dots,\mathcal L^*_{N,J})$ satisfies
\begin{align}
\mathcal L^*_{N}(\cdot)\stackrel{L^*}{\leadsto} \mathbb W;
\end{align}
(iii) $\rho$ satisfies
$\lim_{\delta\downarrow 0}\sup_{\|\theta-\tilde\theta\|< \delta}\rho(\theta,\tilde\theta)\to 0.$
\end{lemma}
\begin{proof}
(i) We first work with $\mathcal L_{N,1}$.
For this, we establish that $L_1$ is Hadamard differentiable.
Note that $\theta_1=L_1(\theta_{-1})$ solves
\begin{align}
E_P[ ( 1\{ Y\leq D'\theta_{-1}+X'\theta_1 \} -\tau)X ]=0.	\label{eq:LMdist1}
\end{align}
Take $\mathcal U=\Theta_{-1}$, $\Xi=\Theta_1$, $\psi(\lambda,u)=E_P[( 1\{ Y\leq Du+X'\lambda \} -\tau)X ]$. Define $\phi:\ell^\infty(\Xi\times \mathcal U)^{k_b}\times \ell^\infty(\mathcal U)\to \ell^\infty(\mathcal U)$, which maps $(\psi,r)$ to a solution $\phi(\psi,r)=\lambda(\cdot)$ such that
\begin{align}
\|\psi(\lambda(u),u)\|^2\le \inf_{\lambda'\in\Theta}\|\psi(\lambda',u)\|^2+r(u)^2.	\label{eq:LMdist2}
\end{align}
Then, one may write $L_1(\cdot)=\phi(\psi,0).$ We then show that $\psi$ satisfies the conditions of Lemma \ref{lem:psi_cond}. Note first that $\mathcal U$ and  $\Xi$ are compact.
 $\psi$ is continuous and $\lambda\mapsto\psi(\lambda,u)$ is the gradient of the convex function $\lambda\mapsto E_P\left[\rho_\tau( Y-Du-X'\lambda)\right].$ The function $L_1(u)=\lambda_0(u)$  is defined as the exact solution of $\psi(\lambda,u)=0$.
Note also that, by Assumption \ref{ass:global_decentralization},
\begin{align}
	\frac{\partial^2}{\partial\theta_1\partial\theta_1'}Q_{P,1}(\theta_1,\theta_{-1})&=\frac{\partial}{\partial\theta_1'}E_P[ ( 1\{ Y\leq D'\theta_{-1}+X'\theta_1 \} -\tau)X ]\notag\\
	&=E_P[ \frac{\partial}{\partial\theta_1'}(F_{Y|D,X,Z}(D'\theta_{-1}+X'\theta_1) -\tau)X ]\notag\\
	&=E_P[f_{Y|D,X,Z}(D'\theta_{-1}+X'\theta_1)XX'],\label{eq:hessian1}
\end{align}
where the second equality follows from the dominated convergence theorem, and the last display is well-defined by the square integrability of $X$.
Similarly,
\begin{align}
\frac{\partial^2}{\partial\theta_1\partial\theta_{-1}'}Q_{P,1}(\theta_1,\theta_{-1})&=E_P[f_{Y|D,X,Z}(D'\theta_{-1}+X'\theta_1)XD'].
\end{align}
 Hence,  the derivative $$\frac{\partial}{\partial (\lambda',u')}\Psi(\lambda,u)=(\frac{\partial^2}{\partial\theta_1\partial\theta_{1}'}Q_{P,1}(\theta_1,\theta_{-1}),\frac{\partial^2}{\partial\theta_1\partial\theta_{-1}'}Q_{P,1}(\theta_1,\theta_{-1}))$$ exists and is continuous by  Assumption \ref{ass:global_decentralization}. By  Assumption \ref{ass:global_decentralization}.\ref{ass:global_decentralization4}, $\dot\Psi_{\lambda_0(u),u}=\frac{\partial^2}{\partial\theta_1\partial\theta_{1}'}Q_{P,1}(L_1(\theta_{-1}),\theta_{-1})$ obeys
\begin{align}
\inf_{u\in\mathcal U}\inf_{\|h\|=1}\|\dot\Psi_{\lambda_0(u),u}h\|=\inf_{\theta_{-1}\in \Theta_{-1}}\inf_{\|h\|=1}\|E_P[f_{Y|D,X,Z}(D'\theta_{-1}+X'\theta_1)XX']h\|>0.	\label{eq:LMdist3}
\end{align}
Then, by Lemma \ref{lem:psi_cond} and Lemma E.2 in CFM, $\phi$ is Hadamard differentiable tangentially to $\mathcal C(\mathcal N\times\mathcal U)^{K}\times\{0\}$ with the Hadamard derivative (of $L_1$)
\begin{align}
\phi'_{\Psi,0}(z,0)=-\left(\frac{\partial^2}{\partial\theta_1\partial\theta_1'}Q_{P,1}(L_1(\cdot),\cdot)\right)^{-1}z(L_1(\cdot),\cdot),\label{eq:phidot_B}
\end{align}
where $(z,0)\mapsto \phi'_{\Psi,0}(z,0)$ is continuous over $z\in \ell^\infty(\Theta)^{K}$.

For $j\ge 2$, the argument is similar. For example, for $j=2$, one may  take $\mathcal U=\Theta_{-2}$, $\Xi=\Theta_2$ and $\psi(\lambda,u)=E_P[ ( 1\{ Y\leq D_2\theta_2+(D_1,X)'u \} -\tau)Z_2]$ and write $L_2(\cdot)=\phi(\psi,0)$. The rest of the argument is the same.

Continuing with $j=1$, 
by Lemma \ref{lem:sample_BR}, one may write $\hat L_j(\cdot)=\phi(\psi_{N},r_{N,1})$ with $\psi_N(\lambda,u)=\frac{1}{N}\sum_{i=1}^N1\{Y_i\le D'_iu+X'_i\lambda\}X_i$ and $\sup_{\theta_{-1}\in\Theta_{-1}}|r_{N,1}(\theta_{-1})|=o_p(N^{-1/2}).$  By Lemma \ref{lem:f_donsker}  and applying the $\delta$-method (as in Lemma E.3 in CFM), we obtain
\begin{align}
\mathcal L_{N}(\cdot)\leadsto \mathbb W,
\end{align}
where $\mathbb W=(\mathbb W_1',\dots,\mathbb W_J')'$ is a tight Gaussian process in $\ell^\infty(\Theta)^\dparam$, where for each $j$,  $\mathbb W_j\in\ell^\infty(\Theta_{-j})^{d_j}$ is given pointwise by
\begin{align}
\mathbb W_j(\theta_{-j})&=-\left(\frac{\partial^2}{\partial\theta_j\partial\theta_j'}Q_{P,j}(L_j(\theta_{-j}),\theta_{-j})\right)^{-1}\mathbb Gf_j(w;L_j(\theta_{-j}),\theta_{-j}),~j=1,\cdots,\dpl;
\end{align}
Hence, its  covariance kernel is as given in \eqref{eq:cov_kernel}.
By Lemma 1.3.8. in \cite{VanderVaartWellner1996}, $\{\mathcal L_{N}\}$ is asymptotically tight, which in turn means that $\{\mathcal L_{N}\}$ is stochastically equicontinuous with respect to $\rho$ by Theorem 1.5.7 in \cite{VanderVaartWellner1996}.

\bigskip
\noindent
(ii) For each $j$, let $\mathcal L^*_{N,j}\in \ell^\infty(\Theta_{-j})^{d_j}$ be defined pointwise by
\begin{align}
  \mathcal L^*_{N,j}(\theta_{-j})=\sqrt N(\hat L^*_j(\theta_{-j}))-\hat L_j(\theta_{-j})).
\end{align}
Below, again we work with the case $j=1$. Using $\phi$ (the solution to \eqref{eq:LMdist2}) and applying Lemma \ref{lem:sample_BR}, we may write
\begin{align}
  \mathcal L^*_{N,1}(\theta_{-1})=\sqrt N(\phi(\hat \psi^*_N,r^*_N)-\phi(\hat \psi_N,r_N)),
\end{align}
where $\hat \psi_N(\lambda,u)=N^{-1}\sum_{i=1}^N( 1\{ Y_i\leq D_iu+X_i'\lambda \} -\tau)X_i$, and $\hat\psi_N^*$ is defined similarly for the bootstrap sample.
Let $E_{P^*}$ denote the conditional expectation with respect to $P^*$, the law of $\{W^*_i\}_{i=1}^N$ conditional on the sample path. Let $BL_1$ denote the space of bounded Lipschitz functions on $\mathbb R^{d_1}$ with Lipschitz constant 1. Then, for any $\epsilon>0$,
\begin{multline}
\sup_{h\in BL_1}\Big|E_{P^*}h\big(\sqrt N\big[\phi(\hat \psi^*_N,r^*_N)-\phi(\hat \psi_N,r_N)\big]\big)-E_{P^*}h\big(\phi'_{\Psi,0}\big(\sqrt N\big[(\hat \psi^*_N,r^*_N)-(\hat \psi_N,r_N)\big]\big)\big)\Big|\\
\le \epsilon +2P^*\Big(\big\|\sqrt N\big[\phi(\hat \psi^*_N,r^*_N)-\phi(\hat \psi_N,r_N)\big]-\phi'_{\Psi,0}\big(\sqrt N\big[(\hat \psi^*_N,r^*_N)-(\hat \psi_N,r_N)\big]\big)\big\|>\epsilon\Big).\label{eq:boot_cons_L1}
\end{multline}
By Lemma \ref{lem:f_donsker} and Theorem 3.6.2 in \cite{VanderVaartWellner1996}, $\sqrt N(\hat \psi^*_N-\hat \psi_N)\stackrel{L^*}{\leadsto} \mathbb Gf_1$. Noting that $h\circ\phi'_{\Psi,0}\in BL_1(\ell^\infty(\Theta_{-1})^{d_1}\times\mathbb R)$ and $r_N=o_p(N^{-1/2})$, it follows that
\begin{align}
\sup_{h\in BL_1}\Big|E_{P^*}h\big(\phi'_{\Psi,0}\big(\sqrt N\big[(\hat \psi^*_N,r^*_N)-(\hat \psi_N,r_N)\big]\big)\big)-E_{P^*}h\circ \phi'_{\Psi,0}(\mathbb Gf_1,0)\Big|\to 0,\label{eq:boot_cons_L2}
\end{align}
with probability approaching 1 due to $r_N=o_P(N^{-1/2})$. Hence, for the conclusion of the theorem, it suffices to show that the second term on the right hand side of \eqref{eq:boot_cons_L1} tends to 0.

As shown in the proof of (i), $\phi$ is Hadamard differentiable at $(\psi,0)$. Hence, by Theorem 3.9.4 in \cite{VanderVaartWellner1996},
\begin{align*}
\sqrt N\big[\phi(\hat \psi^*_N,r^*_N)-\phi(\psi,0)\big]&=\phi'_{\Psi,0}(\sqrt N[(\hat \psi^*_N,r^*_N)-(\psi,0)])+o_{P^*}(1) \\
\sqrt N\big[\phi(\hat \psi_N,r_N)-\phi(\psi,0)]&=\phi'_{\Psi,0}(\sqrt N[(\hat \psi_N,r_N)-(\psi,0)])+o_{P}(1),
\end{align*}
Take the difference of the left and right hand sides respectively and note that $\phi'_{\Psi,0}$ is linear. This implies the right hand side of \eqref{eq:boot_cons_L1} tends to 0 in probability.
 This, together with \eqref{eq:boot_cons_L1}-\eqref{eq:boot_cons_L2}, ensures
\begin{align}
 	\mathcal L^*_{N,1}\stackrel{L^*}{\leadsto}\mathbb W_1,
\end{align}
where $\mathbb W_1(\theta_{-1})=-\frac{\partial^2}{\partial\theta_1\partial\theta_1'}Q_{P,1}(L_1(\theta_{-1}),\theta_{-1})^{-1}\mathbb Gf_j(\cdot;L_1(\theta_{-1}),\theta_{-1}).$
The analysis for any $j\ne 1$ is similar, and one may apply the arguments above jointly across $j=1,\dots, J$, which yields the second claim of the lemma.

\bigskip
\noindent
(iii) Consider the first submatrix of $E_P[(g(W;\theta)-E_P[g(W;\theta)])(g(w;\tilde\theta)-E_P[g(w;\tilde\theta)])']$.
It is given by
\begin{align}
\text{Var}\Big(&-\Big(\frac{\partial^2}{\partial\theta_1\partial\theta_1'}Q_{P,1}(\br_1(\theta_{-1}),\theta_{-1})\Big)^{-1}f_1(w;\br_1(\theta_{-1}),\theta_{-1})\Big)\notag\\
&\qquad -	\text{Var}\Big(-\Big(\frac{\partial^2}{\partial\theta_1\partial\theta_1'}Q_{P,1}(\br_1(\tilde \theta_{-1}),\tilde \theta_{-1})\Big)^{-1}f_1(w;\br_1(\tilde\theta_{-1}),\tilde\theta_{-1})\Big)\notag\\
&=\left(\frac{\partial^2}{\partial\theta_1\partial\theta_1'}Q_{P,1}(\br_1(\theta_{-1}),\theta_{-1})\right)^{-1}\text{Var}(f_1(w;\br_1(\theta_{-1}),\theta_{-1}))\left(\frac{\partial^2}{\partial\theta_1\partial\theta_1'}Q_{P,1}(\br_1(\theta_{-1}),\theta_{-1})\right)^{-1}\notag\\
&\qquad-\left(\frac{\partial^2}{\partial\theta_1\partial\theta_1'}Q_{P,1}(\br_1(\tilde \theta_{-1}),\tilde \theta_{-1})\right)^{-1}\text{Var}(f_1(w;\br_1(\tilde\theta_{-1}),\tilde\theta_{-1}))\left(\frac{\partial^2}{\partial\theta_1\partial\theta_1'}Q_{P,1}(\br_1(\tilde \theta_{-1}),\tilde \theta_{-1})\right)^{-1}.
\end{align}
Note that $\Theta$ is compact and $\theta_{-1}\mapsto \big(\frac{\partial^2}{\partial\theta_1\partial\theta_1'}Q_{P,1}(\br_1(\theta_{-1}),\theta_{-1})\big)^{-1}$ is continuous by Lemma \ref{lem:global_decentralization}, which implies that this map is uniformly continuous. Therefore, it remains to show the uniform continuity of $\theta\mapsto \text{Var}(f_1(w;\theta))$. Note that
\begin{multline}
	\text{Var}(f_1(w;\br_1(\theta_{-1}),\theta_{-1}))=E_P[(1\{ Y\leq D'\theta_{-1}+X'L_1(\theta_{-1}) \} -\tau)XX']\\
	\qquad-E_P[(1\{ Y\leq D'\theta_{-1}+X'L_1(\theta_{-1}) \}
	-\tau)X]E_P[(1\{ Y\leq D'\theta_{-1}+X'L_1(\theta_{-1}) \} -\tau)X]'.
\end{multline}
The right hand side of the display above is continuous on the compact domain $\Theta$, and hence it is uniformly continuous. One can argue the same way for the other subcomponents of $\text{diag}\big(E_P[(g(W;\theta)-E_P[g(W;\theta)])(g(w;\tilde\theta)-E_P[g(w;\tilde\theta)])']\big)$. This completes the proof.
\end{proof}

\begin{lemma}\label{cor:LMdist}
Suppose that Assumptions \ref{ass:ivqr} and \ref{ass:global_decentralization} hold. (i) Let $W_i=(Y_i,D_i',X_i',Z_i')',i=1,\dots,N$ be an i.i.d. sample generated from the IVQR model.  Then,
	\begin{align}
	\sqrt N \big(
\hat \K-\K
\big)
	\leadsto \mathbb W.\label{eq:LMdist0a}
	\end{align}
(ii) Let $W_i^*=(Y_i^*,D_i^{*\prime},X^{*\prime}_i,Z^{*\prime}_i)',i=1,\dots,N$ be an bootstrap sample from the empirical distribution $\hat P_N$ of $\{W_i\}_{i=1}^N$. Then,
\begin{align*}
\sqrt N\big(\hat \K^*-\hat\K \big) \stackrel{L^*}{\leadsto}\mathbb W.
\end{align*}
\end{lemma}

\begin{proof}
(i) By Lemma \ref{lem:sequi}, it follows that
\begin{align*}
\sqrt N\big(\hat \br_1(\cdot)-\br_1(\cdot),\dots,\hat \br_{\dpl}(\cdot)-\br_{\dpl}(\cdot)\big)'\leadsto\mathbb W.
\end{align*}
 Note that,  by the definition of $\hat\br$ and $\br$, one has
\begin{align*}
\sqrt N(\hat K_j(\theta)-K_j(\theta))&=\sqrt N(\hat L_j(\theta_{-j})-L_j(\theta_{-j})),j=1,\cdots,J.
\end{align*}
The conclusion of the lemma then follows. The proof of (ii) is similar, and is therefore omitted.
\end{proof}

\section{Consistency of the Contraction Estimator}\label{sec:consistency_contraction}
Below, we adopt the framework of \cite{Dominitz:2005aa}.
Let $(\mathcal X,d)$ be a metric space. For a contraction map $F:\mathcal X\to\mathcal X$, let $c_F$ be the modulus of contraction such that
\begin{align*}
d(F(x),F(x'))\le c_Fd(x,x'),
\end{align*}
for any $x,x'\in\mathcal X.$ 
As discussed in Section \ref{sec:sample_algorithms} the fixed point estimator $\est$ can be computed using the sample sequential dynamical system (in \eqref{eq:sample_sequential_ds}) or the following sample simultaneous dynamical system.
\begin{align}
\theta^{(s+1)}=\hat K\left(\theta^{(s)}\right),~~s=0,1,2,\dots,~~\param^{(0)}~\text{given}.\label{eq:sample_simultaneous_ds}
\end{align}

\begin{lemma}
Suppose Assumptions \ref{ass:ivqr}, \ref{ass:global_decentralization}, and \ref{ass:contraction_global} hold. Let $\est$ be an estimator constructed by iterating the dynamical system in \eqref{eq:sample_simultaneous_ds} or (in \eqref{eq:sample_sequential_ds}) $s_N$ times, where $s_N\ge -\frac{1}{2}\ln N/\ln c_K$.	Then,
\begin{align*}
\est-\true=O_p(N^{-1/2}).
\end{align*}
\end{lemma}

\begin{proof}
We show the result by applying Theorem 1 in \cite{Dominitz:2005aa} to the estimator obtained from the simultaneous dynamical system. The argument for the sequential system is similar.

By Assumption \ref{ass:contraction_global}, $K$ is a contraction map on $D_K$. Let $\theta^{(s)}$ be obtained from iterating $s$-times the population dynamical system in \eqref{eq:simultaneous_ds}.
The iteration on the dynamical system is covergent at least linearly \citep[][Proposition 1.1]{Bertsekas1989}. Under the condition on $s_N$, arguing as in \cite[p.842]{Dominitz:2005aa},  it follows that $N^{1/2}\|\theta^{(s_N)}-\true\|\le \|\theta^{(0)}-\true\|$. Finally, by Lemma \ref{cor:LMdist} and  tightness of $\mathbb W$, $N^{1/2}\sup_{\theta\in D_K}\|\hat K(\theta)-K(\theta)\|=O_p(1)$. These imply the conditions of Theorem 1 in \cite{Dominitz:2005aa} with $\delta=1/2$. The claim of the lemma then follows.
\end{proof}

\bibliographystyle{econometrica}
\bibliography{references}

\end{document}